\documentclass[runningheads]{llncs}
\usepackage[T1]{fontenc}
\usepackage{graphicx}
\usepackage{subcaption}

\usepackage{hyperref}
\usepackage{color}

\usepackage{standalone}

\usepackage{amsfonts}
\usepackage{thmtools}
\usepackage{mathtools}
\usepackage{cleveref}
\usepackage{arydshln}

\usepackage{cmll}
\usepackage{virginialake}
\odbackgroundtrue
\odframefalse
\odbackgroundfirstfalse
\odframefirstfalse
\vlnostructuressyntax

\usepackage{proofzilla}
\usepackage{xspace}
\usepackage{adjustbox}
\usepackage{twoopt}
\usepackage{pifont}

\usepackage[bbgreekl]{mathbbol} % for \bbmu and \bbnu

\def\set#1{\{#1\}}
\def\Set#1{\left\{#1\right\}}

\def\tuple#1{\langle#1\rangle}
\usepackage{scalerel}
\def\Tuple#1{{\left\langle #1 \right\rangle}}
\def\intset#1#2{\set{#1,\ldots,#2}}
\def\sizeof#1{|#1|}

\def\domain{\mathsf{dom}}

\renewcommand{\emptyset}{\varnothing}

\def\picalc{$\pi$-calculus\xspace}

\def\picalc{$\pi$-calculus\xspace}

\def\defn#1{{\itshape\bfseries\boldmath #1}}

 \def\rclr#1{{\color{red}#1}}
\def\bclr#1{{\color{blue}#1}}
\def\sclr#1{{\color{magenta}#1}}

\def\resp#1{(resp.~{#1})\xspace}

\newcommand{\quand}{\quad\mbox{and}\quad}

\usepackage{marginnote}

\spnewtheorem{nota}{Notation}[section]{\bfseries}{\rmfamily}

\def\conet{conflict net\xspace}
\def\conets{conflict nets\xspace}
\def\slnet{slice net\xspace}
\def\slnets{slice nets\xspace}

\newcommand{\varSet}{\mathcal V}
\def\varset{\varSet}

\def\cneg#1{{#1}^\lbot}

\def\free{\mathsf{free}}
\def\freeof#1{\free(#1)}

\def\lsend#1#2{\langle#1\oc#2\rangle}
\def\lrecv#1#2{(#1\wn#2)}

\def\lNu#1#2{\lnewsymb{#1}.#2}

\def\lYa#1#2{\lyasymb{#1}.#2}
\def\lFa#1#2{\forall{#1}.#2}
\def\lEx#1#2{\exists{#1}.#2}
\newcommand{\lQu}[3][]{\lqusymb_{#1}{#2}.#3}

\def\lNa#1#2{\nabla#1.#2}
\def\lnNa#1#2{\cneg\nabla#1.#2}
\def\lNai#1#2#3{\nabla_{#1}#2.#3}

\newcommand{\cnof}[1]{\left\{\!\!\left\{#1\right\}\!\!\right\}_{\mathsf{co}}\!}
\newcommand{\snof}[1]{\left\{\!\!\left\{#1\right\}\!\!\right\}_{\mathsf{sl}}}

\def\fsubst#1#2{[#1/#2 ]}

\def\fsubsts#1{[ #1 ]}
\def\fsubminus#1{\setminus\set{#1}}

\newcommand{\proves}[1][]{\mathord{\vdash_{#1}\,}}
\def\dD{\mathcal D}

\def\BV{\mathsf{BV}}

\def\MALL{\mathsf{MALL}}

\mathchardef\mhyphen="2D

\def\wF{^{1}}

\def\PIL{\mathsf{PiL}}

\newcommand{\sS}[1][\mathcal S]{\sclr{ #1 }}
\newcommand{\emptystore}{\sclr{ \emptyset }}
\newcommand{\sdash}[1][\sS]{\sS[{#1}] \vdash}

\def\axrule{\mathsf {ax}}
\def\leafr{\mathsf{leaf}}

\def\mixr{\mathsf{mix}}

\newcommand{\rrule}[1][]{{\mathsf{r}^{#1}}}
\newcommand{\brrule}[1][]{\mathsf{\beta}_{#1}}

\newcommand{\urrule}[1][]{\mathsf{\alpha}_{#1}}

\def\nurule{\lnewsymb}

\def\nuurule{\lnewsymb^\lunit}
\def\yurule{\lyasymb^\lunit}

\def\loadr#1{#1_\mathsf{load}}
\def\popr#1{#1_\mathsf{pop}}
\def\unitr#1{#1_{\lunit}}

\def\naur{\unitr\nabla}
\def\naloadr{\loadr\nabla}
\def\napopr{\popr\nabla}

\def\nnapopr{\popr{\cneg\nabla}}

\def\nqpopr{\napopr}

\def\nuur{\unitr\lnewsymb}
\def\nuloadr{\loadr\lnewsymb}
\def\nupopr{\popr\lnewsymb}
\def\yaur{\unitr\lyasymb}
\def\yaloadr{\loadr\lyasymb}
\def\yapopr{\popr\lyasymb}

\def\isnusymb{\lnewsymb}
\def\isyasymb{\lyasymb}
\def\isnasymb{\nabla}
\def\isdnasymb{\cneg\nabla}
\def\isnu#1{\sclr{#1^{\isnusymb}}}
\def\isya#1{\sclr{#1^{\isyasymb}}}

\def\isqu#1{\sclr{#1^{\lqusymb}}}
\newcommand{\isna}[2][]{\sclr{#2^{\isnasymb_{#1}}}}
\newcommand{\isnna}[2][]{\sclr{#2^{\isdnasymb_{#1}}}}
\def\conc{\frown}
\def\conf{\#}
\def\pzlink#1#2#3#4#5#6{
	\tikz[overlay,remember picture,draw,fill,opacity=1]{
		\foreach \ppp in #6{
			\draw[thick,color=#5] (\ppp) -- ++(0,#3pt) -| (#1);
		}
		\draw[thick,color=#5] (#1) -- ++(0,#3pt) -| (#2) node[pos=.25,fill=white,inner sep=1pt]{\scriptsize$#4$};
	}
}
\def\pzlinks#1{
	\foreach \aaa/\bbb/\ccc/\ddd/\eee/\fff in {#1}{
		\pzlink{\aaa}{\bbb}{\ccc}{\ddd}{\eee}{\fff}
	}
}

\defedgetype{flow}{draw,thick,color = magenta}{}

\pzvertex{conf}{\conf}{}
\pzvertex{conc}{\conc}{}
\pzvertex{Nu}{\lnewsymb}{}
\pzvertex{Na}{\nabla}{}
\pzvertex{nNa}{\cneg\nabla}{}
\pzvertex{Ya}{\lyasymb}{}
\pzvertex{Qu}{\lqusymb}{}

\newcommand{\link}[2][black]{{\color{#1}{#2}}}

\newcommand{\la}[1][red]{{\link[#1]a}}
\newcommand{\lb}[1][blue]{{\link[#1]b}}
\newcommand{\laprimo}[1][red]{{\link[#1]{a'}}}
\newcommand{\laprimouno}[1][red]{{\link[#1]{a'_1}}}
\newcommand{\laprimodue}[1][red]{{\link[#1]{a'_2}}}
\newcommand{\lbprimo}[1][blue]{{\link[#1]{b'}}}

\newcommand{\lc}[1][violet]{{\link[#1]c}}
\newcommand{\lcdue}[1][violet]{{\link[#1]{c'}}}
\newcommand{\lctre}[1][violet]{{\link[#1]{c''}}}
\newcommand{\lcnu}[1][skewcolor]{{\link[#1]{cn}}}
\newcommand{\lcya}[1][skewcolor]{{\link[#1]{cy}}}
\newcommand{\lnom}[1][skewcolor]{{\link[#1]n}}
\newcommand{\lcd}[1][black]{\link[#1]{%\epsilon(c)}}
c}}
\newcommand{\labb}[1][black]{\link[#1]{%\delta \duasum \sigma (
		ab
		}}
\newcommand{\laa}[1][black]{\link[#1]{%\delta(a')}}
a'}}
\newcommand{\laaa}[1][black]{\link[#1]{%\delta - x(a')}}
a'}}
\newcommand{\lbb}[1][black]{\link[#1]{%\sigma(b')}}
b'}}
\newcommand{\lbbb}[1][black]{\link[#1]{%\sigma-x(b')}}
b'}}

\newcommand{\lccd}[1][black]{\link[#1]{%\epsilon(c')}}
c'}}
\newcommand{\lcccd}[1][black]{\link[#1]{%\epsilon-x(c')}}
c'}}
\newcommand{\laabb}[1][black]{\link[#1]{%\delta \duasum \sigma(a'b)}}
a'b}}
\newcommand{\laaabb}[1][black]{\link[#1]{%\alpha(a'b)}}
a'b}}
\newcommand{\lac}[1][black]{\link[#1]{%\delta \duasum \gamma(ac)}}
ac}}
\newcommand{\labc}[1][black]{\link[#1]{%\beta(abc)}}
abc}}

\newcommand{\coalto}[1][]{\xrightarrow{#1}}
\newcommand{\coaltos}[1][]{\xrightarrow{#1}^*}
\newcommand{\lcoalto}[2][]{\overset{#2}{\rightarrow}}

\def\cotree{coco-tree\xspace}
\def\cotrees{coco-trees\xspace}
\def\linking{\Lambda}
\def\linktree{\tau(\Lambda)}

\def\treeof#1{\tau(#1)}
\def\witnessof#1{\delta^{#1}}
\def\ltree{\tau}

\def\domof#1{\mathsf{dom}(#1)}

\newcommand{\dualizerof}[1][]{\delta_{#1}}

\def\peq{\approx}
\def\speq{\simeq}

\def\pweq{\peq_{\mathsf w}}

\def\join{\vee}

\def\duasum{+}

\def\newcohe#1#2{\mathsf{coh}(#1,#2)}

\def\pn{\mathcal P}

\def\alphaeq{\equiv_{\alpha}}

\defedgetype{dG}{draw,thick,densely dotted}{}

\def\canon#1{\left\lfloor #1 \right\rfloor}

\begin{document}
\title{Proof Nets for PiL (Full Version)}
\author{
	Matteo Acclavio\inst{1,2} \orcidID{0000-0002-0425-2825}
	\and
	\\
	Giulia Manara\inst{1}\thanks{Funded by the European Union (ERC, CHORDS, 101124225)} 		\orcidID{0009-0003-9583-1017}
}
% %
\authorrunning{M. Acclavio and G. Manara}
\institute{
	FORM, University of Southern Denmark, Denmark 
\and
	University of Sussex, UK
% \email{lncs@springer.com}
% \\
}
\maketitle             
 % typeset the header of the contribution
%
%%%%%%%%%%%%%%%%%%%%%%%%%%%%%%%%%%%%%%%%%%%%%%%%%%%%%%%%%%%%%%%%%%%
%%%%%%%%%%%%%%%%%%%%%%%%%%%%%%%%%%%%%%%%%%%%%%%%%%%%%%%%%%%%%%%%%%%
\begin{abstract}
	We introduce proof nets for PiL, an extension of first-order multiplicative additive linear logic with new operators allowing a shallow encoding of processes in the $\pi$-calculus as formulas.
	
	We provide correctness criterion, sequentialization procedure, and a proof translation algorithm.
	We show that proof nets provide a canonical representation of sequent calculus derivations modulo rule permutations.
\end{abstract}
%%%%%%%%%%%%%%%%%%%%%%%%%%%%%%%%%%%%%%%%%%%%%%%%%%%%%%%%%%%%%%%%%%%
%%%%%%%%%%%%%%%%%%%%%%%%%%%%%%%%%%%%%%%%%%%%%%%%%%%%%%%%%%%%%%%%%%%

%%%%%%%%%%%%%%%%%%%%%%%%%%%%%%%%%%%%%%%%%%%%%%%%%%%%%%%%%%%%%%%%%%%
%%%%%%%%%%%%%%%%%%%%%%%%%%%%%%%%%%%%%%%%%%%%%%%%%%%%%%%%%%%%%%%%%%%
%%%%%%%%%%%%%%%%%%%%%%%%%%%%%%%%%%%%%%%%%%%%%%%%%%%%%%%%%%%%%%%%%%%
\section{Introduction}\label{sec:intro}
%%%%%%%%%%%%%%%%%%%%%%%%%%%%%%%%%%%%%%%%%%%%%%%%%%%%%%%%%%%%%%%%%%%
%%%%%%%%%%%%%%%%%%%%%%%%%%%%%%%%%%%%%%%%%%%%%%%%%%%%%%%%%%%%%%%%%%%
%%%%%%%%%%%%%%%%%%%%%%%%%%%%%%%%%%%%%%%%%%%%%%%%%%%%%%%%%%%%%%%%%%%

The system $\PIL$ is an extension of (a particular case of)
\emph{first-order multiplicative additive linear logic} ($\MALL\wF$) with two new operators:
a non-commutative non-associative self-dual connective ($\lprec$), and a nominal quantifier\footnote{
	In the sense of Pitt's and Gabbay's nominal logic \cite{pitts:nominal,gabbay:pitts:nominal}, which is a logic for reasoning about names and binding, where the nominal quantifier $\lnewsymb$ allows to express freshness conditions on variables.
}
($\lnewsymb$) together with its dual ($\lyasymb$).
It has been introduced in \cite{acc:man:mon:FaP} to prove the completeness of choreographic programming \cite{montesi:book} with respect to (race-free) deadlock-free processes of the \emph{$\pi$-calculus}\cite{M80}.
Thanks to the $\lprec$ modeling sequential composition, and the nominal quantifier $\lnewsymb$ modelling name binding, $\PIL$ provides a logical framework allowing for a shallow encoding of processes as formulas, and for a proof-theoretic account of execution of processes as proof-search.
In particular, the completeness result relies on the correspondence between cut-free derivations in $\PIL$ and choreographies. 
This correspondence relies on the fact that proof equivalence in $\PIL$ induced by independent rule permutations captures the semantics of execution trees of processes modulo the interleaving of independent actions, and the natural equivalence between choreographies modulo the mechanism of out-of-order execution.
For this reason, the study of proof equivalence in $\PIL$ is a way to study the semantics of processes in the $\pi$-calculus, and of choreographies.
In this paper, we provide formalisms for proofs in $\PIL$ providing canonical representation of sequent calculus derivations modulo different classes of independent rule permutations to provide the basis for developing new tools to study the semantics of process calculi.

We build on the rich literature on \emph{proof nets} for linear logic (see, e.g., \cite{gir:ll,dan:reg:89,girard:96:PN,ret:handsome,bellin:subnets,hug:unification,hug:van:slice,hei:hug:conflict,hei:hug:str:ALL1}), a graph-based formalism for proofs that has been widely studied as a solution to the problem of proof identity in proof theory.
Proof nets achieve canonicity by replacing sequent rules with graphical constructors that minimize dependencies between inference steps, thereby collapsing entire classes of sequent derivations into a single syntactic object. 
However, the syntax of proof nets is strictly more expressive than that of the sequent calculus. It admits the construction of terms, called \emph{proof structures}, that do not correspond to valid derivations. 
For this reason, every proof net formalism must be equipped with a (polynomial-time) \emph{correctness criterion} that identifies the proof structures representing valid proofs, together with a polynomial-time \emph{sequentialization procedure} translating proof nets back into sequent derivations.

Proof nets for the multiplicative fragment of linear logic can be seen as a decoration of the parse tree of the conclusion of a proof with \emph{axiom links} connecting leaves which are dual atoms.
\begin{equation}
	\begin{array}{ccc}
		\vlderivation{
			\vliin{\ltens}{}{
				\vdash (A \lpar B)\ltens C, \cneg C , (\cneg A \ltens \cneg B)
			}{
				\vliin{\ltens}{}{
					\vdash A \lpar B, \cneg A \ltens \cneg B
				}{
					\vlin{\axrule}{}{\vdash \vA1, \vnA1 \pzlinks{A1/nA1/12/\la/red/}}{\vlhy{}}
				}{
					\vlin{\axrule}{}{\vdash \vB1, \vnB1 \pzlinks{B1/nB1/12/\lc/violet/}}{\vlhy{}}
				}
			}{
				\vlin{\axrule}{}{\vdash \vC1, \vnC1 \pzlinks{C1/nC1/12/\lc/violet/}}{\vlhy{}}
			}
		}
	&\qquad &
		\begin{array}{c}
			\vdash (\vA1 \lpar \vB1)\ltens \vC1, \vnC1 , (\vnA1 \ltens \vnB1)
			\pzlinks{A1/nA1/23/\la/red/}
			\pzlinks{B1/nB1/18/\lb/blue/}
			\pzlinks{C1/nC1/12/\lc/violet/}
		\end{array}
	\end{array}
\end{equation}
Correctness can be characterized either by checking acyclicity and connectivity conditions on associated graphs \cite{dan:reg:89,retore:phd}, or by verifying that the structure contracts to a single vertex under a system of rewriting rules \cite{danos:phd}.

Additives complicate this picture because the inference rule for the additive conjunction $\lwith$ requires context duplication.
This requires proof nets to track the `additive branching' structure of the proof which is of a different nature with respect to the `multiplicative branching' structure induced by the $\ltens$ rule.
\begin{equation}\label{eq:introMALL}
	\adjustbox{max width=.92\textwidth}{$
	\begin{array}{ccc}
		\vlderivation{
			\vliin{\ltens}{}{
				\vdash (A \lwith A), (\cneg A \otimes \cneg B), B
			}{
				\vliin{\lwith}{}{
					\vdash (A \lwith A), \cneg A 
				}{
					\vlin{\axrule}{}{\vdash \vA1, \vnA1 \pzlinks{A1/nA1/12/\la/red/}}{\vlhy{}}
				}{
					\vlin{\axrule}{}{\vdash \vA2, \vnA2 \pzlinks{A2/nA2/12/\lb/blue/}}{\vlhy{}}
				}
			}{
				\vlin{\axrule}{}{\vdash \vB1, \vnB1 \pzlinks{B1/nB1/12/\lc/violet/}}{\vlhy{}}
			}
		}
	%%%%
	&\quad\overset ?\sim\quad&
	%%%%
		\vlderivation{
			\vliin{\lwith}{}{
				\vdash (A \lwith A), (\cneg A \otimes \cneg B), B
			}{
				\vliin{\ltens}{}{
					\vdash A , (\cneg A \otimes \cneg B), B
				}{
					\vlin{\axrule}{}{\vdash \vA1, \vnA1 \pzlinks{A1/nA1/12/\la/red/}}{\vlhy{}}
				}{
					\vlin{\axrule}{}{\vdash \vB1, \vnB1 \pzlinks{B1/nB1/12/\lcdue/violet/}}{\vlhy{}}
				}
			}{
				\vliin{\ltens}{}{
					\vdash A , (\cneg A \otimes \cneg B), B
				}{
					\vlin{\axrule}{}{\vdash \vA2, \vnA2 \pzlinks{A2/nA2/12/\lb/blue/}}{\vlhy{}}
				}{
					\vlin{\axrule}{}{\vdash \vB2, \vnB2 \pzlinks{B2/nB2/12/\lctre/violet/}}{\vlhy{}}
				}
			}
		}
	\end{array}
	$}
\end{equation}

For this reason, correctness criteria for proof nets for $\MALL$ require a richer structure.
In Girard's \emph{monomial nets} \cite{gir:meanII}, this is achieved by annotating axiom links with Boolean variables.
In Hughes and Heijltjes' \emph{conflict nets} \cite{hughes:conflict,hei:hug:conflict}, axiom links are organized into trees, where they are related either by a `multiplicative' {co}ncordance relation or by an `additive' {co}nflict relation.
%%%%%%%%%%%%%%%%%%%%%%%%%%%%%%%%%%%%%%%%%%%%%%%%%%%%%%%%%%%%%%%%%%%
\begin{equation}\label{eq:introMALLpart2}
	\begin{array}{c@{\qquad}c@{\qquad}c}
	\\[-10pt]
		\begin{array}{c@{\quad}c@{\quad}c@{\quad}c}
			\vpz5{\la}
			&&
			\vpz7{\lb}
		\\[-5pt]
			&{\vpz2\conf}&
			&\vpz1{\lc}
		\\[-5pt]
		&&\vpz0\conc
		\Gedges{
			pz5/pz2,pz7/pz2,
			pz1/pz0, pz0/pz2%
		}
		\end{array}
	%%%%
	&\neq&
	%%%%
		\begin{array}{ccc}
			\vpz1{\la}
			\quad\qquad
			\vpz5{\lcdue}
			&&
			\vpz2{\lb}
			\quad\qquad
			\vpz7{\lctre}
			\\[-5pt]
			\vpz3{\conc}
			&&
			\vpz4{\conc}
			\\[-5pt]
			&\vpz0{\conf}
		\end{array}
		\Gedges{
			pz0/pz3,pz0/pz4,
			pz3/pz5,pz3/pz1,
			pz4/pz7,pz4/pz2%
		}
	\end{array}
\end{equation}
%%%%%%%%%%%%%%%%%%%%%%%%%%%%%%%%%%%%%%%%%%%%%%%%%%%%%%%%%%%%%%%%%%%

Quantifiers introduce additional complexity in the definition of proof nets, as they require tracking the witnesses of quantification and the potential dependencies between them. 
For example, the following two derivations may be considered non-equivalent, since the witnesses of the existential quantifiers are different.

%%%%%%%%%%%%%%%%%%%%%%%%%%%%%%%%%%%%%%%%%%%%%%%%%%%%%%%%%%%%%%%%%%%
\begin{equation}\label{eq:intro2}
	\begin{array}{c@{\qquad}c@{\qquad}c}
		\vlupsmash{\vlderivation{
			\vliq{2\times\exists}{}{
				\vdash \lEx x P(x), \lEx y \cneg{P(y)}
			}{
				\vlin{\axrule}{}{\vdash \vpz1{P(z)}, \cneg{\vpz2{P(z)}}
				}{\vlhy{}}
			}
		}}
		&\overset{?}{\sim}&
		\vlupsmash{\vlderivation{
			\vliq{2\times\exists}{}{
				\vdash \lEx x P(x), \lEx y \cneg{P(y)}
			}{
				\vlin{\axrule}{}{\vdash \vpz3{P(w)}, \cneg{\vpz4{P(w)}}
				}{\vlhy{}}
			}
		}}
	\end{array}
\end{equation}
%%%%%%%%%%%%%%%%%%%%%%%%%%%%%%%%%%%%%%%%%%%%%%%%%%%%%%%%%%%%%%%%%%%

The information required to distinguish the two derivations in \Cref{eq:intro2} is captured by associating each link with 
a substitution, called a \emph{dualizer}, mapping the variables occurring in the link to the corresponding existential witnesses.
The name \emph{dualizer} reflects the fact that a link may connect atoms that are not yet dual, and that it is the dualizer that 
enforces duality by unifying them.
The derivations in \Cref{eq:intro2} yield the same links but different dualizers:
In the one on the left \resp{right}, the witnesses of the existential quantifiers is the variable $z$ \resp{$w$}, and the duality between $P(x)$ and $\cneg{P(y)}$ is enforced by the dualizer $\dualizerof[\la]=\fsubsts{z/x,z/y}$ \resp{$\dualizerof[\lb]=\fsubsts{w/x,w/y}$}.
%%%%%%%%%%%%%%%%%%%%%%%%%%%%%%%%%%%%%%%%%%%%%%%%%%%%%%%%%%%%%%%%%%%
\begin{equation}\label{eq:intro3}
	\begin{array}{c@{\qquad}c}
		\vdash \lEx x \vpz1{P(x)}, \lEx y \cneg{\vpz2{P(y)}}
		\pzlinks{pz1/pz2/12/\la/red/}
	&
		\vdash \lEx x \vpz3{P(x)}, \lEx y \cneg{\vpz4{P(y)}}
		\pzlinks{pz3/pz4/12/\lb/blue/}
	\end{array}
\end{equation} 
%%%%%%%%%%%%%%%%%%%%%%%%%%%%%%%%%%%%%%%%%%%%%%%%%%%%%%%%%%%%%%%%%%%

%%%%%%%%%%%%%%%%%%%%%%%%%%%%%%%%%%%%%%%%%%%%%%%%%%%%%%%%%%%%%%%%%%%
%%%%%%%%%%%%%%%%%%%%%%%%%%%%%%%%%%%%%%%%%%%%%%%%%%%%%%%%%%%%%%%%%%%
\subsection{Challenges}\label{subsec:challenges}
%%%%%%%%%%%%%%%%%%%%%%%%%%%%%%%%%%%%%%%%%%%%%%%%%%%%%%%%%%%%%%%%%%%
%%%%%%%%%%%%%%%%%%%%%%%%%%%%%%%%%%%%%%%%%%%%%%%%%%%%%%%%%%%%%%%%%%%

We adopt \emph{conflict nets} as the main underlying framework for proof nets for $\PIL$, as they provide a good trade-off between canonicity, computational complexity of correctness criterion, sequentialization procedure, and proof translation algorithm%
\footnote{
    We remind the reader that cut-elimination for conflict nets is not efficient.
    However, in this paper we are interested in proof nets as canonical representatives of \emph{cut-free} proofs, since in \cite{acc:man:mon:FaP} they are used to represent both computation trees of processes and choreographies.
	Therefore, we are not interested in study the semantics of cut-elimination in this paper, which present additional challenges we shall address in future work.
}.
Note that the only proof nets for $\MALL\wF$ available in the literature are the ones based on Girard's proof nets, which present some issues in terms of canonicity and computational efficiency of the proof translation \cite{hug:van:slice,hei:hug:conflict}.
At the same time, works such as \cite{hughes:firstorder} and \cite{hei:hug:str:ALL1} provide techniques for handling quantifiers in the multiplicative fragment and additive fragment of linear logic based on contractability correctness criteria similar to the ones for conflict nets.
However, so far these techniques have never been studied in presence of both multiplicative and additive connectives.

Another challenge in defining proof nets for $\PIL$ is the presence of the non-commutative and non-associative connective $\lprec$, it has strong connections with the connective $\lseq$ of Retoré's \emph{pomset logic} \cite{retore:phd,ret:newPomset}, for which no contractability correctness criterion is known, and Guglielmi's $\BV$ logic, for which no correctness criterion for proof nets is known.
However, the connective $\lprec$ is weaker than the $\lseq$ since the latter prevents the existence of a cut-free sequent calculus \cite{tiu:SIS-II,tito:lutz:csl22,tito:str:SIS-III}. At the same time, $\lprec$ is fundamentally different from the non-commutative connectives studied in works such as \cite{yetter:90,abr:rue:noncomI}, which are associative and not self-dual.

Finally, our nominal quantifiers require a fine bookkeeping of the freshness conditions on variables to guarantees that nominal quantified variables can be shared by at most a unique pair of dual nominal quantifiers.
This is a novel feature (even with respect to other logics with nominal quantifiers), requires novel techniques.

%%%%%%%%%%%%%%%%%%%%%%%%%%%%%%%%%%%%%%%%%%%%%%%%%%%%%%%%%%%%%%%%%%%
%%%%%%%%%%%%%%%%%%%%%%%%%%%%%%%%%%%%%%%%%%%%%%%%%%%%%%%%%%%%%%%%%%%
\subsection{Contributions}
%%%%%%%%%%%%%%%%%%%%%%%%%%%%%%%%%%%%%%%%%%%%%%%%%%%%%%%%%%%%%%%%%%%
%%%%%%%%%%%%%%%%%%%%%%%%%%%%%%%%%%%%%%%%%%%%%%%%%%%%%%%%%%%%%%%%%%%
We define \textbf{\conets} for $\PIL$, and we provide a (polynomial-time) correctness criterion, a proof translation algorithm, and a sequentialization procedure.
Then, we provide a rewriting to reduce \conets to canonical form, called \textbf{\slnets} while preserving the correctness of the structure.
For both proof nets, we prove canonicity results, showing that they provide a canonical representation of sequent calculus derivations modulo distinct classes of independent rule permutations.

%%%%%%%%%%%%%%%%%%%%%%%%%%%%%%%%%%%%%%%%%%%%%%%%%%%%%%%%%%%%%%%%%%%
%%%%%%%%%%%%%%%%%%%%%%%%%%%%%%%%%%%%%%%%%%%%%%%%%%%%%%%%%%%%%%%%%%%
\subsection{Structure of the paper}
%%%%%%%%%%%%%%%%%%%%%%%%%%%%%%%%%%%%%%%%%%%%%%%%%%%%%%%%%%%%%%%%%%%
%%%%%%%%%%%%%%%%%%%%%%%%%%%%%%%%%%%%%%%%%%%%%%%%%%%%%%%%%%%%%%%%%%%
In \Cref{sec:back} we recall the definition and the main properties of the sequent system $\PIL$, we recall the notation for substitutions, and the definition of links and their dualizers, which will be used in the definition of proof nets.
In \Cref{sec:CN} we define \conets for $\PIL$, and we provide a correctness criterion, a proof translation algorithm, and a sequentialization procedure.
In \Cref{sec:SN} we define \slnets and we show how they can be obtained from \conets by a rewriting procedure called \emph{flattening}.
In \Cref{sec:canon} we show that \conets and \slnets provide a canonical representation of sequent calculus derivations modulo distinct classes of independent rule permutations. 
We conclude in \Cref{sec:conc} with a discussion of the results, and of future work.

Because of space constraints, some proofs are only sketched in the main text, and are fully detailed in \Cref{app:proofs}.

% \matteo{
% 	Because of space constraints, some proofs are only sketched in the paper, but they are fully detailed in the appendix of the extended version of the paper \cite{acc:man:IJCAR26full}.
% }

%%%%%%%%%%%%%%%%%%%%%%%%%%%%%%%%%%%%%%%%%%%%%%%%%%%%%%%%%%%%%%%%%%
%%%%%%%%%%%%%%%%%%%%%%%%%%%%%%%%%%%%%%%%%%%%%%%%%%%%%%%%%%%%%%%%%%%
%%%%%%%%%%%%%%%%%%%%%%%%%%%%%%%%%%%%%%%%%%%%%%%%%%%%%%%%%%%%%%%%%%%
%%%%%%%%%%%%%%%%%%%%%%%%%%%%%%%%%%%%%%%%%%%%%%%%%%%%%%%%%%%%%%%%%%%
\section{Background}\label{sec:back}
%%%%%%%%%%%%%%%%%%%%%%%%%%%%%%%%%%%%%%%%%%%%%%%%%%%%%%%%%%%%%%%%%%%
%%%%%%%%%%%%%%%%%%%%%%%%%%%%%%%%%%%%%%%%%%%%%%%%%%%%%%%%%%%%%%%%%%%
%%%%%%%%%%%%%%%%%%%%%%%%%%%%%%%%%%%%%%%%%%%%%%%%%%%%%%%%%%%%%%%%%%%

The sequent system $\PIL$ is an extension (of a special case) of \emph{first-order multiplicative additive linear logic} ($\MALL\wF$) obtained by including new operators allowing for a faithful embedding of processes of the \picalc as formulas \cite{acc:man:mon:FaP}: the connective $\lprec$ allows to represent the sequential composition of processes, while the nominal quantifier $\lnewsymb$ allows to represent name restriction and name passing.
In this section we recall the definition of $\PIL$ and some of its properties, which will be useful for the definition of proof nets in the next sections.

We consider \defn{formulas} generated from a countable set of variables $\varset$ and the dual binary predicate symbols 
\defn{send} ($\lsend --$) and 
\defn{receive} ($\lrecv --$), 
the \defn{unit} ($\lunit$),
the binary connectives 
\defn{par} ($\lpar$), 
\defn{tensor} ($\ltens$), 
\defn{prec} ($\lprec$), 
\defn{oplus} ($\lplus$), and 
\defn{with} ($\lwith$), 
the standard 
\defn{universal} ($\forall$) and 
\defn{existential} ($\exists$) quantifiers, 
and the nominal quantifiers 
\defn{new} ($\lnewsymb$) and 
\defn{ya} ($\lyasymb$) 
using the following grammar
\begin{equation}
	A,B \coloneqq
	a				\mid
	A\odot B		\mid
	\lQu xA		
	\quad\text{with }
	\begin{cases}
		a \in \set{\lunit, \lsend xy , \lrecv xy \mid x,y\in\varset}
		\text{ \defn{atomic}}
	\\
		\odot \in \set{\lpar,\ltens,\lprec,\lplus,\lwith}
	\\
		\lqusymb \in \set{\forall, \exists, \lnewsymb, \lyasymb}, \quad x\in\varset
	\end{cases}
\end{equation}
modulo the \defn{$\alpha$-equivalence} generated by the following equations:
\begin{equation}
	\begin{array}{r@{\;\alphaeq\;}l@{\quad}l}
		a 	 & 	a
	&
		\mbox{if $a \in \set{\lunit, \lsend xy , \lrecv xy \mid x,y\in\varset}$}
	\\
		A_1 \circleddot A_2 	& B_1 \circleddot B_2
	&
		\mbox{if $A_i \alphaeq B_i$ and $\circleddot \in \set{\lpar,\lprec,\ltens,\lplus,\lwith}$}
	\\
		\lQu x A 		&	\lQu y A \fsubst yx
	&
		\mbox{if $y$ fresh for $A$ and $\lqusymb \in \set{\lnewsymb, \lyasymb, \forall, \exists}$}
	\end{array}
\end{equation}

We say that a quantification $\lQu x A$ with $\lqusymb\in\set{\forall,\exists,\lnewsymb,\lyasymb}$ \defn{binds} the occurrences of the variable $x$ in $A$, and define the set of \defn{free variables} of a formula $A$ as the set $\freeof A$ of variables occurring in $A$ which are not bound by any quantifier.
To simplify our presentation, in this paper we identify formulas and their \defn{formula-trees}, that is, we assume each atom, connective, and quantifier in a formula to be uniquely identified by its corresponding node in the formula-tree of the formula.

A \defn{sequent} $\Gamma$ is a set of occurrences of formulas, or, equivalently, a forest of formula trees.
Its set of \defn{free variables} $\freeof \Gamma$ are the union of the free variables of the formulas occurring in $\Gamma$.

%%%%%%%%%%%%%%%%%%%%%%%%%%%%%%%%%%%%%%%%%%%%%%%%%%%%%%%%%%%%%%%%%%%
\begin{definition}\label{def:store}
	A \defn{nominal variable} is an element of the form $\isna x$ with $x\in\varSet$ and 
	$\nabla\in\set{\isnusymb,\isyasymb}$.
	If $\sS$ is a set of nominal variables, we say that $x$ \defn{occurs} in $\sS$ if $\isnu x$ or 
	$\isya x$ is an element of $\sS$.
	A  \defn{(nominal) store} $\sS$ is a set of nominal variables such that each variable occurs at most once in $\sS$.
	We denote by $\sS_1,\sS_2$ the union of two stores $\sS_1$ and $\sS_2$ sharing no variable.
	
	A \defn{judgement} (for a sequent $\Gamma$) is an expression of the form $\sdash\Gamma$ with $\sS$ a store such that variables occurring in $\sS$ are free in $\Gamma$.
	We may write $\sdash[ {\isna[1]{x_1},\ldots,\isna[n]{x_n} }]\Gamma$ for the judgement $\set{ \isna[1]{x_1},\ldots,\isna[n]{x_n}} \vdash \Gamma$ and 
	simply $\vdash \Gamma$ for the judgement $\emptystore \vdash \Gamma$.
\end{definition}
%%%%%%%%%%%%%%%%%%%%%%%%%%%%%%%%%%%%%%%%%%%%%%%%%%%%%%%%%%%%%%%%%%%

%%%%%%%%%%%%%%%%%%%%%%%%%%%%%%%%%%%%%%%%%%%%%%%%%%%%%%%%%%%%%%%%%%%
\begin{nota}
	From now on, we always assume judgements and sequents to be \defn{clean}, meaning that each variable $x$ may occur bound by at most one quantifier, and if it occurs bound or in the store, then cannot occur free in the sequent.
	Note however, that we may provide examples with non-clean judgements to improve readability and to highlight the connection between variables bound by dual quantifiers.
\end{nota}
%%%%%%%%%%%%%%%%%%%%%%%%%%%%%%%%%%%%%%%%%%%%%%%%%%%%%%%%%%%%%%%%%%%

%%%%%%%%%%%%%%%%%%%%%%%%%%%%%%%%%%%%%%%%%%%%%%%%%%%%%%%%%%%%%%%%%%%
The system $\PIL$ is defined by all rules in \Cref{fig:rules}.
The left-hand side of the figure shows the standard rules for \emph{first-order multiplicative additive linear logic} ($\MALL\wF$) decorated with stores.
On the right-hand side we recall the rules for the (self-dual) unit $\lunit$, and for the $\lprec$ operator, as well as the rules for nominal quantifiers.
The rule for the $\lprec$ is \emph{multiplicative} (in this sense of \cite{dan:reg:89,gir:meanII,acc:mai:20}) and can introduce any number of $\lprec$-formulas at once; the case $n=0$ corresponds to the $\mixr$ rule.
The rules $\nuur$ and $\yaur$ simply remove the nominal quantifier while respecting the freshness condition ($\dagger$), as the standard rule $\forall$ for the universal quantifier.
Similarly, the rule $\nuloadr$ \resp{$\yaloadr$} removes quantification while respecting the same freshness condition of universal quantification ($\dagger$), and additionally loads the fresh variable $x$, bound by the nominal quantifier, as a nominal variable $\isnu x$ \resp{$\isya x$} in the store.
The rule $\nupopr$ \resp{$\yapopr$} behaves similarly to the rule $\exists$ for the existential quantifier, but uses as a witness of the quantification a nominal variable $\isnu y$ \resp{$\isya y$} in the store, consuming it.
The combination of the rules $\nuloadr$ and $\nupopr$ \resp{$\yaloadr$ and $\yapopr$}, 
together with the linear control of resources of the conjunction rules,
guarantees that each nominal variable can be used as witness for at most one nominal quantification.

%%%%%%%%%%%%%%%%%%%%%%%%%%%%%%%%%%%%%%%%%%%%%%%%%%%%%%%%%%%%%%%%%%%
% fig Formulas
%%%%%%%%%%%%%%%%%%%%%%%%%%%%%%%%%%%%%%%%%%%%%%%%%%%%%%%%%%%%%%%%%%%
\begin{figure}[t]
	\adjustbox{max width=\textwidth}{$\begin{array}{c@{\quad\vrule\quad}c}
		\begin{array}{c@{\qquad}c}
			\multicolumn{2}{c}{
				\vlinf{\axrule}{}{\sdash \lsend xy, \lrecv xy}{}
			}
		\\[15pt]
			\vlinf{\lpar}{}{\sdash \Gamma, A\lpar B}{\sdash \Gamma, A, B}
		&
			\vliinf{\ltens}{}{\sdash[\sS_1 , \sS_2] \Gamma, A\ltens B,\Delta}{\sdash[\sS_1] \Gamma, A}{\sdash[\sS_2] B, \Delta}
		\\[15pt]
			\vlinf{\lplus}{}{\sdash\Gamma, A_1\lplus A_2}{ \sdash \Gamma, A_i}
		&
			\vliinf{\lwith}{}{\sdash \Gamma,A\lwith B}{\sdash \Gamma,A}{\sdash \Gamma,B}
		\\[15pt]
			\vlinf{\forall}{\dagger}{\sdash \Gamma, \lFa x{A}}{\sdash \Gamma, A}
		&
			\vlinf{\exists}{}{\sdash \Gamma, \lEx x{A}}{\sdash \Gamma, A\fsubst yx}
		\end{array}
	&
		\begin{array}{c}
			\vlinf{\lunit}{}{\sdash \lunit}{}
		\qquad
			\vliinf{\lprec}{n \geq 0}{
				\sdash[\sS_1 , \sS_2] \Gamma,\Delta, A_1\lprec B_1 , \dots, A_n \lprec B_n
			}{ \sdash[\sS_1] \Gamma, A_1, \dots, A_n}{\sdash[\sS_2] \Delta, B_1, \dots, B_n}
		\\[15pt]
			\vlinf{\nuur}{\dagger}{\sdash \Gamma, \lNu x A}{\sdash[\sS] \Gamma, A}
		\quad
			\vlinf{\nuloadr}{\dagger}{\sdash \Gamma, \lNu xA}{\sdash[\sS , \isnu x] \Gamma , A}
		\quad
			\vlinf{\nupopr}{}{\sdash[\sS , \isnu y] \Gamma, \lYa xA}{\sdash \Gamma,A\fsubst yx}
		\\[15pt]
			\vlinf{\yaur}{\dagger}{\sdash \Gamma, \lYa x A}{\sdash[\sS] \Gamma, A}
		\quad
			\vlinf{\yaloadr}{\dagger}{\sdash \Gamma,\lYa xA}{\sdash[\sS , \isya x] \Gamma , A}
		\quad
			\vlinf{\yapopr}{}{\sdash[\sS , \isya y] \Gamma, \lNu xA}{\sdash \Gamma, A\fsubst yx}
		\end{array}
	\end{array}$}
	\caption{Sequent calculus rules of $\PIL$, where $\dagger\coloneqq x\notin \freeof{\Gamma}$.}
	\label{fig:rules}
\end{figure}
%%%%%%%%%%%%%%%%%%%%%%%%%%%%%%%%%%%%%%%%%%%%%%%%%%%%%%%%%%%%%%%%%%%
%%%%%%%%%%%%%%%%%%%%%%%%%%%%%%%%%%%%%%%%%%%%%%%%%%%%%%%%%%%%%%%%%%%

%%%%%%%%%%%%%%%%%%%%%%%%%%%%%%%%%%%%%%%%%%%%%%%%%%%%%%%%%%%%%%%%%%%
\begin{nota}
	We write $\proves[\PIL]\Gamma$ to denote the existence of a derivation in $\PIL$ with conclusion the judgement $\sdash[\emptyset]\Gamma$.
	We denote by  $\vldownsmash{\vlderivation{\vlpr{\dD}{}{\sdash\Gamma}}}$ 
	a  derivation with conclusion $\sdash \Gamma$.
\end{nota}
%%%%%%%%%%%%%%%%%%%%%%%%%%%%%%%%%%%%%%%%%%%%%%%%%%%%%%%%%%%%%%%%%%%

\begin{remark}
	In the definition of $\PIL$, we did not discuss the definition of negation, which is defined by extending the one of $\MALL\wF$ by the De Morgan laws $\cneg \lunit=\lunit$, $\cneg{(A\lprec B)}=\cneg A \lprec\cneg B$, and $\cneg{(\lNu xA)}=\lYa x{(\cneg A)}$. 
	However, for the purpose of this paper, it suffices to only keep in mind that $\lsend --$ and $\lrecv --$, as well as that two nominal quantifiers, are dual to each other.

	Similarly, we do not discuss the cut rule because it has no computational interpretation in the proof-search semantics of $\PIL$ from \cite{acc:man:mon:FaP}.
	However, cut-elimination for $\PIL$ holds, and its proof requires a non-trivial treatment of the nominal quantifiers and the store. Details are provided in \cite{acc:man:mon:FaPext}.
\end{remark}
\begin{nota}
	For brevity, we may denote by $\nabla$ any nominal quantifier, and when we write $\cneg \nabla$ we assume that $\cneg \lnewsymb=\lyasymb$ and $\cneg \lyasymb=\lnewsymb$.
\end{nota}

%%%%%%%%%%%%%%%%%%%%%%%%%%%%%%%%%%%%%%%%%%%%%%%%%%%%%%%%%%%%%%%%
%%%%%%%%%%%%%%%%%%%%%%%%%%%%%%%%%%%%%%%%%%%%%%%%%%%%%%%%%%%%%%%%
\subsection{Links, Substitutions, and Dualizers}\label{sec:links}
%%%%%%%%%%%%%%%%%%%%%%%%%%%%%%%%%%%%%%%%%%%%%%%%%%%%%%%%%%%%%%%%
%%%%%%%%%%%%%%%%%%%%%%%%%%%%%%%%%%%%%%%%%%%%%%%%%%%%%%%%%%%%%%%%

%%%%%%%%%%%%%%%%%%%%%%%%%%%%%%%%%%%%%%%%%%%%%%%%%%%%%%%%%%%%%%%%

\def\wmaps{\sigma}
\def\wmapt{\tau}
\def\wmapr{\rho}

We conclude this section by recalling and extending the definitions of \emph{link} and \emph{dualizer} from the literature.

%%%%%%%%%%%%%%%%%%%%%%%%%%%%%%%%%%%%%%%%%%%%%%%%%%%%%%%%%%%%%%%%
\begin{definition}
	A \defn{link} $\la$ on a judgement $\sdash \Gamma$ is either:
	\begin{itemize}
		\item a \defn{sub-sequent} of $\sdash \Gamma$, that is, a sequent $\Gamma'$ which is a sub-forest of $\Gamma$;
		\item or a \defn{nominal link}, that is, either
		\begin{itemize}
			\item 
			a two-elements set $\set{x,y}\subset\varset$ of variables occurring in $\Gamma$ such that $x$ is bound by $\lnewsymb$ and $y$ is bound by a $\lyasymb$; or
			
			\item a two-elements set $\set{\isnu x,y}$ \resp{$\set{\isya x,y}$} consisting of a nominal variable $\isnu x$ \resp{$\isya x$} occurring in the store and a variable $y$ occurring in $\Gamma$ bound by the nominal quantifier $\lyasymb$ \resp{$\lnewsymb$};
		\end{itemize}
	\end{itemize}
	A link is \defn{axiomatic} if it is of the form $\set{\lunit}$ or $\set{\lsend xy,\lrecv zt}$, 
	or if it is a nominal link of the form $\set{x,y}\subset\varset$ containing no nominal variable.
	
	A \defn{linking} \resp{\defn{axiomatic linking}} on $\sdash \Gamma$ is a set of links 
	\resp{axiomatic links} on $\sdash \Gamma$.
\end{definition}

%%%%%%%%%%%%%%%%%%%%%%%%%%%%%%%%%%%%%%%%%%%%%%%%%%%%%%%%%%%%%%%%
\begin{nota}
	We represent a link $\la$ by drawing a horizontal line labeled by $\la$ connected via vertical segments 
	to each (root of a) subformula and to each variable in the link.
	For example, 
	$\sclr z =\set{\isnu z, v}$, $\lc= \set{x,y}$, $\la=\set{ A, C, D \lplus E}$, and $\lb= \set{B,(\lNu xC)\lpar(D\oplus E), \lYa y F}$ are valid links for the judgement below, which has been decorated with the links to show how they are represented.

	\begin{equation}\label{eq:exPreLink}
		\begin{array}{c}
			\sdash[\isya w , \isnu{\vz1}] \lYa {\vv2}{\vA1\ltens \vB1}, (\vNu1 \vx1.{\vC1})\vpz1{\lpar}(\vD1\vlplus1 \vE1), \vYa1 \vy1.\vF1
			\pzlinks{A1/C1/12/\la/red/lplus1}
			\pzlinks{x1/y1/-10/\lc/violet/}
			\pzlinks{B1/Ya1/-16/\lb/blue/pz1}
			\pzlinks{z1/v2/12/z/magenta/}
			\\[10pt]
		\end{array}
	\end{equation}
	Note that $\set{D,D\oplus E}$ is not a link for the judgement  above because the formula $D$ is repeated twice, therefore is not a sub-forest of the sequent in the judgement.
\end{nota}
%%%%%%%%%%%%%%%%%%%%%%%%%%%%%%%%%%%%%%%%%%%%%%%%%%%%%%%%%%%%%%%%

%%%%%%%%%%%%%%%%%%%%%%%%%%%%%%%%%%%%%%%%%%%%%%%%%%%%%%%%%%%%%%%%
A \defn{substitution} $\wmaps=\fsubsts{x_1/y_1,\ldots,x_n/y_n}$ is a map with domain $\set{y_1,\ldots, y_n}\subset\varset$ and image $\set{x_1,\ldots,x_n}\subset\varset$.\footnote{Note that we assume substitutions to map variables to variables because we have no function symbols in $\PIL$, that is, terms are variables only.}
If $X$ is a formula \resp{variable or judgement}, we write $\wmaps(X)$ or $X\wmaps$ for the formula \resp{variable or judgement} obtained from $X$ by simultaneously substituting each occurrence of the variable $y_i$ with an occurrence of the variable $x_i$ for all $i\in\intset1n$.

\begin{nota}
	We adopt the following notation for substitutions:
	\begin{itemize}
		\item 
		$\dualizerof[\emptyset]$ for the \defn{empty substitution};

		\item 
		$\wmaps\wmapt$ for the \defn{composition} of $\wmaps$ and $\wmapt$ (in which $\wmaps$ is applied \emph{after} $\wmapt$);

		\item 
		$\wmaps\fsubminus{x}$ for the substitution obtained restricting $\wmaps$ by removing the variable $x$ from its domain, and $\wmaps x$ from its image;

		\item
		$\wmaps \geq \wmapr$ if $\wmaps$ for \defn{more general} than $\wmapt$, that is, if $\wmaps=\wmapt\wmapr$ for some $\wmapt$;

		\item
		$\newcohe\wmaps{\wmapt}$ if 
		$\wmaps$ and $\wmapt$ are \defn{coherent} on their domains, that is, if $x\in\domof{\wmaps}\cap\domof{\wmapt}$ then $\wmaps(x)=\wmapt(x)$;

		\item
		$\wmaps\join\wmapt$ for the most general map $\wmapr$ such that $\wmaps\leq\wmapr$  and $\wmapt\leq\wmapr$;

	\end{itemize}
\end{nota}

%%%%%%%%%%%%%%%%%%%%%%%%%%%%%%%%%%%%%%%%%%%%%%%%%%%%%%%%%%%%%%%%
\begin{definition}\label{def:dualizer}
	A \defn{dualizer} $\dualizerof[\la]$ for a link $\la$ on a judgement $\sdash \Gamma$ is a substitution whose domain consists of 
	\begin{itemize}
		\item one of the two variables in $\la$ if $\la$ is a nominal link; or
		\item variables occurring in {$\la$} bound by an existential quantifier ($\exists$) or nominal quantifier ($\lnewsymb$ or $\lyasymb$) in $\Gamma$ otherwise.
	\end{itemize}
	
	A \defn{witness map} $\dualizerof^{\linking}$ for a linking $\linking$ is a map associating a (possibly empty) dualizer $\dualizerof[\la]^{\linking}$ to each link $\la$ in $\linking$.
	If $\linking$ is an axiomatic linking, then we say that a witness map $\dualizerof^{\linking}$ for $\linking$ is \defn{valid} if the following conditions hold for each $\la\in\linking$:
	\begin{itemize}
		\item if $\la = \set{\lsend xy,\lrecv zt} $, then  $\dualizerof[\la]^{\linking}(x) = \dualizerof[\la]^{\linking}(z)$ and $\dualizerof[\la]^{\linking}(y) = \dualizerof[\la]^{\linking}(t)$;
		
		\item if $\la=\set{x,y}$ is an axiomatic nominal link, then $\dualizerof[\la]^{\linking}(x)= \dualizerof[\la]^{\linking}(y)$;
		
		\item if $\la=\set\lunit$, then $\dualizerof[\la]^{\linking}=\dualizerof[\emptyset]$.
	\end{itemize}
\end{definition}
%%%%%%%%%%%%%%%%%%%%%%%%%%%%%%%%%%%%%%%%%%%%%%%%%%%%%%%%%%%%%%%%

%%%%%%%%%%%%%%%%%%%%%%%%%%%%%%%%%%%%%%%%%%%%%%%%%%%%%%%%%%%%%%%%
\begin{example}
	Consider the judgement  shown below, and its links $\la = \set{\lsend xy, \lrecv zw}$ and $\lb = \set{\lrecv tv,  {\lsend sv}}$ and $\lc=\set{s,t}$.
	The witness map 
	$\dualizerof^{\linking}=\set{\la\mapsto\dualizerof[\la],\lb\mapsto\dualizerof[\lb],\lc\mapsto\dualizerof[\lc]}$ with $\dualizerof[\la] = \fsubsts{x/z, y / w}$ and $\dualizerof[\lb] =\dualizerof[\lc] = \fsubsts{s/t}$ is valid.
	$$
	\vdash \lEx{x}\lEx{y} \vpz1{\lsend xy},\ \vpz2{\lrecv zw} \ltens (\lYa {\vt1} \vpz3{\lrecv tv}) , \lNu {\vs1} \vpz4{\lsend sv}
	\pzlinks{pz1/pz2/12/\la/red/}
	\pzlinks{pz3/pz4/12/\lb/blue/}
	\pzlinks{t1/s1/-10/\lc/violet/}
	$$
\end{example}
%%%%%%%%%%%%%%%%%%%%%%%%%%%%%%%%%%%%%%%%%%%%%%%%%%%%%%%%%%%%%%%%

%%%%%%%%%%%%%%%%%%%%%%%%%%%%%%%%%%%%%%%%%%%%%%%%%%%%%%%%%%%%%%%%
%%%%%%%%%%%%%%%%%%%%%%%%%%%%%%%%%%%%%%%%%%%%%%%%%%%%%%%%%%%%%%%%
%%%%%%%%%%%%%%%%%%%%%%%%%%%%%%%%%%%%%%%%%%%%%%%%%%%%%%%%%%%%%%%%
\section{Conflict Nets}\label{sec:CN}
%%%%%%%%%%%%%%%%%%%%%%%%%%%%%%%%%%%%%%%%%%%%%%%%%%%%%%%%%%%%%%%%
%%%%%%%%%%%%%%%%%%%%%%%%%%%%%%%%%%%%%%%%%%%%%%%%%%%%%%%%%%%%%%%%
%%%%%%%%%%%%%%%%%%%%%%%%%%%%%%%%%%%%%%%%%%%%%%%%%%%%%%%%%%%%%%%%

Conflict nets for $\MALL$ \cite{hei:hug:conflict} are trees alternating \emph{concord} ($\conc$) and \emph{conflict} ($\conf$) nodes and axiomatic links as leaves, and satisfying a contractability criterion with respect to a rewriting procedure called \emph{coalescence}.
Intuitively, leaves of these trees represent derivable sequents, and each coalescence step checks that an inference rule of the system can correctly apply to construct a new derivable sequent from the ones involved in the step. That is, coalescence checks that the structure of the proof net correctly reflects to the top-down construction of a derivation in the sequent calculus.
In this section, we extend the definition of conflict nets to the first-order setting of $\PIL$.

We recall the definition of \emph{concord-conflict tree} from \cite{hei:hug:conflict}.
%%%%%%%%%%%%%%%%%%%%%%%%%%%%%%%%%%%%%%%%%%%%%%%%%%%%%%%%%%%%%%%%
\begin{definition}\label{def:cotree}
	A \defn{concord-conflict tree} (or \defn{\cotree} for short) for a linking $\linking$ on a judgement $\sdash \Gamma$ is a tree $\linktree$ with $\linking$ as its set of leaves, 
	and internal nodes labeled by $\conc$ (\defn{concord} nodes) or by $\conf$ (\defn{conflict} nodes).
	We denote by $\canon{\linktree}$ the \defn{canonization} of a \cotree $\linking$, that is, the \cotree obtained by merging adjacent $\conc$-nodes \resp{$\conf$-nodes}, and by removing nodes with a single child by attaching its child to its parent.
	A \cotree $\linktree$ is \defn{axiomatic} if $\linking$ is so; it is \defn{canonical} if $\linktree=\canon{\linktree}$.
\end{definition}
%%%%%%%%%%%%%%%%%%%%%%%%%%%%%%%%%%%%%%%%%%%%%%%%%%%%%%%%%%%%%%%%

We enrich the definition of conflict nets and coalescence from \cite{hei:hug:conflict} with the information about witnesses, which is required to properly deal with quantifiers in the first-order setting similarly to \cite{hug:unification,hei:hug:str:ALL1,hug:str:wu:CP1}.
%%%%%%%%%%%%%%%%%%%%%%%%%%%%%%%%%%%%%%%%%%%%%%%%%%%%%%%%%%%%%%%%
\begin{definition}\label{def:coalescence}
	A \defn{pre-structure} is a pair $\tuple{\linktree,\dualizerof^{\linking}}$ on $\Gamma$ made of a \cotree $\linktree$ and a witness map $\dualizerof^{\linking}$ for the linking $\linking$ of $\linktree$.
	It is \defn{trivial} if $\linktree$ consists of a single leaf $\la$ and $\dualizerof[\la]^{\linking}=\dualizerof[\emptyset]$.
	A \defn{proof structure} for a sequent $\Gamma$ is a pre-structure $\tuple{\linktree,\dualizerof}$ where $\linktree$ is axiomatic and canonical, and $\dualizerof^{\linking}$ is valid.

	A \defn{proof net} $\tuple{\linktree,\dualizerof^{\linking}}$ for a sequent $\Gamma$ 
	is a proof structure for $\Gamma$ which is also \defn{coalescent}, that is, it admits a sequence of 
	\defn{coalescence steps} from \Cref{fig:coalescence} transforming $\tuple{\linktree,\dualizerof}$ into a trivial pre-structure.
	See \Cref{fig:coalescence} for an example.
\end{definition}
%%%%%%%%%%%%%%%%%%%%%%%%%%%%%%%%%%%%%%%%%%%%%%%%%%%%%%%%%%%%%%%%

%%%%%%%%%%%%%%%%%%%%%%%%%%%%%%%%%%%%%%%%%%%%%%%%%%%%%%%%%%%%%%%%
% Fig Coalescence
%%%%%%%%%%%%%%%%%%%%%%%%%%%%%%%%%%%%%%%%%%%%%%%%%%%%%%%%%%%%%%%%
\begin{figure}[t]
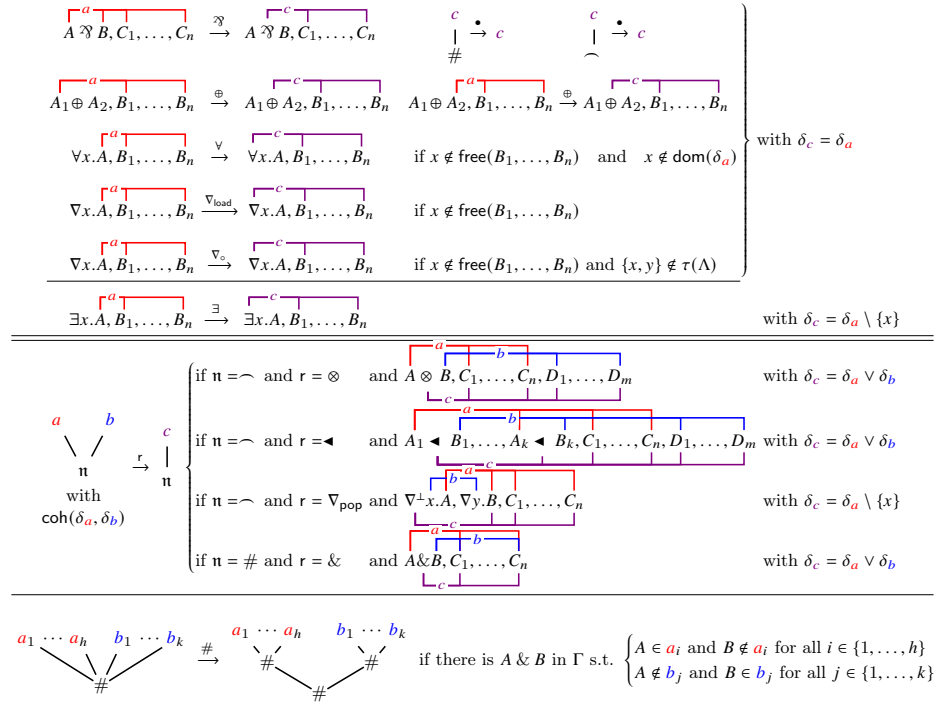

	\adjustbox{max width=\textwidth}{$\begin{array}{c}
		\begin{array}{ll}
			\left.\begin{array}{rcl@{\quad}l}
				\vA1 \vpz1{\lpar} \vB1, \viC1,\ldots ,\viC n \pzlinks{A1/B1/12/\la/red/{C1,Cn}}
			&\coalto[\lpar]&
				\vA1 \vpz1{\lpar} \vB1, \viC1,\ldots ,\viC n \pzlinks{pz1/C1/12/\lc/violet/{Cn}}
			&
			\qquad
				\begin{array}{c}\vpz1{\lc} \\[10pt]\vconf1 \end{array}
				\multiGedges{conf1}{pz1}
			\coalto[\bullet]
				\vpz0{\lc}
			\qquad\qquad
				\begin{array}{c}\vpz2{\lc}\\[10pt]\vconc1\end{array}
				\multiGedges{conc1}{pz2}
			\coalto[\bullet]
				\vpz5{\lc}
			\\\\
				\viA1 \!\vpz1{\lplus}\! \viA2, \viB1,\ldots ,\viB n \pzlinks{A1/B1/12/\la/red/{Bn}}
			&\coalto[\lplus]&
				\viA1 \!\vpz1{\lplus}\! \viA2, \viB1,\ldots ,\viB n \pzlinks{pz1/B1/12/\lc/violet/{Bn}}
			&
				\viA1 \!\vpz1{\lplus}\! \viA2, \viB1,\ldots ,\viB n \pzlinks{A2/B1/12/\la/red/{Bn}}
			\coalto[\lplus]
				\viA1 \!\vpz1{\lplus}\! \viA2, \viB1,\ldots ,\viB n \pzlinks{pz1/B1/12/\lc/violet/{Bn}}
			\\\\
				\vpz1{\forall}\!x.\vA1, \viB1,\ldots ,\viB n \pzlinks{A1/B1/12/\la/red/{Bn}}
			&\coalto[\forall]&
				\vpz1{\forall}\!x.\vA1, \viB1,\ldots ,\viB n \pzlinks{pz1/B1/12/\lc/violet/{Bn}}
			&
				\text{ if } x\notin\freeof{B_1,\ldots,B_n} \quand x\notin\domof{\dualizerof[\la]}
			\\\\
				\vpz1{\nabla}\!x.\vA1, \viB1,\ldots ,\viB n \pzlinks{A1/B1/12/\la/red/{Bn}}
			&\coalto[\naloadr]&
				\vpz1{\nabla}\!x.\vA1, \viB1,\ldots ,\viB n \pzlinks{pz1/B1/12/\lc/violet/{Bn}}
			&
				\text{ if } x\notin\freeof{B_1,\ldots,B_n}
			\\\\
				\vpz1{\nabla}\!x.\vA1, \viB1,\ldots ,\viB n \pzlinks{A1/B1/12/\la/red/{Bn}}
			&\coalto[\naur]&
				\vpz1{\nabla}\!x.\vA1, \viB1,\ldots ,\viB n \pzlinks{pz1/B1/12/\lc/violet/{Bn}}
			&
				\text{ if } x\notin\freeof{B_1,\ldots,B_n} \mbox{ and } \set{x,y}\notin\linktree
			\\\hline
			\end{array}\right\}
		&
			\text{with }\dualizerof[\lc]=\dualizerof[\la]
		\\\\
			\begin{array}{rcl}
				\quad\vpz1{\exists}\!x.\vA1, \viB1,\ldots ,\viB n \pzlinks{A1/B1/12/\la/red/{Bn}}
			&\;\coalto[\;\exists\;]\;&
				\vpz1{\exists}\!x.\vA1, \viB1,\ldots ,\viB n \pzlinks{pz1/B1/12/\lc/violet/{Bn}}
			\end{array}
		&
			\text{ with } \dualizerof[\lc]=\dualizerof[\la]\fsubminus{x}
		\end{array}
	\\\hline\hline\\
		\begin{array}{cccl}
			\begin{array}{c}
				\vpz1{\la}
				\qquad
				\vpz2{\lb}
				\\\\
				\vpz0{\mathfrak n}
				\\
				\text{with}
				\\
				\newcohe{\dualizerof[\la]}{\dualizerof[\lb]}
			\end{array}
			\multiGedges{pz0}{pz1,pz2}
		&\coalto[\rrule]&
			\begin{array}{c}
				\vpz1{\lc}
				\\\\
				\vpz0{\mathfrak n}
				\\
				\\
			\end{array}
			\multiGedges{pz0}{pz1}
		&
			\left\{\begin{array}{@{\text{if }}l@{\text{ and }}l@{\text{ and }}l@{\text{ with }}l}
				\mathfrak n=\conc
				&
				\rrule=\ltens
				&
				\vA1 \vpz1{\ltens} \vB1, \viC1,\ldots ,\viC n,\viD1,\ldots ,\viD m
				\pzlinks{A1/C1/16/\la/red/{Cn}} 
				\pzlinks{B1/D1/12/\lb/blue/{Dm}}
				\pzlinks{pz1/C1/-12/\lc/violet/{Cn,D1,Dm}}
				&
				\dualizerof[\lc] = \dualizerof[\la] \join \dualizerof[\lb]
			\\[20pt]
				\mathfrak n=\conc
				&
				\rrule=\lprec
				&
				\viA1 \vpz1{\lprec} \viB1,\ldots,\viA{k} \vpz2{\lprec} \viB{k}, \viC1,\ldots ,\viC n,\viD1,\ldots ,\viD m 
				\pzlinks{A1/Ak/16/\la/red/{C1,Cn}} 
				\pzlinks{B1/Bk/12/\lb/blue/{D1,Dm}}
				\pzlinks{pz1/pz2/-12/\lc/violet/{C1,Cn,D1,Dm}}
				&
				\dualizerof[\lc] = \dualizerof[\la] \join \dualizerof[\lb]
			\\[20pt]
				\mathfrak n=\conc
				&
				\rrule=\napopr
				&
				\vnNa1 \vx1.\vA1, \vNa1 \vy1.\vB1, \viC1,\ldots ,\viC n
				\pzlinks{A1/B1/16/\la/red/{C1,Cn}}
				\pzlinks{x1/y1/12/\lb/blue/}
				\pzlinks{nNa1/B1/-12/\lc/violet/{C1,Cn}}
				&
				\dualizerof[\lc]=\dualizerof[\la]\fsubminus{x}
			\\[20pt]
				\mathfrak n=\conf
				&
				\rrule=\lwith
				&
				\vA1 \vlwith1 \vB1, \viC1,\ldots ,\viC n
				\pzlinks{A1/C1/16/\la/red/{Cn}}
				\pzlinks{B1/Cn/12/\lb/blue/{C1}}
				\pzlinks{lwith1/C1/-12/\lc/violet/{Cn}}
				&
				\dualizerof[\lc] = \dualizerof[\la] \join \dualizerof[\lb]
			\end{array}\right.
		\end{array}
	\\\\\hline\\
		\!\!
		\begin{array}{cccl}
			\begin{array}{c}
				\vpz1{\la_1}\cdots \vpz3{\la_h}
				\quad
				\vpz2{\lb_1}\cdots\vpz4{\lb_k}
				\\\\
				\vconf1
			\end{array}
			\multiGedges{conf1}{pz1,pz2,pz3,pz4}
		&\coalto[\conf]&
			\begin{array}{ccc}
				\vpz1{\la_1}\cdots \vpz3{\la_h}
				&&
				\vpz2{\lb_1}\cdots\vpz4{\lb_k}
				\\[5pt]
				\vconf2&& \vconf3
				\\[5pt]
				&\vconf1
			\end{array}
			\multiGedges{conf1}{conf2,conf3}
			\multiGedges{conf2}{pz1,pz3}
			\multiGedges{conf3}{pz2,pz4}
		&
			\text{if there is } A \lwith B \text{ in } \Gamma \text{ s.t. }
			\begin{cases}
				A \in \la_i \text{ and } B \notin \la_i \text{ for all } i \in \intset{1}{h}
				\\
				A \notin \lb_j \text{ and } B \in \lb_j \text{ for all } j \in \intset{1}{k}
			\end{cases}
			\hskip-1.5em
		\end{array}
	\end{array}$}
	\caption{Coalescence steps for $\PIL$.}
	\label{fig:coalescence}
\end{figure}
%%%%%%%%%%%%%%%%%%%%%%%%%%%%%%%%%%%%%%%%%%%%%%%%%%%%%%%%%%%%%%%%
%%%%%%%%%%%%%%%%%%%%%%%%%%%%%%%%%%%%%%%%%%%%%%%%%%%%%%%%%%%%%%%%

We can define the proof translation by translating top-down the rules of a given derivation in $\PIL$ as shown in \Cref{fig:deseq}, allowing us to prove that our \conets are complete with respect to $\PIL$.
% 
%%%%%%%%%%%%%%%%%%%%%%%%%%%%%%%%%%%%%%%%%%%%%%%%%%%%%%%%%%%%%%%%
\begin{lemma}\label{lem:confTranslation}
	Let $\Gamma$ be a sequent.
	If $\dD$ is a derivation in $\PIL$ with conclusion $\Gamma$,
	then there is a \conet for $\Gamma$.
\end{lemma}
\begin{proof}
	We define $\cnof{\dD}=\tuple{\treeof{\dD},\witnessof{\dD}}$ by translating top-down the rules of $\dD$ using the construction in \Cref{fig:deseq}, where
	we denote by $\ltree_1 \conc \ltree_2$ \resp{$\ltree_1 \conf \ltree_2$} the \cotree with root a $\conc$-node \resp{$\conf$-node} having as children the roots of $\ltree_1$ and $\ltree_2$. 
	Note that the translation uses the (not necessarily sound) rule $\leafr$ which generalizes $\axrule$ up-to the existence of a proper unifier.
	This is due to the necessity of updating the witness map and the leaves of the \cotree  if we want to give an inductive translation of derivations in $\PIL$ which is compositional with respect to the structure of the derivation.

	To prove that  $\cnof{\dD}$ is coalescent, it suffices to consider the sequence of coalescence steps corresponding to the $\PIL$ rules used in the proof, in the same order in which the rules of $\dD$ have been translated: each side condition for the coalescence steps in \Cref{fig:coalescence} is satisfied by construction.
	Moreover, in the translation of the rules $\exists$, $\nuloadr$, and $\yaloadr$, the witness map is updated so that, once the corresponding quantifier is coalesced, the associated dualizers are `emptied'.
\end{proof}
%%%%%%%%%%%%%%%%%%%%%%%%%%%%%%%%%%%%%%%%%%%%%%%%%%%%%%%%%%%%%%%%

%%%%%%%%%%%%%%%%%%%%%%%%%%%%%%%%%%%%%%%%%%%%%%%%%%%%%%%%%%%%%%%%
% FIG: Desequ PIL
%%%%%%%%%%%%%%%%%%%%%%%%%%%%%%%%%%%%%%%%%%%%%%%%%%%%%%%%%%%%%%%%
\begin{figure}[t]
	\centering\adjustbox{max width=\textwidth}{$\begin{array}{c}
		%
		%multiplicative
		\cnof{
			\vlderivation{
				\vliin{\rrule[2]}{}{\sdash[\sS_1,\sS_2] \Gamma}{
					\vlpr{\dD_1}{}{\sdash[\sS_1] \Gamma_1}
				}{
					\vlpr{\dD_2}{}{\sdash[\sS_2] \Gamma_2}
				}
			}
		}
		=
		\Tuple{
			\canon{\treeof{\dD_1} \conc \treeof{\dD_2}}
		, 
			\witnessof{\dD_1} \cup \witnessof{\dD_2}
		}
	\qquad
		% 
		%unary
		\cnof{
			\vlderivation{\vlin{\rrule[1]}{}{\sdash \Gamma}{\vlpr{\dD_1}{}{\sdash[\sS_1] \Gamma_1}}}
		}
		=
		\cnof{\dD_1}
	\\[20pt]
		%unit
		\cnof{
			\vlderivation{\vlin{\lunit}{}{\sdash \lunit}{\vlhy{}}}
		}
		=
		\Tuple{
			\set{\vpz1{\lunit}},
			\dualizerof[\emptyset]
		}
	\qquad
		%
		%additive
		\cnof{
			\vlderivation{
				\vliin{\lwith}{}{\sdash \Gamma}{
					\vlpr{\dD_1}{}{\sdash \Gamma_1}
				}{
					\vlpr{\dD_2}{}{\sdash \Gamma_2}
				}
			}
		}
		=
		\Tuple{
			\canon{\treeof{\dD_1} \conf \treeof{\dD_2}}
		, 
			\witnessof{\dD_1} \uplus \witnessof{\dD_2}
		}
	\\[20pt]
		%
		%axioms
		\cnof{
			\vlinf{\leafr}{}{\sdash \lsend xy, \lrecv zt}{\vlhy{}}
		}
		=
		\Tuple{
			\set{\lsend xy,\lrecv zt}
		, 
			\dualizerof[\emptyset]
		}
	\quad
		%
		%exists
		\cnof{
			\vlderivation{
				\vlin{\exists}{}{\sdash \Gamma, \lEx xA}{
					\vlpr{\dD_y}{}{\sdash \Gamma, A\fsubst yx}
				}
			}
		}
		=
		\Tuple{
			\treeof{\dD_x}
		, 
			\Set{
				\la_x \mapsto \dualizerof[\la_x]^{\dD} \mid \la_y \in \domain(\witnessof{\dD_y})
			}
		}
	\\[20pt]
		%
		%pop
		\cnof{
			\vlderivation{
				\vlin{\nqpopr}{}{\sdash[\sS, \isna y] \Gamma, \lnNa xA}{
					\vlpr{\dD_y}{}{\sdash \Gamma, A\fsubst yx}
				}
			}
		}
		=
		\Tuple{
			\canon{\set{ x, \isna y} \conc \treeof{\dD_x} }
		, 
			\Set{\set{ x, \isna y}\mapsto \fsubst{y}{x}}
			\cup
			\Set{
				\la_x \mapsto \dualizerof[\la_x]^{\dD} \mid \la_y \in \domain(\witnessof{\dD_y})
			}
		}
	\end{array}$}

	\caption{
		Translation of a derivation in $\PIL$ into a conflict net,
		with 
		$\rrule[1]\in\set{\lpar,\oplus,\forall,\nurule,\yurule,\nuloadr,\yaloadr}$,
		and
		$\rrule[2]\in\set{\ltens,\lprec}$, 
		and
		where we denote by
		$\dD_x $ \resp{$\la_x$} the derivation \resp{link} obtained from $\dD_y$ \resp{$\la_y$} by replacing all occurrences of the variable $y$ involved in the rule $\rrule\in\set{\exists,\napopr}$ with $x$ and where $\dualizerof[\la_x]^{\dD} =\fsubst{y}{x} \dualizerof[\la_y]^{\dD_y}$ if $\la_x\neq\la_y$ and $\dualizerof[\la_y]^{\dD_y}$ otherwise.
	}
	\label{fig:deseq}
\end{figure}
%%%%%%%%%%%%%%%%%%%%%%%%%%%%%%%%%%%%%%%%%%%%%%%%%%%%%%%%%%%%%%%%
%%%%%%%%%%%%%%%%%%%%%%%%%%%%%%%%%%%%%%%%%%%%%%%%%%%%%%%%%%%%%%%%

To prove the adequacy of conflict nets with respect to $\PIL$, it now suffices to provide a sequentialization procedure associating to a given \conet for $\Gamma$ a derivation in $\PIL$ with conclusion $\Gamma$.

%%%%%%%%%%%%%%%%%%%%%%%%%%%%%%%%%%%%%%%%%%%%%%%%%%%%%%%%%%%%%%%%
\begin{restatable}{theorem}{thmCNsoundAndComplete}\label{thm:CNsoundAndComplete}
	There is a \conet $\tuple{\linktree,\dualizerof}$ for $\Gamma$
	iff $\proves[\PIL]\Gamma$.
\end{restatable}
\begin{proof}
	The right-to-left implication is consequence of \Cref{lem:confTranslation}.
	
	To prove the other implication, we provide a sequentialization procedure for \conets to construct a derivation in $\PIL$ from a given \conet for $\Gamma$ by induction on the length $n$ of a coalescence path for the \conet, which exists by definition of coalescent \cotree with witnesses.
	We define \emph{deductive \cotrees} pre-structures where leaves which are sequents (i.e., not nominal links) are labeled by derivations with conclusion of a judgement for the given sequent.
	
	If $n=0$, then the \conet is trivial.
	The only possible cases are if either $\Gamma=\lunit$ or $\Gamma=\lsend xy, \lrecv xy$ for some $x,y\in\varset$ since $
	\dualizerof[\Gamma]$ must be valid and empty.
	In both cases, we can easily construct a derivation in $\PIL$ with conclusion $\Gamma$ by applying the rules $\lunit$ or $\leafr$ respectively.

	If $n>0$, then we reason by cases on the last coalescence step in the coalescence path as shown in \Cref{fig:coacot}, reminding the reader that the coalescence steps defined `top-down' operating on leaves.
	In particular, each coalescence step corresponds to a rule of $\PIL$, the side conditions for the coalescence steps ensure that the corresponding rule of $\PIL$ can be applied, and the new leaf produced by the coalescence step is labeled by a derivation in $\PIL$ obtained by applying the corresponding rule to the derivations labeling the leaves which are coalesced.
\end{proof}
%%%%%%%%%%%%%%%%%%%%%%%%%%%%%%%%%%%%%%%%%%%%%%%%%%%%%%%%%%%%%%%%

%%%%%%%%%%%%%%%%%%%%%%%%%%%%%%%%%%%%%%%%%%%%%%%%%%%%%%%%%%%%%%%%
% FIG sequentialization
%%%%%%%%%%%%%%%%%%%%%%%%%%%%%%%%%%%%%%%%%%%%%%%%%%%%%%%%%%%%%%%%
\begin{figure*}[t]
	\centering\adjustbox{max width=\textwidth}{$\begin{array}{c}
		\begin{array}{c|c}
		\begin{array}[t]{c|c|c}
			\dD_{\la}& \mbox{step} & \dD_{\lc}
		\\\hline
			%Par
		 	\vlderivation{
		 		\vlpr{\pi}{}{\dualizerof[{\la}](\sdash A, B, \Gamma)}
	 		}
		 	&\lcoalto{\lpar}&
		 	\vlderivation{
		 		\vlin{\lpar}{}{
					\dualizerof[{\lc}](\sdash A \lpar B, \Gamma)
				}{
		 			\vlpr{\pi}{}{
						\dualizerof[{\la}](\sdash A, B, \Gamma)
					}
	 			}
		 	}
 		\\
			%Plus
	 		\vlderivation{
				\vlpr{\pi }{}{\dualizerof[{\la}](\sdash A_i, \Gamma)}
			}
	 		&\lcoalto{\oplus}&
	 		\vlderivation{
	 			\vlin{\lpar}{}{
					\dualizerof[{\lc}](\sdash A_1 \oplus A_2, \Gamma)
				}{
	 				\vlpr{\pi }{}{\dualizerof[{\la}](\sdash A_i , \Gamma)}
	 			}
	 		}
 		\\
 			%forall
 			\vlderivation{
 				\vlpr{\pi}{}{\dualizerof[{\la}](\sdash A , \Gamma)}
 			}
 			&\lcoalto{\forall}&
 			\vlderivation{
				\vlin{\forall}{}{
					\dualizerof[{\lc}](\sdash \lFa xA , \Gamma)
 				}{
					\vlpr{\pi}{}{\dualizerof[{\la}](\sdash A , \Gamma)}
				}
 			}
	   \\
			%exists
			\vlderivation{
				\vlpr{\pi}{}{\dualizerof[{\la}](\sdash A\fsubst yx , \Gamma)}
			}
			&\lcoalto{\exists}&
			\vlderivation{
				\vlin{\exists}{}{
						\dualizerof[{\lc}](\sdash \lEx xA, \Gamma)
					}{
					\vlpr{\pi}{}{
						\dualizerof[{\la}](\sdash A\fsubst yx , \Gamma)
					}
				}
			}
		\end{array}
		&
\begin{array}[t]{c|c|c}
			\dD_{\la}& \mbox{step} & \dD_{\lc}
		\\\hline
			%nuur
			\vlderivation{
				\vlpr{\pi}{}{\dualizerof[{\la}](\sdash A , \Gamma)}
			}
			&\lcoalto{\naur}&
			\vlderivation{
			   \vlin{\naur}{}{
				   \dualizerof[{\lc}](\sdash \lNa xA , \Gamma)
				}{
				   \vlpr{\pi}{}{\dualizerof[{\la}](\sdash A , \Gamma)}
			   }
			}
	   \\
 			%nupop
 			\vlderivation{
 				\vlpr{\pi}{}{\dualizerof[{\la}](\sdash A\fsubst yx , \Gamma)}
 			}
 			&\lcoalto{\napopr}&
 			\vlderivation{
				\vlin{\napopr}{}{
 					 \dualizerof[{\lc}]( \sdash[\sS,\isna x]  \lnNa xA, \Gamma)
 				}{
					\vlpr{\pi}{}{\dualizerof[{\la}](\sdash A\fsubst yx, \Gamma)}
				}
 			}
 		\\
			%Load
			\vlderivation{
				\vlpr{\pi}{}{\dualizerof[{\la}](\sdash[\sS,\isna x] A , B, \Gamma)}
			}
			&\lcoalto{\naloadr}&
			\vlderivation{\vlin{\naloadr}{}{
					\dualizerof[{\lc}](\sdash \lNa xA , \Gamma)
				}{
					\vlpr{\pi}{}{
						\dualizerof[{\la}](\sdash[\sS,\isna x] A , \Gamma)
					}
				}
			}
		\\
		\end{array}
		\end{array}
	\\[25pt]\\
		\begin{array}[t]{c|c|c|c}
			\dD_{\la}& \dD_{\lb}&  \mbox{step} & \dD_{\lc}
		\\\hline
		%tensor
			\vlderivation{\vlpr{\pi_1}{}{\dualizerof[{\la}](\sdash[\sS_1] A, \Gamma)}}
		&
			\vlderivation{\vlpr{\pi_2}{}{\dualizerof[{\lb}](\sdash[\sS_2] B, \Delta)}}
		&\lcoalto{\ltens}&
			\vlderivation{
				\vliin{\ltens}{}{
					\dualizerof[{\lc}](\sdash[\sS_1,\sS_2]A \ltens B, \Gamma,\Delta)
				}{
					\vlpr{\pi_1}{}{\dualizerof[{\la}](\sdash[\sS_1]A, \Gamma)}
				}{
					\vlpr{\pi_2}{}{\dualizerof[{\lb}](\sdash[\sS_2]B, \Delta)}
				}
			}
		\\
		%with
			\vlderivation{\vlpr{\pi_1}{}{\dualizerof[{\la}](\sdash[\sS]A, \Gamma)}}
			&
			\vlderivation{\vlpr{\pi_2}{}{\dualizerof[{\lb}](\sdash[\sS]B, \Delta)}}
			&\lcoalto{\lwith}&
			\vlderivation{
				\vliin{\lwith}{}{
					\dualizerof[{\lc}](\sdash[\sS] A \lwith B, \Gamma,\Delta)
				}{
					\vlpr{\pi_1}{}{\dualizerof[{\la}](\sdash[\sS]A, \Gamma)}
				}{
					\vlpr{\pi_2}{}{\dualizerof[{\lb}](\sdash[\sS]B, \Delta)}
				}
			}
		\\
			\vlderivation{\vlpr{\pi_1}{}{\dualizerof[{\la}](\sdash[\sS_1]A_1,\dots, A_n, \Gamma)}}
			&
			\vlderivation{\vlpr{\pi_2}{}{\dualizerof[{\lb}](\sdash[\sS_2]B_1, \dots, B_n, \Delta)}}
			&\lcoalto{\lprec}&
			\vlderivation{
				\vliin{\lprec}{}{
					\dualizerof[{\lc}](\sdash[\sS_1,\sS_2]A_1 \lprec B_1, \dots, A_n \lprec B_n, \Gamma,\Delta)
				}{
					\vlpr{\pi_1}{}{\dualizerof[{\la}](\sdash[\sS_1]A_1,\dots, A_n, \Gamma)}
				}{
					\vlpr{\pi_2}{}{\dualizerof[{\lb}](\sdash[\sS_2]B_1, \dots, B_n, \Delta)}
				}
			}
		\\
		\end{array}
	\end{array}$}
	\caption{
		Effect of coalescence steps in \Cref{fig:coalescence} on \cotrees with leaves labeled by derivations.
		The steps $\bullet$ and $\conf$ change no link labels. 
	}
	\label{fig:coacot}
\end{figure*}
%%%%%%%%%%%%%%%%%%%%%%%%%%%%%%%%%%%%%%%%%%%%%%%%%%%%%%%%%%%%%%%%
% END FIG sequentialization
%%%%%%%%%%%%%%%%%%%%%%%%%%%%%%%%%%%%%%%%%%%%%%%%%%%%%%%%%%%%%%%%

The \defn{size} of a \conet $\tuple{\linktree,\dualizerof^{\linking}}$ for a sequent $\Gamma$ is the number $\sizeof{\linktree}$ of nodes in the \cotree $\linktree$ plus the size of the sequent $\Gamma$.
%%%%%%%%%%%%%%%%%%%%%%%%%%%%%%%%%%%%%%%%%%%%%%%%%%%%%%%%%%%%%%%%
\begin{proposition}\label{prop:coal}
	It is possible to check if a \cotree with witnesses is coalescent in polynomial time with respect to its size.
\end{proposition}
\begin{proof}
	The result follows from the same argument (and algorithms) used in \cite{hei:hug:conflict} in the proof of the similar result for $\MALL$, where the complexity is $\mathcal O(n^4)$.

	The new multiplicative coalescence steps (the ones involving the $\lprec$) are as complex as the $\ltens$. Note that the additional side condition on the $\ltens$ and $\lprec$ steps can be checked in linear time with respect to the size of the formula (tree).
	The additional manipulations of the dualizers are all linear in the size of the dualizer (see \cite{mar:mon:unification}), thus linear in the size of the formula.
	Thus, complexity has an upper bound of $\mathcal O(n^5)$.
\end{proof}
%%%%%%%%%%%%%%%%%%%%%%%%%%%%%%%%%%%%%%%%%%%%%%%%%%%%%%%%%%%%%%%%

%%%%%%%%%%%%%%%%%%%%%%%%%%%%%%%%%%%%%%%%%%%%%%%%%%%%%%%%%%%%%%%%
\begin{example}
	Consider the sequent $\Gamma=\lsend ab\lwith \lsend ab, \lrecv ab \ltens \lNu x\lsend xz, \lYa y\lrecv yz$ and the following coalescence path:
	\\
	\noindent\adjustbox{max width=\textwidth}{$
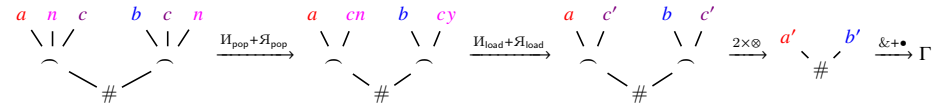
		\begin{array}{ccccccc}
			\vpz1{\la}
			&
			\vpz8{\lnom}
			&
			\vpz5{\lc}
			&&
			\vpz2{\lb}
			&
			\vpz7{\lc}
			&
			\vpz9{\lnom}
			\\[10pt]
			&\vpz3{\conc}&
			&&
			&\vpz4{\conc}
			\\
			&&&\vpz0{\conf}
		\end{array}
		\multiGedges{pz0}{pz3,pz4}
		\multiGedges{pz3}{pz1,pz5,pz8}
		\multiGedges{pz4}{pz2,pz7,pz9}
		\coalto[\nupopr+\yapopr]
		\begin{array}{c@{}c@{}c}
			\vpz1{\la}
			\quad
			\vpz5{\lcnu}
			&&
			\vpz2{\lb}
			\quad
			\vpz9{\lcya}
		\\[10pt]
			\vpz3{\conc}
			&&
			\vpz4{\conc}
			\\
			&\vpz0{\conf}
		\end{array}
		\multiGedges{pz0}{pz3,pz4}
		\multiGedges{pz3}{pz1,pz5}
		\multiGedges{pz4}{pz2,pz9}
		\coalto[\nuloadr+\yaloadr]
		\begin{array}{ccc}
			\vpz1{\la}
		\quad
			\vpz5{\lcdue}
			&&
			\vpz2{\lb}
		\quad
			\vpz9{\lcdue}
			\\[10pt]
			\vpz3{\conc}
			&&
			\vpz4{\conc}
			\\
			&\vpz0{\conf}
		\end{array}
		\multiGedges{pz0}{pz3,pz4}
		\multiGedges{pz3}{pz1,pz5}
		\multiGedges{pz4}{pz2,pz9}
	\coalto[2\times\ltens]
		\begin{array}{c}
			\vpz3{\laprimo}
		\qquad
			\vpz4{\lbprimo}
		\\
			\vpz0{\conf}
		\end{array}
		\multiGedges{pz0}{pz3,pz4}
	\coalto[\lwith + \bullet]
		\Gamma
	$}
	\\
	where we are considering the links on $\Gamma$ and on two judgements appearing during sequentialization:

	$$
		\vdash
		\vpz8{\lsend ab} \!\vlwith1\! \vpz9{\lsend ab}
		,\!
		\vpz1{\lrecv ab} 
		\vltens1 
		\vNu1{\vx1}{\vpz2{\lsend xz}}
		,
		\vYa1{\vy1}\vpz3{\lrecv yz}
		\pzlinks{pz8.75/pz1.105/4/\la/red/}
		\pzlinks{pz9.-75/pz1.-105/-4/\lb/blue/}
		\pzlinks{pz2/pz3/24/\lc/violet/}
		\pzlinks{Nu1.-75/Ya1.-105/-6/\lcdue/violet/}
		\pzlinks{x1/y1/20/\lnom/skewcolor/}
		\pzlinks{pz8/Ya1/16/\laprimo/red/{Nu1,pz1}}
		\pzlinks{pz9/Ya1/-16/\lbprimo/blue/{Nu1,pz1}}
	\quad
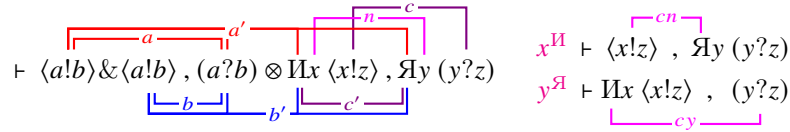
		\begin{array}{c}
			\vpz0{\isnu x}
			\vdash
			\vpz3{\lsend xz}
			\;,\;
			\vYa1{\vy1}\vpz4{\lrecv yz}
			\pzlinks{pz3/Ya1/12/\lcnu/skewcolor/}
		\\
			\vpz2{\isya y}
			\vdash
			\vNu1{\vx1}{\vpz3{\lsend xz}}
			\;,\;
			\vpz4{\lrecv yz}
			\pzlinks{Nu1/pz4/-12/\lcya/skewcolor/}
		\end{array}
	$$

	\noindent with 
	$\dualizerof[\la]=\dualizerof[\lb]=\dualizerof[\laprimo]=\dualizerof[\lbprimo]=\dualizerof[\lcdue]=\dualizerof[\emptyset]$
	and 
	$\dualizerof[\lnom]=\dualizerof[\lc]=\dualizerof[\lcnu]=\dualizerof[\lcya]=\fsubst{x}{y}$.

	The derivation we obtain by sequentialization is the following:
	\begin{equation}\label{eq:coalExample}
	\adjustbox{max width=.92\textwidth}{$\vlderivation{
		\vliin{\lwith}{}{
			\vdash \lsend ab \with \lsend ab, \lrecv ab \ltens \lNu{\vx1}{\lsend xz}, \vYa1{\vy1}\lrecv yz
		}{
			\vliin{\ltens}{}{
				\vdash \lsend ab, \lrecv ab \ltens \lNu{\vx1}{\lsend xz}, \vYa1{\vy1}\lrecv yz
			}{
				\vlin{\axrule}{}{\vdash\lsend ab, \lrecv ab }{\vlhy{}}
			}{
				\vlin{\nuloadr}{}{
					\vdash \lNu{\vx1}{\lsend xz}, \vYa1{\vy1}\lrecv yz
				}{
					\vlin{\nupopr}{}{
						\isnu y \vdash \lNu{\vx1}{\lsend xz}, \lrecv yz
					}{
						\vlin{\axrule}{}{
							\vdash \lsend yz , \lrecv yz
						}{\vlhy{}}
					}
				}
			}
		}{
			\vliin{\ltens}{}{
				\vdash \lsend ab, \lrecv ab \ltens \lNu{\vx1}{\lsend xz}, \vYa1{\vy1}\lrecv yz
			}{
				\vlin{\axrule}{}{\vdash\lsend ab, \lrecv ab }{\vlhy{}}
			}{
				\vlin{\yaloadr}{}{
					\vdash \lNu{\vx1}{\lsend xz}, \vYa1{\vy1}\lrecv yz
				}{
					\vlin{\yapopr}{}{
						\isya x \vdash \lsend xz, \lYa y \lrecv yz
					}{
						\vlin{\axrule}{}{
							\vdash \lsend xz , \lrecv xz
						}{\vlhy{}}
					}
				}
			}
		}
	}$}
	\end{equation}
	Its sub-derivations are associated to the links we observe in the coalescence path:
	\begin{itemize}
		\item the derivations above the premises of the $\lwith$-rule are the derivations associated to the links $\laprimo$ and $\lbprimo$;
		\item the derivation above the left premise of the left \resp{right} $\ltens$-rule is the derivations associated to the link $\la$ \resp{$\lb$}, and the ones above the right premises are the derivation(s) associated to the link(s) $\lcdue$;
		\item the derivations above the premise of the $\naloadr$-rule is the derivation associated to the link $\lcnu$, and the derivation above the premise of the $\yaloadr$-rule is the derivation associated to the link $\lcya$;
		\item the derivations above the premises of the $\nupopr$-rule and the $\yapopr$-rule are the derivations associated to the links $\lcnu$ and $\lcya$ respectively.
	\end{itemize}
\end{example}
%%%%%%%%%%%%%%%%%%%%%%%%%%%%%%%%%%%%%%%%%%%%%%%%%%%%%%%%%%%%%%%%

%%%%%%%%%%%%%%%%%%%%%%%%%%%%%%%%%%%%%%%%%%%%%%%%%%%%%%%%%%%%%%%%
%%%%%%%%%%%%%%%%%%%%%%%%%%%%%%%%%%%%%%%%%%%%%%%%%%%%%%%%%%%%%%%%
%%%%%%%%%%%%%%%%%%%%%%%%%%%%%%%%%%%%%%%%%%%%%%%%%%%%%%%%%%%%%%%%
\section{Slice nets}\label{sec:SN}
\def\flatto{\Rightarrow}
\def\miniflatto{\Rightarrow_{\mathsf{small}}}
\def\flattos{\Rightarrow^*}
\def\flatof#1{[#1]_{\mathsf{flat}}}
%%%%%%%%%%%%%%%%%%%%%%%%%%%%%%%%%%%%%%%%%%%%%%%%%%%%%%%%%%%%%%%%
%%%%%%%%%%%%%%%%%%%%%%%%%%%%%%%%%%%%%%%%%%%%%%%%%%%%%%%%%%%%%%%%
%%%%%%%%%%%%%%%%%%%%%%%%%%%%%%%%%%%%%%%%%%%%%%%%%%%%%%%%%%%%%%%%

In this section, we define \emph{\slnet} as special \conet whose underlying \cotree is `flat', i.e., it has at most a single $\conf$-node at its root.
While correctness of \slnets follows from the result on \conets, to prove the converse we introduce a procedure allowing us to transform any \conet into a \slnet.

%%%%%%%%%%%%%%%%%%%%%%%%%%%%%%%%%%%%%%%%%%%%%%%%%%%%%%%%%%%%%%%%
\begin{definition}
	A \defn{\slnet} for a sequent $\Gamma$ is a \conet $\tuple{\linktree,\dualizerof^{\linking}}$ (for $\Gamma$) such that the underlying \cotree $\linktree$ has at most a single $\conf$-node at its root.
\end{definition}
%%%%%%%%%%%%%%%%%%%%%%%%%%%%%%%%%%%%%%%%%%%%%%%%%%%%%%%%%%%%%%%%

The flattening procedure, seen on \cotrees, consists of moving $\conc$-nodes above $\conf$-nodes in a \cotree.
This is done by moving all children of a $\conc$-node except one of its $\conf$-child above this latter, by replacing each child of the selected $\conf$-node with a $\conc$-node having as children the selected $\conf$-child and all the children of $\conc$-node we want to move up.
Formally, this is done by applying the following rewriting rule on \cotrees.
%%%%%%%%%%%%%%%%%%%%%%%%%%%%%%%%%%%%%%%%%%%%%%%%%%%%%%%%%%%%%%%%
\begin{definition}
	A \defn{flattening step} on \cotrees is defined by the following rewriting rule on graphs
	\begin{equation}\label{eq:flattening}
		\begin{array}{cccccc}
			\vpz1{\tau_{1}}
			&\ldots&
			\vpz2{\tau_{n}}
			\\[10pt]
			&
			\vpz3{\conf} && \vpz4{\tau_1'} 
			& \ldots & 
			\vpz5{\tau_m'} 
		\\[10pt]
			&&&\vpz6{\conc}
		\\[10pt]
			&&&
			\vpz{r}{}
		\\[-15pt]
		\end{array}
		\multiGedges{pz3}{pz1,pz2}
		\multiGedges{pz6}{pz3,pz4,pz5}
		\dGedges{pz6/pzr}
	\qquad\flatto\qquad
		\canon{
			\begin{array}{ccccccccc} 
				\vpz1{{\tau_{1}}} 
				& 
				\vpz2{\tau_1'} 
			& \dots &
				\vpz3{\tau_m'}
				&&
				\vpz4{\tau_n}
				&
				\vpz5{\tau_1'}
			&\dots & 
				\vpz6{\tau_m'} 
			\\[10pt] 
				& \vpz7{\conc} &
			&& \ldots &&
				\vpz8{\conc}
			\\[10pt]
			&&&&
				\vpz9{\conf}
			\\[10pt]
			&&&&
				\vpz{r}{}
			\\[-15pt]
			\end{array}
			\multiGedges{pz7}{pz1,pz2,pz3}
			\multiGedges{pz8}{pz4,pz5,pz6}
			\multiGedges{pz9}{pz7,pz8}
			\dGedges{pz9/pzr}
		}
	\end{equation}
	where $\tau_1,\ldots,\tau_n,\tau'_1,\ldots,\tau'_m$ are \cotrees, and where the dotted edge represent, when it exists, the edge connecting root of the subtree on which we apply the flattening step to its parent.
	Flattening extends to \conets by applying the flattening step on the underlying \cotree, and preserving the dualizer for each of the copied link.

	We call \defn{flattening} the rewriting relation defined by the flattening steps.
\end{definition}
%%%%%%%%%%%%%%%%%%%%%%%%%%%%%%%%%%%%%%%%%%%%%%%%%%%%%%%%%%%%%%%%

We then prove that the normal forms of flattening are unique, and that flattening preserves coalescence.

%%%%%%%%%%%%%%%%%%%%%%%%%%%%%%%%%%%%%%%%%%%%%%%%%%%%%%%%%%%%%%%%
%FIG: Local confluence of flattening
%%%%%%%%%%%%%%%%%%%%%%%%%%%%%%%%%%%%%%%%%%%%%%%%%%%%%%%%%%%%%%%%
\begin{figure}[t]
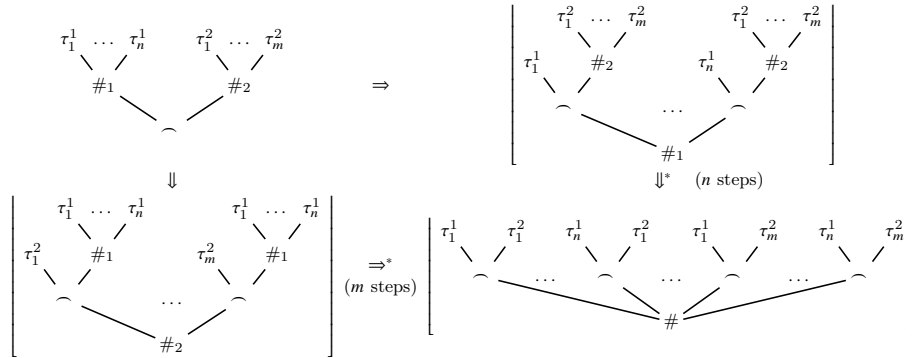

	$$\adjustbox{max width=\textwidth}{$
		\begin{array}{ccc}
			\begin{array}{ccccccc}
			\vpz1{\tau^1_{1}}
			&\ldots&
			\vpz2{\tau^1_{n}}
			&&
			\vpz3{\tau^2_{1}}
			&\ldots&
			\vpz4{\tau^2_{m}}
			\\[10pt]
			&
			\vpz5{\conf_1} &  
			&&&
			\vpz6{\conf_2} 
		\\[10pt]
			&&&\vpz7{\conc}
		\end{array}
		\Gedges{
			pz1/pz5,pz2/pz5,pz3/pz6,pz4/pz6,
			pz5/pz7,pz6/pz7%
		}
		&\flatto&
		\canon{
			\begin{array}{ccccccccc}
				&\vpz1{\tau^2_{1}} & \dots & \vpz2{\tau^2_{m}}
				&&&
				\vpz3{\tau^2_{1}} & \dots & \vpz4{\tau^2_{m}}
			\\[10pt]
				\vpz5{{\tau^1_{1}}} 
				&& \vpz6{\conf_2} &
				&&
				\vpz7{\tau^1_n}
				&&\vpz8{\conf_2} & 
			\\[10pt] 
				& \vpz9{\conc} &
				&& \ldots &&
				\vpz{10}{\conc}
			\\[10pt]
				&&&&
				\vpz{11}{\conf_1}
			\end{array}
			\Gedges{
				pz1/pz6,pz2/pz6,pz3/pz8,pz4/pz8,
				pz5/pz9,pz6/pz9,pz7/pz10,pz8/pz10,
				pz9/pz11,pz10/pz11%
			}
		}
		\\\\[-10pt]
		\Downarrow&& \qquad \qquad \Downarrow^* \quad \mbox{($n$ steps)}
		\\\\[-10pt]
		\canon{
			\begin{array}{ccccccccc}
				&\vpz1{\tau^1_{1}} & \dots & \vpz2{\tau^1_{n}}
				&&&
				\vpz3{\tau^1_{1}} & \dots & \vpz4{\tau^1_{n}}
			\\[10pt]
				\vpz5{{\tau^2_{1}}} 
				&& \vpz6{\conf_1} &
				&&
				\vpz7{\tau^2_m}
				&&\vpz8{\conf_1} & 
			\\[10pt] 
				& \vpz9{\conc} &
				&& \ldots &&
				\vpz{10}{\conc}
			\\[10pt]
				&&&&
				\vpz{11}{\conf_2}
			\end{array}
			\Gedges{
				pz1/pz6,pz2/pz6,pz3/pz8,pz4/pz8,
				pz5/pz9,pz6/pz9,pz7/pz10,pz8/pz10,
				pz9/pz11,pz10/pz11%
			}
		}
		&
		\begin{array}{c}
			\flattos \\ \mbox{($m$ steps)}
			\end{array}
		&
		\canon{
			\begin{array}{ccccccccccccccc}
				\vpz1{\tau^1_{1}} && \vpz2{\tau^2_{1}}
				&&
				\vpz3{\tau^1_{n}} && \vpz4{\tau^2_{1}}
				&&
				\vpz5{\tau^1_{1}} && \vpz6{\tau^2_{m}}
				&&
				\vpz7{\tau^1_{n}} && \vpz8{\tau^2_{m}}
			\\[10pt]
				&\vpz9{\conc}&&\ldots&&\vpz{10}{\conc}
				&&\ldots&&\vpz{11}{\conc}&&\ldots&&\vpz{12}{\conc}
			\\[10pt]
				&&&&&&&\vpz{13}{\conf}
			\end{array}
			\Gedges{
				pz1/pz9,pz2/pz9,pz3/pz10,pz4/pz10,pz5/pz11,pz6/pz11,pz7/pz12,pz8/pz12,
				pz9/pz13,pz10/pz13,pz11/pz13,pz12/pz13%
			}
		}
		\end{array}
		$
	}$$
	\caption{Convergence of the non-trivial critical pair for flattening.}
	\label{fig:lconfluence}
\end{figure}
%%%%%%%%%%%%%%%%%%%%%%%%%%%%%%%%%%%%%%%%%%%%%%%%%%%%%%%%%%%%%%%%
%%%%%%%%%%%%%%%%%%%%%%%%%%%%%%%%%%%%%%%%%%%%%%%%%%%%%%%%%%%%%%%%

%%%%%%%%%%%%%%%%%%%%%%%%%%%%%%%%%%%%%%%%%%%%%%%%%%%%%%%%%%%%%%%%
\begin{restatable}{lemma}{lemFlatteningNorm}\label{lem:flat:norm}
	Flattening is strongly normalizing and has a unique normal form.
\end{restatable}
\begin{proof}
	Flattening admits a unique form of non-trivial critical pairs, which is confluent (see \Cref{fig:lconfluence}).

	To prove termination, we define a measure $\mu$ on \cotrees defined as follows: 
	if the root of $\tau$ is a leaf, then $\mu(\tau)=2$; 
	if $\tau = \conf(\tau_1, \tau_2)$, then $\mu(\tau)=\mu(\tau_1)+\mu(\tau_2)+1$; 
	and if $\tau = \conc(\tau_1, \tau_2)$, then $\mu(\tau)=\mu(\tau_1)\cdot \mu(\tau_2)$.
	Each flattening step strictly decreases the measure of the \cotree: 
	let $\tau$ be the sub\cotree rooted at the $\conf$ on the left-hand side of \Cref{eq:flattening}, 
	and let $\tau'$ be the sub\cotree on the right-hand side. 
	By construction $m>0$, thus
	$\mu(\tau)= (\sum_{i = 1}^{n}\mu(\tau_i ) +1) \cdot \prod_{i=1}^{m}\mu(\tau'_i)
	> \sum_{i = 1}^{n}\mu(\tau_i ) \cdot \prod_{i=1}^{m}\mu(\tau'_i) +1 \geq \mu(\tau')$.
	The last inequality follows from the fact that the operation of canonization ($\canon{\cdot}$) further reduces the measure.

	Therefore, flattening is terminating, and by Newman's Lemma, it is confluent, yielding the uniqueness of normal forms.
\end{proof}
%%%%%%%%%%%%%%%%%%%%%%%%%%%%%%%%%%%%%%%%%%%%%%%%%%%%%%%%%%%%%%%%

%%%%%%%%%%%%%%%%%%%%%%%%%%%%%%%%%%%%%%%%%%%%%%%%%%%%%%%%%%%%%%%%
\begin{restatable}{theorem}{thmFlat}\label{thm:flat}
	Let $\pn$ be a \conet.
	If $\pn'$ be is the normal form of $\pn$ with respect to flattening, then $\pn'$ is a \slnet.
\end{restatable}
\begin{proof}
	Flattening preserves the fact that the \cotree is axiomatic, and it is defined to preserve canonicity.
	Moreover, the validity of the witness map is preserved by the fact that flattening may `duplicate' links and their dualizer, but it does not change the dualizer of any link.
	To conclude, we have to show that flattening preserves coalescence.

	Thanks to \Cref{lem:flat:norm}, we can assume that the $\conf$-node involved in the flattening step is the root of $\linktree$, because we can always assume the coalescence path first reduces the \cotree rooted in the $\conf$-node involved in the flattening step, and that no coalescence steps can be applied to $\tau_1,\ldots,\tau_n,\tau'_1,\ldots,\tau'_m$, in particular, that they are \cotrees made of a single link $\la_1,\ldots,\la_n,\lb_1,\ldots,\lb_m$ respectively.
	
	If $n=2$, then we know we can apply a $\lwith$-step to $\la_1$ and $\la_2$, obtaining a link $\la$, and then, after a $\bullet$-step, we know we can apply some coalescence steps from the ones in the top part of \Cref{fig:coalescence} on $\la$ and $\lb_1,\ldots,\lb_n$ until the \cotree coalesces on a single link $\laprimo$.
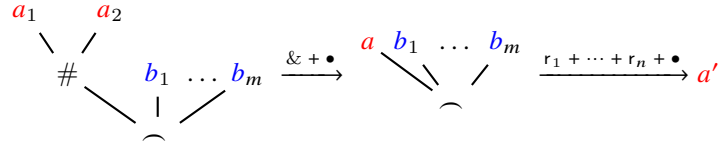
	$$\begin{array}{ccccl}
		\begin{array}{cccccc}
			\vpz1{{\la_1}}
			&&
			\vpz2{{\la_2}}
			\\[10pt]
			&\vpz3{\conf} &
			& 
			\vpz4{{\lb_1}} & \ldots & \vpz5{{\lb_m}} \\[10pt]
			&&&\vpz6{\conc}
		\end{array}
		\Gedges{pz1/pz3,pz2/pz3,pz3/pz6,pz4/pz6,pz5/pz6}
	&\coalto[\lwith \;+\; \bullet]&
		\begin{array}{cccc}
			\vpz1{\la} & \vpz2{\lb_1} & \ldots & \vpz3{\lb_m}
			\\[10pt]
			&&\vpz4{\conc}
		\end{array}
		\Gedges{pz1/pz4,pz2/pz4,pz3/pz4}
	&\coalto[\rrule_1  \;+\;  \cdots  \;+\;  \rrule_n  \;+\;  \bullet]&
		\laprimo
	\end{array}$$
	Therefore, we know that the same sequence of coalescence steps can be applied to $\la_1,\lb_1,\ldots,\lb_m$ and to $\la_2,\lb_1,\ldots,\lb_m$ 
	reaching links $\laprimouno$ and $\laprimodue$ respectively, since both $\la_1$ and $\la_2$ are sub-sequents of $\la$.
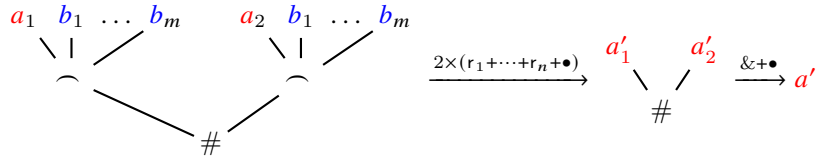
	$$\begin{array}{ccccl}
		\begin{array}{ccccccccc} 
			\vpz1{{{\la_1}}} & \vpz2{{\lb_1}} & \dots &
			\vpz3{{\lb_m}}
			&&
			\vpz4{{\la_2}}
			&
			\vpz5{{\lb_1}}
			&
			\dots & \vpz6{{\lb_m}} 
		\\[10pt] 
			& \vpz7{\conc} &
			&& &&
			\vpz8{\conc}
		\\[10pt]
			&&
			&&
			\vpz9{\conf}
		\end{array}
		\multiGedges{pz9}{pz7,pz8}
		\multiGedges{pz7}{pz1,pz2,pz3}
		\multiGedges{pz8}{pz4,pz5,pz6}
	&\coalto[2\times({\rrule_1+\cdots+\rrule_n}+\bullet)]&
		\begin{array}{ccc} 
			\vpz1{{{\laprimouno}}} 
			&&
			\vpz2{{\laprimodue}}
		\\[10pt] 
			&
			\vpz9{\conf}
		\end{array}
		\multiGedges{pz9}{pz1,pz2}
	&\coalto[\lwith+\bullet]&
		\laprimo
	\end{array}$$
	For $n>2$, we apply a similar reasoning after remarking that in this case the coalescence has to start with a $\conf$-step as follows
	
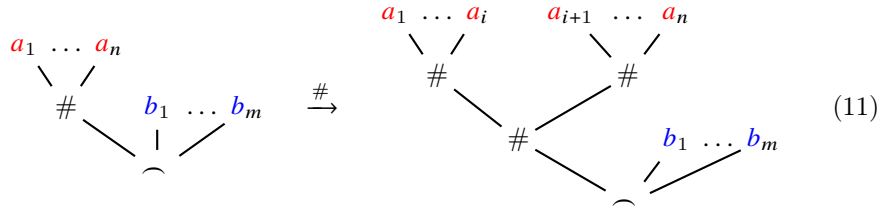
\begin{equation}\label{eq:bho}
		\begin{array}{cccccc}
			\vpz1{\la_1}
			&\ldots&
			\vpz2{\la_n}
			\\[10pt]
			&
			\vpz3{\conf}
			&&
			\vpz4{\lb_1}
			&\ldots&
			\vpz5{\lb_m}
			\\[10pt]
			&&&\vpz6{\conc}
		\end{array}
		\Gedges{pz1/pz3,pz2/pz3,pz3/pz6,pz4/pz6,pz5/pz6}
	\quad\coalto[\conf]\quad
		\begin{array}{ccccccccc}
			\vpz1{\la_1}
			&\ldots &\vpz7{\la_i} && \vpz8{\la_{i+1}} &\ldots&
			\vpz2{\la_n}
			\\[10pt]
			&
			\vpz9{\conf}
			&&&&
			\vpz{10}{\conf}
			&
			\\[10pt]
			&&&
			\vpz3{\conf}
			&&&
			\vpz4{\lb_1}
			&
			\ldots
			&
			\vpz5{\lb_m}
			\\[10pt]
			&&&&&\vpz6{\conc}
		\end{array}
		\Gedges{pz1/pz9,pz7/pz9,pz8/pz10,pz2/pz10,pz9/pz3,pz10/pz3,pz3/pz6,pz4/pz6,pz5/pz6}
	\end{equation}
	and then conclude by \Cref{lem:flat:norm} since one flattening step to the \cotree on the left-hand side of \Cref{eq:bho} gives a \cotree which can be obtained by applying two flattening steps on the \cotree on the right-hand side of \Cref{eq:bho} as shown below.
	$$\adjustbox{max width=\textwidth}{$
	\begin{array}{cccccccccccc}
		\vpz0{\la_1}&\ldots&\vpz1{\la_i}&&&&\vpz2{\la_{i+1}}&\ldots&\vpz3{\la_n}
		\\[10pt]
		&\vpz4{\conf}&&\vpz5{\lb_1}&\ldots&\vpz6{\lb_m}&&\vpz7{\conf}&&\vpz8{\lb_{1}}&\ldots&\vpz9{\lb_m}
		\\[10pt]
		&&\vpz{10}{\conc}&&&&&&\vpz{11}{\conc}
		\\[10pt]
		&&&&&\vpz{12}{\conf}
	\end{array}
	\Gedges{
		pz0/pz4,pz1/pz4,pz2/pz7,pz3/pz7,
		pz4/pz10,pz5/pz10,pz6/pz10,pz7/pz11,pz8/pz11,pz9/pz11,
		pz10/pz12,pz11/pz12%
	}
	\quad\flatto\flatto\quad
	\begin{array}{ccccccccc}
		\vpz1{\la_1}
		&
		\vpz2{\lb_1}
		&
		\dots
		&
		\vpz3{\lb_m}
		&&
		\vpz4{\la_n}
		&
		\vpz5{\lb_1}
		&
		\dots
		&
		\vpz6{\lb_m}
		\\[10pt]
		&
		\vpz7{\conc}
		&&&
		\ldots
		&&
		\vpz8{\conc}
		\\[10pt]
		&&
		&&
		\vpz9{\conf}
	\end{array}
	\Gedges{
		pz1/pz7,pz2/pz7,pz3/pz7,pz4/pz8,pz5/pz8,pz6/pz8,
		pz7/pz9,pz8/pz9%
	}
	$}$$
\end{proof}
%%%%%%%%%%%%%%%%%%%%%%%%%%%%%%%%%%%%%%%%%%%%%%%%%%%%%%%%%%%%%%%%

We can now conclude the proof of the main result of this section, which states that \slnets are sound and complete for $\PIL$.

%%%%%%%%%%%%%%%%%%%%%%%%%%%%%%%%%%%%%%%%%%%%%%%%%%%%%%%%%%%%%%%%
\begin{restatable}{theorem}{Slice}\label{thm:slice}
	Let $\Gamma$ be a sequent.
	Then $\proves[\PIL]\Gamma$ iff there is a \slnet for $\Gamma$.
\end{restatable}
\begin{proof}
	It follows from the result on \conets \Cref{thm:CNsoundAndComplete} and the fact that flattening preserves coalescence (\Cref{thm:flat}).
\end{proof}
%%%%%%%%%%%%%%%%%%%%%%%%%%%%%%%%%%%%%%%%%%%%%%%%%%%%%%%%%%%%%%%%

%%%%%%%%%%%%%%%%%%%%%%%%%%%%%%%%%%%%%%%%%%%%%%%%%%%%%%%%%%%%%%%%
%%%%%%%%%%%%%%%%%%%%%%%%%%%%%%%%%%%%%%%%%%%%%%%%%%%%%%%%%%%%%%%%
%%%%%%%%%%%%%%%%%%%%%%%%%%%%%%%%%%%%%%%%%%%%%%%%%%%%%%%%%%%%%%%%
\section{Canonicity}\label{sec:canon}
%%%%%%%%%%%%%%%%%%%%%%%%%%%%%%%%%%%%%%%%%%%%%%%%%%%%%%%%%%%%%%%%
%%%%%%%%%%%%%%%%%%%%%%%%%%%%%%%%%%%%%%%%%%%%%%%%%%%%%%%%%%%%%%%%
%%%%%%%%%%%%%%%%%%%%%%%%%%%%%%%%%%%%%%%%%%%%%%%%%%%%%%%%%%%%%%%%

%%%%%%%%%%%%%%%%%%%%%%%%%%%%%%%%%%%%%%%%%%%%%%%%%%%%%%%%%%%%%%%%
% Fig: rule permutations
%%%%%%%%%%%%%%%%%%%%%%%%%%%%%%%%%%%%%%%%%%%%%%%%%%%%%%%%%%%%%%%%
\begin{figure*}[t]
	\centering
	\adjustbox{max width=\textwidth}{$\begin{array}{c}
		\mbox{Local rule permutations}
	\\[20pt]
		\vlderivation{
			\vliin{\rclr{\brrule[1]}}{}{ \sdash[\sS_1, \sS_2,\sS_3] \Gamma_1,\Gamma_2,\Gamma_3,\rclr{\Theta_1}, {\Theta_2}}{
				\vlhy{\sdash[\sS_1]\Gamma_1, \rclr{\Delta_1}}
			}{
				\vliin{\bclr{\brrule[2]}}{}{
					\sdash[\sS_2,\sS_3] \Gamma_2,\Gamma_3,\rclr{\Delta_2},\bclr{\Theta_2}
				}{
					\vlhy{ \sdash[\sS_2] \Gamma_2,{\Delta_2},\bclr{\Delta_3}}}{\vlhy{ \sdash[\sS_3] \Gamma_3, \bclr{\Delta_4}}
				}
			}
		}
		\peq
		\vlderivation{
			\vliin{\bclr{\brrule[2]}}{}{
				\sdash[\sS_1, \sS_2,\sS_3]  \Gamma_1, \Gamma_2, \Gamma_3,{\Theta_1},\bclr{\Theta_2}
			}{
				\vliin{\rclr{\brrule[1]}}{}{
					\sdash[\sS_1,\sS_2] \Gamma_1,\Gamma_2, \rclr{\Theta_1}, \bclr{\Delta_2}
				}{
					\vlhy{\sdash[\sS_1] \Gamma_1,  \rclr{\Delta_1}}
				}{
					\vlhy{\sdash[\sS_3] \Gamma_2, \rclr{\Delta_2},{\Delta_3}}
				}
			}{\vlhy{\sdash[\sS_3] \Gamma_3, \bclr{\Delta_4}}}
		}
	\\[20pt]
		\vlderivation{
			\vlin{\bclr{\urrule[2]}}{}{
				\sdash[\sS_1,\sS_2]\Gamma, {C_1},\bclr{C_2}}{
				\vlin{\rclr{\urrule[1]}}{}{
					\sdash[\sS_1,\sS_2']\Gamma, \rclr{C_1}, \bclr{\Delta_2}
				}{\vlhy{\sdash[\sS_1',\sS_2']\Gamma, \rclr{\Delta_1},{\Delta_2}}}}
		}
		\peq
		\vlderivation{
			\vlin{\rclr{\urrule[1]}}{}{
				\sdash[\sS_1,\sS_2]\Gamma, \rclr{C_1}, {C_2}
			}{
				\vlin{\bclr{\urrule[2]}}{}{
					\sdash[\sS'_1,\sS_2]\Gamma,  \rclr{\Delta_1},\bclr{C_2}
				}{
					\vlhy{\sdash[\sS'_1,\sS'_2]\Gamma, {\Delta_1},\bclr{\Delta_2}}
				}
			}
		}
	\qquad
		\vlderivation{
			\vlin{\rclr{\urrule[1]}}{}{\sdash[\sS_1, \sS_2]\Gamma_1, \Gamma_2,\rclr{C_1},\bclr{\Theta_2}}{
				\vliin{\bclr{\brrule}}{}{\sdash[\sS_1', \sS_2] \Gamma_1, \Gamma_2,\rclr{\Delta_1},\bclr{\Theta_2}
				}{
					\vlhy{ \sdash[\sS_1']\Gamma_1, {\Delta_1},\bclr{\Delta_2}}
				}{
					\vlhy{\sdash[\sS_2] \Gamma_2, \bclr{\Delta_3}}
				}
			}
		}
		\peq
		\vlderivation{
			\vliin{\bclr{\brrule}}{}{
				\sdash[\sS_1, \sS_2] \Gamma_1, \Gamma_2,{C_1},\bclr{\Theta_2}
			}{
				\vlin{\rclr{\urrule[1]}}{}{
					\sdash[\sS_1]\Gamma,\rclr{C_1}, \bclr{\Delta_2}
				}{
					\vlhy{\sdash[\sS_1']\Gamma_1, \rclr{\Delta_1}, {\Delta_2}}
				}
			}{
				\vlhy{\sdash[\sS_2] \Gamma_2, \bclr{\Delta_3}}
			}
		}
	\\[20pt]
		\vlderivation{
			\vliin{\rclr{\lwith}}{}{
				\sdash \Gamma,\rclr{A\lwith B}, C\lwith D
			}{
				\vliin{\bclr{\lwith}}{}{\sdash \Gamma,\rclr{A},\bclr{C\lwith D}}{
					\vlhy{\sdash \Gamma,A,\bclr{C}}
				}{
					\vlhy{\sdash \Gamma,A,\bclr{D}}
				}
			}{
				\vliin{\bclr{\lwith}}{}{\sdash \Gamma,\rclr{B},\bclr{C\lwith D}}{
					\vlhy{\sdash \Gamma,B,\bclr{C}}
				}{
					\vlhy{\sdash \Gamma,B,\bclr{D}}
				}
			}
		}
		\peq
		\vlderivation{
			\vliin{\bclr{\with}}{}{
				\sdash \Gamma, A\lwith B, \bclr{C\lwith D}
			}{
				\vliin{\rclr{\lwith}}{}{
					\sdash \Gamma, \rclr{A\lwith B}, \bclr{C},
				}{
					\vlhy{\sdash \Gamma, \rclr{B},C}
				}{
					\vlhy{\sdash \Gamma, \rclr{A},C}
				}
			}{
				\vliin{\rclr{\lwith}}{}{
					\sdash \Gamma, \rclr{A\lwith B}, \bclr{D},
				}{
					\vlhy{\sdash \Gamma, \rclr{B},D}
				}{
					\vlhy{\sdash \Gamma, \rclr{A},D}
				}
			}
		}
	\\[20pt]
		\vlderivation{
			\vlin{\bclr{\alpha}}{}{
				\sdash[\sS] \Gamma,A\lwith B,\bclr{C}
			}{
				\vliin{\rclr{\lwith}}{}{
					\sdash[\sS'] \Gamma, \rclr{A\lwith B}, \bclr{\Delta}
				}{
					\vlhy{\sdash[\sS'] \Gamma, \rclr{B},\Delta}
				}{
					\vlhy{\sdash[\sS'] \Gamma, \rclr{A},\Delta}
				}
			}
		}
		\peq
		\vlderivation{
			\vliin{\rclr{\lwith}}{}{
				\sdash \Gamma,\rclr{A\lwith B}, C
			}{
				\vlin{\bclr{\urrule}}{}{
					\sdash \Gamma,\rclr{A},\bclr{C}
				}{
					\vlhy{\sdash[\sS'] \Gamma,A,\bclr{\Delta}}
				}
			}{
				\vlin{\bclr{\urrule}}{}{
					\sdash \Gamma,\rclr{B},\bclr{C}
				}{
					\vlhy{\sdash[\sS'] \Gamma,B,\bclr{\Delta}}
				}
			}
		}
	\\[20pt]\hline\\[-5pt]
		\mbox{Non-local rule permutations}
	\\
		\vlderivation{
			\vliin{\rclr{\brrule}}{}{
				\sdash[\sS_1,\sS_2] \rclr{\Gamma}, C\lwith D
			}{
				\vlpr{\dD}{}{\sdash[\sS_1] \rclr{\Gamma_1}}
			}{
				\vliin{\bclr{\lwith}}{}{
					\sdash[\sS_2] \rclr{\Gamma_2},\bclr{C\lwith D}
				}{
					\vlhy{\sdash[\sS_2] \Gamma_2,\bclr{C}}
				}{
					\vlhy{\sdash[\sS_2] \Gamma_2,\bclr{D}}
				}
			}
		}
		\quad\speq\quad
		\vlderivation{
			\vliin{\bclr{\with}}{}{
				\sdash[\sS_1,\sS_2] \Gamma,\bclr{C\lwith D}
			}{
				\vliin{\rclr{\brrule}}{}{
					\sdash[\sS_1,\sS_2] \rclr{\Gamma}, \bclr{C},
				}{
					\vlpr{\dD}{}{\sdash[\sS_1] \rclr{\Gamma_1}}
				}{
					\vlhy{\sdash[\sS_2] \rclr{\Gamma_2}, C}
				}
			}{
				\vliin{\rclr{\brrule}}{}{
					\sdash[\sS_1,\sS_2] \rclr{\Gamma}, \bclr{D},
				}{
					\vlpr{\dD}{}{\sdash[\sS_1] \rclr{\Gamma_1}}
				}{
					\vlhy{\sdash[\sS_2] \rclr{\Gamma_2}, D}
				}
			}
		}
	\end{array}$}
	\caption{
		Rule permutations 
		in $\PIL$
        with
		$\brrule,\brrule[1],\brrule[2]\in\set{\ltens,\lprec}$
		and with
		$\urrule,\urrule[1],\urrule[2]\in\set{\lpar,\lplus,\exists,\forall,\nuurule,\yurule,\nupopr,\yapopr,\nuloadr,\yaloadr}$.
	}
	\label{fig:permutations1}
\end{figure*}
%%%%%%%%%%%%%%%%%%%%%%%%%%%%%%%%%%%%%%%%%%%%%%%%%%%%%%%%%%%%%%%%
%%%%%%%%%%%%%%%%%%%%%%%%%%%%%%%%%%%%%%%%%%%%%%%%%%%%%%%%%%%%%%%%

In this section we discuss the proof equivalence enforced by conflict nets and slice nets, called \defn{local} and \defn{strong} canonicity, respectively.

For this, we define two equivalence relations over derivations, and a notion of isomorphism for \conets (and thus for \slnets) which we use to state the canonicity results.
%%%%%%%%%%%%%%%%%%%%%%%%%%%%%%%%%%%%%%%%%%%%%%%%%%%%%%%%%%%%%%%%
\begin{definition}
	We define the two following equivalence relations over derivations in $\PIL$ using the rule permutations in \Cref{fig:permutations1}:
	\begin{itemize}
		\item \defn{local equivalence} (denoted $\dD_1\peq \dD_2$) if it is possible to transform $\dD_1$ into $\dD_2$ using a sequence of local rule permutations;
		\item \defn{strong equivalence} (denoted $\dD_1\speq \dD_2$) if it is possible to transform $\dD_1$ into $\dD_2$ using a sequence of both local and non-local rule permutations.
	\end{itemize}
\end{definition}
%%%%%%%%%%%%%%%%%%%%%%%%%%%%%%%%%%%%%%%%%%%%%%%%%%%%%%%%%%%%%%%%

%%%%%%%%%%%%%%%%%%%%%%%%%%%%%%%%%%%%%%%%%%%%%%%%%%%%%%%%%%%%%%%%
\begin{definition}
	Two conflict nets $\tuple{\linking_1,\dualizerof[1]}$ and $\tuple{\linking_2,\dualizerof[2]}$
	are \defn{isomorphic} (denoted $\tuple{\linking_1,\dualizerof[1]} = \tuple{\linking_2,\dualizerof[2]}$) 
	if there is a label-preserving isomorphism between $\linking_1$ and $\linking_2$  and such that $\dualizerof[1]=\dualizerof[2]$.
\end{definition}
%%%%%%%%%%%%%%%%%%%%%%%%%%%%%%%%%%%%%%%%%%%%%%%%%%%%%%%%%%%%%%%%

%%%%%%%%%%%%%%%%%%%%%%%%%%%%%%%%%%%%%%%%%%%%%%%%%%%%%%%%%%%%%%%%
\begin{restatable}{lemma}{confCoal}\label{lem:confCoal}
	Let $\pn=\tuple{\linktree,\dualizerof^\linking}$ be a \conet for $\Gamma$ and $\dD_{\pn}$ be a derivation obtained by sequentializing it.
	If $\pn \coaltos \pn'$,
	then $\pn'$ can be sequentialized into a derivation $\dD_{\pn'}$ such that $\dD_{\pn'}\peq \dD_{\pn}$.
\end{restatable}
\begin{proof}
	It is sufficient to show that all critical pairs of coalescence steps converge and that the resulting derivations are equivalent modulo the local rule permutations shown in \Cref{fig:permutations1}.
	The most convoluted case is the one for the pair $\conf / \conf$, which can be treated in the exact same way as in \cite{hei:hug:str:ALL1} after remarking that 
	the additional information provided by the witness map does not interfere with the proof because the join operator on dualizers is associative and commutative.
	The cases which are not present in \cite{hei:hug:conflict} are discussed in \Cref{app:proofs}.
\end{proof}
%%%%%%%%%%%%%%%%%%%%%%%%%%%%%%%%%%%%%%%%%%%%%%%%%%%%%%%%%%%%%%%%

%%%%%%%%%%%%%%%%%%%%%%%%%%%%%%%%%%%%%%%%%%%%%%%%%%%%%%%%%%%%%%%%
\begin{lemma}\label{lem:flattoId}
	Let $\pn$ and $\pn'$ be two \conets for $\Gamma$ and let $\dD_{\pn}$ and $\dD_{\pn'}$ be two derivations obtained by sequentializing them.
	If $\pn \flatto\pn'$, then $\dD_{\pn'}\speq \dD_{\pn}$.
\end{lemma}
\begin{proof}
	For the same reasons of the ones used in the proof of \Cref{thm:flat}, we only discuss the case in which the flattening step is applied to the root of the \cotree, that the $\conf$-node affected by the step has only two children, and that no coalescence step can be applied in the subtree affected by flattening.

	To simplify the reasoning, use the following \emph{small-step} reduction
	\begin{equation}\label{eq:miniflat}
		\begin{array}{ccccccc}
			\vpz1{\la_{1}}
		&&
			\vpz2{\la_{2}}
		\\[10pt]
			&
			\vpz3{\conf} 
			&& 
			\vpz4{\lb} 
			& 
			\vpz5{\lb_1}
			\cdots
			\vpz6{\lb_m}
		\\[10pt]
			&&&\vpz7{\conc}
		\end{array}
		\multiGedges{pz3}{pz1,pz2}
		\multiGedges{pz7}{pz3,pz4,pz5,pz6}
	\miniflatto
		\canon{
			\begin{array}{ccccccc}
				\vpz{10}{\la_{1}}
				\qquad
				\vpz{14}{\lb}
				&&
				\vpz{11}{\la_{2}}
				\qquad
				\vpz{15}{\lb} 
			\\[10pt]
				\vpz1{\conc}
				&&
				\vpz2{\conc}
			\\
				&
				\vpz3{\conf} && 
				\vpz5{\lb_1}\ldots\vpz6{\lb_m}
			\\[10pt]
				&&\vpz7{\conc}
			\end{array}
			\multiGedges{pz3}{pz1,pz2}
			\multiGedges{pz7}{pz3,pz5,pz6}
			\multiGedges{pz1}{pz14,pz10}
			\multiGedges{pz2}{pz15,pz11}
		}
	\end{equation}
	decomposing the flattening step into a sequence of steps, and we show that each small-step corresponds to an application of a non-local permutation in \Cref{fig:permutations1}.
	
	If the \conet obtained by applying the small step is still coalescent (we can always assume this is the case, since we can force the small step to be applied only if the resulting \conet is coalescent), then the coalescence path of the left-hand side starts with a $\lwith$-step, followed by a $\bullet$ and a $\rrule$-step involving $\lb$.
	Similarly, we can start the coalescence path of the right-hand side with two $\rrule$-steps, one for $\la_1$ and one for $\la_2$, followed by a $\lwith$-step and a $\bullet$-step.
	The resulting derivations are the same up to the non-local rule permutation in \Cref{fig:permutations1} for the $\lwith$ and $\rrule$ rules.
\end{proof}
%%%%%%%%%%%%%%%%%%%%%%%%%%%%%%%%%%%%%%%%%%%%%%%%%%%%%%%%%%%%%%%%

%%%%%%%%%%%%%%%%%%%%%%%%%%%%%%%%%%%%%%%%%%%%%%%%%%%%%%%%%%%%%%%%
\begin{theorem}\label{thm:PNareCanon}
	Let $\dD$ and $\dD'$ be two derivations in $\PIL$.
	Then,
	\begin{enumerate}
		\item\label{pweq} $\dD\peq\dD'$ iff $\cnof{\dD}=\cnof{\dD'}$ .
		\item\label{spweq} $\dD\speq\dD'$ iff $\snof{\dD}=\snof{\dD'}$ .
	\end{enumerate}
	where we denote by $\snof{\cdot}$ the function mapping a derivation to the \slnet obtained by flattening the conflict net associated to it via $\cnof{\cdot}$ defined in \Cref{fig:deseq}.
\end{theorem}
\begin{proof}
	To prove the left-to-right implication of both points, it suffices to show that if $\dD$ and $\dD'$ are two derivations with the bottom-most rules in a configuration as in \Cref{fig:permutations1}, then $\cnof{\dD}=\cnof{\dD'}$ and $\snof{\dD}=\snof{\dD'}$, respectively.
	For the rules in $\MALL$, the proof is the one in \cite{hei:hug:conflict} and \cite{hug:van:slice}, respectively, and the additional cases for the rules of $\PIL$ are straightforward.

	The right-to-left implication of \Cref{pweq} follows from \Cref{lem:confCoal} since each non-trivial critical pair of coalescence steps corresponds to a local rule permutation in \Cref{fig:permutations1}.
	The right-to-left implication of \Cref{spweq} is more involved. 
	It relies on \Cref{thm:flat} (flattening preserves coalescence), and 
	\Cref{lem:flattoId} (flattening steps correspond to non-local rule permutations in \Cref{fig:permutations1} plus possibly some local rule permutations).
	When we can conclude from \Cref{pweq}.
\end{proof}
%%%%%%%%%%%%%%%%%%%%%%%%%%%%%%%%%%%%%%%%%%%%%%%%%%%%%%%%%%%%%%%%

%%%%%%%%%%%%%%%%%%%%%%%%%%%%%%%%%%%%%%%%%%%%%%%%%%%%%%%%%%%%%%%%
%%%%%%%%%%%%%%%%%%%%%%%%%%%%%%%%%%%%%%%%%%%%%%%%%%%%%%%%%%%%%%%%
%%%%%%%%%%%%%%%%%%%%%%%%%%%%%%%%%%%%%%%%%%%%%%%%%%%%%%%%%%%%%%%%
\section{Conclusion and Perspectives}\label{sec:conc}
%%%%%%%%%%%%%%%%%%%%%%%%%%%%%%%%%%%%%%%%%%%%%%%%%%%%%%%%%%%%%%%%
%%%%%%%%%%%%%%%%%%%%%%%%%%%%%%%%%%%%%%%%%%%%%%%%%%%%%%%%%%%%%%%%
%%%%%%%%%%%%%%%%%%%%%%%%%%%%%%%%%%%%%%%%%%%%%%%%%%%%%%%%%%%%%%%%

In this paper we defined \emph{conflict nets} and \emph{slice nets} for $\PIL$, proved their soundness and completeness, provided sequentialization procedures and proof translations, and studied the proof equivalences they capture.

The results presented here, in conjunction with the ones in \cite{acc:man:mon:FaP}, open to the possibility of using proof nets as a basis for the study of the semantics of concurrent systems.
For example, a direct correspondence between choreographies and proof nets for $\PIL$ can be easily obtained by combining the proof translation from this paper and the results in \cite{acc:man:mon:FaP}.

Moreover, the canonicity results we proved suggest that proof nets can provide a canonical representatives of execution trees for concurrent programs, since they allow to identify execution trees while ignoring irrelevant differences (such as permutations of independent actions), while still distinguishing execution trees that differ due to externally observable phenomena (such as race conditions or side effects).
The distinction between local and strong canonicity also suggests that different notions of proof equivalence may be relevant for the definition of different notions of equivalences for concurrent systems, depending on the specific requirements of the application at hand.
For instance, local canonicity may be more suitable for applications where the order of branching during an exectution is relevant, while strong canonicity may be more suitable for applications where the order of branching is irrelevant, such as in the case of bisimulation equivalences for concurrent systems.

In future work, we may study cut-elimination for our proof nets, where the mechanics of nominal quantifiers is far for being trivial, and, already in the sequent calculus, it requires additional rules \cite{acc:man:mon:FaPext}.
We expect that proof nets for first-order multiplicative additive linear logic can be obtained by a straightforward extension of our definitions, allowing for general terms and predicate symbols (and removing unit, the $\lprec$, and nominal quantifiers, as well as the corresponding coalescence steps).\footnote{
	In this setting, an axiomatic linking $\linking$ on $\Gamma$ consists of links of the form $\set{P(t_1,\dots,t_n), \cneg P(s_1,\dots,s_n)}$, and the dualizer for each such link is a substitution (with the same constarins on their domain as in \Cref{def:dualizer}) unifying $P(t_1,\dots,t_n)$ and $P(s_1,\dots,s_n)$.
}
Such proof nets could be used to extend the result on interpolation for linear logic \cite{saurin:Interpolation,PN:interpolation}, relying on proof nets with contractability correctness criteria similar to the one we used for conflict nets.

Concerning the nominal quantifiers, proof nets underlay again their `linear' nature in handling fresh names, which suggest them as a possible basis for the study of `multiplicative' quantifiers (see, e.g., \cite{nic:pia:tes:multQuant})
Finally, an important open question is the characterization of the proof equivalence induced by abstracting from the witnesses of nominal quantifiers, thus defining a stronger notion of proof equivalence reflecting the agnostic nature of nominal quantifiers with respect to the choice of fresh names, as well as of all quantifiers as in \cite{hug:unification,hei:hug:str:ALL1} 
(see \Cref{sec:witnesses} for a preliminary discussion on this topic).

%%%%%%%%%%%%%%%%%%%%%%%%%%%%%%%%%%%%%%%%%%%%%%%%%%%%%%%%%%%%%%%%
%%%%%%%%%%%%%%%%%%%%%%%%%%%%%%%%%%%%%%%%%%%%%%%%%%%%%%%%%%%%%%%%
\bibliographystyle{splncs04}
\bibliography{biblio}

@misc{acc:man:mon:FaPext,
	title={Formulas as Processes, Deadlock-Freedom as Choreographies (Extended Version)}, 
	author={Matteo Acclavio and Giulia Manara and Fabrizio Montesi},
	year={2025},
	eprint={2501.08928},
	archivePrefix={arXiv},
	primaryClass={cs.LO},
	url={https://arxiv.org/abs/2501.08928}, 
}

@InProceedings{acc:man:mon:FaP,
	author="Acclavio, Matteo
	and Manara, Giulia
	and Montesi, Fabrizio",
	editor="Vafeiadis, Viktor",
	title="Formulas as Processes, Deadlock-Freedom as Choreographies",
	booktitle="Programming Languages and Systems",
	year="2025",
	publisher="Springer Nature Switzerland",
	address="Cham",
	pages="23--55",
	isbn="978-3-031-91118-7"
}

@article{nic:pia:tes:multQuant, 
	title={Non-contractive Logics, Paradoxes, and Multiplicative Quantifiers}, 
	volume={17}, 
	DOI={10.1017/S1755020323000138}, 
	number={4}, 
	journal={The Review of Symbolic Logic}, 
	author={Nicolai, Carlo and Piazza, Mario and Tesi, Matteo}, 
	year={2024}, 
	pages={996–1017}
}

@article{dan:reg:89,
  title={The structure of multiplicatives},
  author={Danos, Vincent and Regnier, Laurent},
  journal={Archive for Mathematical logic},
  volume={28},
  number={3},
  pages={181--203},
  year={1989},
  doi={10.1007/BF01622878},
  publisher={Springer}
}

@article{gir:meanII,
  title={On the meaning of logical rules {II}: multiplicatives and additives},
  author={Girard, Jean-Yves},
  journal={NATO ASI Series F Computer and Systems Sciences},
  volume={175},
  pages={183--212},
  year={2000},
  publisher={Citeseer}
}

@InProceedings{acc:mai:20,
	author =	{Matteo Acclavio and Roberto Maieli},
	title =	{Generalized Connectives for Multiplicative Linear Logic},
	booktitle =	{28th EACSL Annual Conference on Computer Science Logic (CSL 2020)},
	pages =	{6:1--6:16},
	series =	{LIPIcs},
	ISBN =	{978-3-95977-132-0},
	ISSN =	{1868-8969},
	year =	{2020},
	volume =	{152},
	editor =	{Maribel Fern{\'a}ndez and Anca Muscholl},
	publisher =	{Schloss Dagstuhl--Leibniz-Zentrum fuer Informatik},
	address =	{Dagstuhl, Germany},
	URL =		{https://drops.dagstuhl.de/opus/volltexte/2020/11649},
	URN =		{urn:nbn:de:0030-drops-116490},
	doi =		{10.4230/LIPIcs.CSL.2020.6}
}

@InProceedings{saurin:Interpolation,
  author =	{Saurin, Alexis},
  title =	{{Interpolation as Cut-Introduction: On the Computational Content of Craig-Lyndon Interpolation}},
  booktitle =	{10th International Conference on Formal Structures for Computation and Deduction (FSCD 2025)},
  pages =	{32:1--32:21},
  series =	{Leibniz International Proceedings in Informatics (LIPIcs)},
  ISBN =	{978-3-95977-374-4},
  ISSN =	{1868-8969},
  year =	{2025},
  volume =	{337},
  editor =	{Fern\'{a}ndez, Maribel},
  publisher =	{Schloss Dagstuhl -- Leibniz-Zentrum f{\"u}r Informatik},
  address =	{Dagstuhl, Germany},
  URL =		{https://drops.dagstuhl.de/entities/document/10.4230/LIPIcs.FSCD.2025.32},
  URN =		{urn:nbn:de:0030-drops-236478},
  doi =		{10.4230/LIPIcs.FSCD.2025.32}
}

@misc{PN:interpolation,
	author = {Fiorillo, Luca and Osorio-Valencia, Juan and Saurin, Alexis},
	title = {On Correctness, Sequentialization and Interpolation},
	howpublished = {Talk at TLLA 2025: 17th International Workshop on the Theory and Applications of Linear Logic and Related Topics},
	year = {2025},
	note = {\url{https://lipn.univ-paris13.fr/TLLA/2025/abstracts/14-fiorillo-osoriovalencia-saurin.pdf}}
}

@book{mar:mon:unification,
	title={Unification in linear time and space: A structured presentation},
	author={Martelli, Alberto and Montanari, Ugo},
	year={1976},
	publisher={Istituto di Elaborazione della Informazione, Consiglio Nazionale delle Ricerche}
}

@article{ret:handsome,
	title = {Handsome proof-nets: perfect matchings and cographs},
	journal = {Theoretical Computer Science},
	volume = {294},
	number = {3},
	pages = {473-488},
	year = {2003},
	note = {Linear Logic},
	issn = {0304-3975},
	doi = {https://doi.org/10.1016/S0304-3975(01)00175-X},
	url = {https://www.sciencedirect.com/science/article/pii/S030439750100175X},
	author = {Christian Retoré}
}

@article{bellin:subnets,
	title={Subnets of proof-nets in multiplicative linear logic with MIX},
	author={Bellin, Gianluigi},
	journal={Mathematical Structures in Computer Science},
	volume={7},
	number={6},
	pages={663--669},
	year={1997},
	publisher={Cambridge University Press}
}

@inproceedings{hug:unification,
	author = {Hughes, Dominic J. D.},
	title = {Unification nets: canonical proof net quantifiers},
	year = {2018},
	isbn = {9781450355834},
	publisher = {Association for Computing Machinery},
	address = {New York, NY, USA},
	url = {https://doi.org/10.1145/3209108.3209159},
	doi = {10.1145/3209108.3209159},
	booktitle = {Proceedings of the 33rd Annual ACM/IEEE Symposium on Logic in Computer Science},
	pages = {540–549},
	numpages = {10},
	location = {Oxford, United Kingdom},
	series = {LICS '18}
}

@InProceedings{hei:hug:str:ALL1,
	author =	{Heijltjes, Willem B. and Hughes, Dominic J. D. and Stra{\ss}burger, Lutz},
	title =	{{Proof Nets for First-Order Additive Linear Logic}},
	booktitle =	{4th International Conference on Formal Structures for Computation and Deduction (FSCD 2019)},
	pages =	{22:1--22:22},
	series =	{Leibniz International Proceedings in Informatics (LIPIcs)},
	ISBN =	{978-3-95977-107-8},
	ISSN =	{1868-8969},
	year =	{2019},
	volume =	{131},
	editor =	{Geuvers, Herman},
	publisher =	{Schloss Dagstuhl -- Leibniz-Zentrum f{\"u}r Informatik},
	address =	{Dagstuhl, Germany},
	URL =		{https://drops.dagstuhl.de/entities/document/10.4230/LIPIcs.FSCD.2019.22},
	URN =		{urn:nbn:de:0030-drops-105297},
	doi =		{10.4230/LIPIcs.FSCD.2019.22}
}

@inproceedings{girard:96:PN,
	author = {Jean-Yves Girard},
	title = {Proof-nets : the parallel syntax for proof-theory},
	booktitle = {Logic and Algebra},
	publisher = {Marcel Dekker, New York},
	editor = {Aldo Ursini and Paolo Agliano},
	year = {1996}
}

@unpublished{hughes:firstorder,
	author    = {Dominic Hughes},
	title={First-order proofs without syntax},
	note="Berkeley Logic Colloquium",
	year=2014,
}

@INPROCEEDINGS{hug:str:wu:CP1,
	author={Hughes, Dominic J. D. and Straßburger, Lutz and Wu, Jui-Hsuan},
	booktitle={2021 36th Annual ACM/IEEE Symposium on Logic in Computer Science (LICS)},
	title={Combinatorial Proofs and Decomposition Theorems for First-order Logic},
	year={2021},
	volume={},
	number={},
	pages={1-13},
	doi={10.1109/LICS52264.2021.9470579}}

@inproceedings{hei:hug:conflict,
	author = {Hughes, Dominic and Heijltjes, Willem},
	title = {Conflict nets: Efficient locally canonical MALL proof nets},
	year = {2016},
	isbn = {9781450343916},
	publisher = {Association for Computing Machinery},
	address = {New York, NY, USA},
	url = {https://doi.org/10.1145/2933575.2934559},
	doi = {10.1145/2933575.2934559},
	booktitle = {Proceedings of the 31st Annual ACM/IEEE Symposium on Logic in Computer Science},
	pages = {437–446},
	numpages = {10},
	location = {New York, NY, USA},
	series = {LICS '16}
}

@book{M80,
	author       = {Robin Milner},
	title        = {A Calculus of Communicating Systems},
	series       = {Lecture Notes in Computer Science},
	volume       = {92},
	publisher    = {Springer},
	year         = {1980},
	url          = {https://doi.org/10.1007/3-540-10235-3},
	doi          = {10.1007/3-540-10235-3},
	isbn         = {3-540-10235-3},
	bibsource    = {dblp computer science bibliography, https://dblp.org}
}

@article{gabbay:pitts:nominal,
	author = {Gabbay, Murdoch J. and Pitts, Andrew M.},
	title = {A New Approach to Abstract Syntax with Variable Binding},
	year = {2002},
	issue_date = {Jul 2002},
	publisher = {Springer-Verlag},
	address = {Berlin, Heidelberg},
	volume = {13},
	number = {3–5},
	issn = {0934-5043},
	url = {https://doi.org/10.1007/s001650200016},
	doi = {10.1007/s001650200016},
	journal = {Form. Asp. Comput.},
	month = {jul},
	pages = {341–363},
	numpages = {23}
}

@PHDTHESIS{danos:phd,
	url = "http://www.theses.fr/1990PA077188",
	title = "La Logique Linéaire appliquée à l'étude de divers processus de normalisation (principalement du Lambda-calcul)",
	author = "Danos, Vincent",
	year = "1990",
	note = "Thèse de doctorat dirigée par Girard, Jean-Yves Mathématiques Paris 7 1990",
	note = "1990PA077188",
	school = "Université Paris 7"
}

@article{yetter:90, 
	title={Quantales and (noncommutative) linear logic}, 
	volume={55}, DOI={10.2307/2274953}, 
	number={1}, 
	journal={Journal of Symbolic Logic}, 
	author={Yetter, David N.}, 
	year={1990}, 
	pages={41–64}
}

@article{abr:rue:noncomI,
	title = {Non-commutative logic {I}: the multiplicative fragment},
	journal = {Annals of Pure and Applied Logic},
	volume = {101},
	number = {1},
	pages = {29-64},
	year = {1999},
	issn = {0168-0072},
	doi = {https://doi.org/10.1016/S0168-0072(99)00014-7},
	url = {https://www.sciencedirect.com/science/article/pii/S0168007299000147},
	author = {V.Michele Abrusci and Paul Ruet}
}

@article{pitts:nominal,
	title = {Nominal logic, a first order theory of names and binding},
	journal = {Information and Computation},
	volume = {186},
	number = {2},
	pages = {165-193},
	year = {2003},
	note = {Theoretical Aspects of Computer Software (TACS 2001)},
	issn = {0890-5401},
	doi = {https://doi.org/10.1016/S0890-5401(03)00138-X},
	url = {https://www.sciencedirect.com/science/article/pii/S089054010300138X},
	author = {Andrew M. Pitts}
}

@article{gir:ll,
	title = {Linear logic},
	journal = {Theoretical Computer Science},
	volume = {50},
	number = {1},
	pages = {1-101},
	year = {1987},
	issn = {0304-3975},
	doi = {10.1016/0304-3975(87)90045-4},
	author = {Jean-Yves Girard}
}

@book{montesi:book,
	author={Montesi, Fabrizio},
	title={Introduction to Choreographies},
	place={Cambridge},
	doi={10.1017/9781108981491},
	publisher={Cambridge University Press},
	year={2023}
}

@phdthesis{retore:phd,
	author =  {Retor{\'e}, Christian},
	title =   "R{\'e}seaux et S{\'e}quents Ordonn{\'e}s",
	school =  {Universit{\'e} Paris VII},
	year =    1993,
}

@article{ret:newPomset,
	title={Pomset Logic: The other approach to noncommutativity in logic},
	author={Retor{\'e}, Christian},
	journal={Joachim Lambek: The Interplay of Mathematics, Logic, and Linguistics},
	pages={299--345},
	year={2021},
	publisher={Springer}
}

@article{tiu:SIS-II,
	author = "Tiu, Alwen Fernanto",
	title = "A System of Interaction and Structure {II}:
	{T}he Need for Deep Inference",
	journal=lmcs,
	volume=2,
	number=2,
	pages="1--24",
	year = "2006",
	doi = "10.2168/LMCS-2(2:4)2006",
}

@unpublished{tito:str:SIS-III,
	TITLE = {{A System of Interaction and Structure III: The Complexity of BV and Pomset Logic}},
	AUTHOR = {Lê Thành Dũng Nguyên and Lutz Straßburger},
	URL = {https://hal.inria.fr/hal-03909547},
	NOTE = {working paper or preprint},
	YEAR = {2022},
	HAL_ID = {hal-03909547},
	HAL_VERSION = {v1},
}

@InProceedings{tito:lutz:csl22,
	author =	{Lê Thành Dũng Nguyên and Lutz Straßburger},
	title =	{{BV and Pomset Logic are not the same}},
	booktitle =	{30th EACSL Annual Conference on Computer Science Logic (CSL 2022)},
	pages =	{3:1--3:17},
	series =	{Leibniz International Proceedings in Informatics (LIPIcs)},
	ISBN =	{978-3-95977-218-1},
	ISSN =	{1868-8969},
	year =	{2022},
	volume =	{216},
	editor =	{Manea, Florin and Simpson, Alex},
	publisher =	{Schloss Dagstuhl -- Leibniz-Zentrum f{\"u}r Informatik},
	address =	{Dagstuhl, Germany},
	URL =		{https://drops.dagstuhl.de/opus/volltexte/2022/15723},
	URN =		{urn:nbn:de:0030-drops-157231},
	doi =		{10.4230/LIPIcs.CSL.2022.3}
}

@misc{hughes:conflict,
      title={Abstract p-time proof nets for MALL: Conflict nets}, 
      author={Dominic J. D. Hughes},
      year={2008},
      eprint={0801.2421},
      archivePrefix={arXiv},
      primaryClass={math.LO},
      url={https://arxiv.org/abs/0801.2421}, 
}

@article{hug:van:slice,
	author = {Hughes, Dominic J. D. and Van Glabbeek, Rob J.},
	title = {Proof nets for unit-free multiplicative-additive linear logic},
	year = {2005},
	issue_date = {October 2005},
	publisher = {Association for Computing Machinery},
	address = {New York, NY, USA},
	volume = {6},
	number = {4},
	issn = {1529-3785},
	url = {https://doi.org/10.1145/1094622.1094629},
	doi = {10.1145/1094622.1094629},
	journal = {ACM Trans. Comput. Logic},
	month = {oct},
	pages = {784–842},
	numpages = {59}
}
%%%%%%%%%%%%%%%%%%%%%%%%%%%%%%%%%%%%%%%%%%%%%%%%%%%%%%%%%%%%%%%%
%%%%%%%%%%%%%%%%%%%%%%%%%%%%%%%%%%%%%%%%%%%%%%%%%%%%%%%%%%%%%%%%

\clearpage

%%%%%%%%%%%%%%%%%%%%%%%%%%%%%%%%%%%%%%%%%%%%%%%%%%%%%%%%%%%%%%%%
%%%%%%%%%%%%%%%%%%%%%%%%%%%%%%%%%%%%%%%%%%%%%%%%%%%%%%%%%%%%%%%%
%%%%%%%%%%%%%%%%%%%%%%%%%%%%%%%%%%%%%%%%%%%%%%%%%%%%%%%%%%%%%%%%
\appendix
%%%%%%%%%%%%%%%%%%%%%%%%%%%%%%%%%%%%%%%%%%%%%%%%%%%%%%%%%%%%%%%%
%%%%%%%%%%%%%%%%%%%%%%%%%%%%%%%%%%%%%%%%%%%%%%%%%%%%%%%%%%%%%%%%
%%%%%%%%%%%%%%%%%%%%%%%%%%%%%%%%%%%%%%%%%%%%%%%%%%%%%%%%%%%%%%%%

%%%%%%%%%%%%%%%%%%%%%%%%%%%%%%%%%%%%%%%%%%%%%%%%%%%%%%%%%%%%%%%%
%%%%%%%%%%%%%%%%%%%%%%%%%%%%%%%%%%%%%%%%%%%%%%%%%%%%%%%%%%%%%%%%
%%%%%%%%%%%%%%%%%%%%%%%%%%%%%%%%%%%%%%%%%%%%%%%%%%%%%%%%%%%%%%%%
\section{Details of the proofs of \Cref{lem:confCoal}}\label{app:proofs}
%%%%%%%%%%%%%%%%%%%%%%%%%%%%%%%%%%%%%%%%%%%%%%%%%%%%%%%%%%%%%%%%
%%%%%%%%%%%%%%%%%%%%%%%%%%%%%%%%%%%%%%%%%%%%%%%%%%%%%%%%%%%%%%%%
%%%%%%%%%%%%%%%%%%%%%%%%%%%%%%%%%%%%%%%%%%%%%%%%%%%%%%%%%%%%%%%%

%%%%%%%%%%%%%%%%%%%%%%%%%%%%%%%%%%%%%%%%%%%%%%%%%%%%%%%%%%%%%%%%
\def\nowDer{
	\newline The two distinct derivations labeling the link in the bottom-right corner of the diagram according to the anticlockwise and clockwise sequence of coalescence steps are:
}
\begin{lemma}
	Let $\pn$ be a \conet such that $\pn_1 \leftarrow \pn \coalto \pn_2$ then there is a \conet $\pn'$ such that $\pn_1 \coaltos \pn'$ and $\pn_2 \coaltos \pn'$, 
	and the derivations corresponding to the coalescence paths $\pn\coalto\pn_2\coaltos \pn'$ and $\pn\coalto\pn_1\coaltos \pn'$ are equivalent modulo the local rule permutations in \Cref{fig:permutations1}.
\end{lemma}
\begin{proof}
	We only discuss the critical pairs for coalescence rules not already discussed in \cite{hei:hug:conflict}.
	Together with the confluence diagram of each critical pair, we show the two derivations corresponding to the two sequences of coalescence steps.
	\begin{itemize}
		\item Case $\lpar/\lprec$:

		$$
		\adjustbox{max width=\textwidth}{$\begin{array}{ccc}
			\viA1 \vpz1{\lprec} \viB1,\ldots,\viA{n} \vpz2{\lprec} \viB{n}, \vC1 \vpz3{\lpar} \vD1, \vpz4{\Gamma},\vpz5{\Delta}
			\pzlinks{A1/{An}/12/\la/red/{C1,D1,pz4}}
			\pzlinks{B1/{Bn}/-12/\lb/blue/{pz5}}
			&\to&
			\viA1 \vpz1{\lprec} \viB1,\ldots,\viA n \vpz2{\lprec} \viB n,\vC1 \vpz3{\lpar} \vD1, \vpz4{\Gamma},\vpz5{\Delta}
			\pzlinks{pz1/pz2/12/\labb/pzgreen/{C1,D1,pz4,pz5}}
			\\\\
			\downarrow && \downarrow
			\\\\
			\viA1 \vpz1{\lprec} \viB1,\ldots,\viA n \vpz2{\lprec} \viB n, \vC1 \vpz3{\lpar} \vD1, \vpz4{\Gamma},\vpz5{\Delta}
			\pzlinks{A1/{An}/12/\laa/red/{pz3,pz4}}
			\pzlinks{B1/{Bn}/-12/\lb/blue/{pz5}}
			&\to&
			\viA1 \vpz1{\lprec} \viB1,\ldots,\viA n \vpz2{\lprec} \viB n, \vC1 \vpz3{\lpar} \vD1, \vpz4{\Gamma},\vpz5{\Delta}
			\pzlinks{pz1/pz2/12/\laabb/magenta/{pz3,pz4,pz5}}
		\end{array}$}
		$$

		With $\dualizerof[\laa]= \dualizerof[\la]$ and $\dualizerof[\labb] = \dualizerof[\laabb] = \dualizerof[\la] \duasum \dualizerof[\lb]$.
		\nowDer

		$$\adjustbox{max width=\textwidth}{$\begin{array}{c}
			\vlderivation{
			\vlin{\lpar}{}{
				\sdash A_1\lprec B_1,\ldots, A_n\lprec B_n,C\lpar D,\Gamma, \Delta
			}{
				\vliin{\lprec}{}{
					\sdash A_1\lprec B_1, \ldots, A_n\lprec B_n,C, D,\Gamma, \Delta
				}{
					\vlhy{\sdash A_1,\ldots, A_n,C \lpar D,\Gamma}
				}{
					\vlhy{\sdash B_1,\ldots, B_n,\Delta}
				}
			}
		}
		\\[20pt]
		\quad\peq\quad
		\\[5pt]
		\vlderivation{
			\vliin{\lprec}{}{\sdash A_1\lprec B_1,\ldots, A_n\lprec B_n,C\lpar D,\Gamma, \Delta}{
				\vlin{\lpar}{}{
					\sdash A_1,\ldots, A_n,C \lpar D,\Gamma
				}{
					\vlhy{\sdash A_1,\ldots, A_n,C, D,\Gamma}
				}
			}{
				\vlhy{\sdash B_1,\ldots, B_n,\Delta}
			}
		}
		\end{array}
		$}
		$$

		\item Case $\lplus/\lprec$:
		$$
		\adjustbox{max width=\textwidth}{$\begin{array}{ccc}
			\viA1 \vpz1{\lprec} \viB1,\ldots,\viA{n} \vpz2{\lprec} \viB{n}, \vC1 \vpz3{\oplus} \vD1, \vpz4{\Gamma},\vpz5{\Delta}
			\pzlinks{A1/{An}/12/\la/red/{C1,pz4}}
			\pzlinks{B1/{Bn}/-12/\lb/blue/{pz5}}
			&\to&
			\viA1 \vpz1{\lprec} \viB1,\ldots,\viA{n} \vpz2{\lprec} \viB{n},\vC1 \vpz3{\oplus} \vD1, \vpz4{\Gamma},\vpz5{\Delta}
			\pzlinks{pz1/pz2/12/\labb/pzgreen/{C1,pz4,pz5}}
			\\\\
			\downarrow && \downarrow
			\\\\
			\viA1 \vpz1{\lprec} \viB1,\ldots,\viA{n} \vpz2{\lprec} \viB{n}, \vC1 \vpz3{\oplus} \vD1, \vpz4{\Gamma},\vpz5{\Delta}
			\pzlinks{A1/{An}/12/\laa/red/{pz3,pz4}}
			\pzlinks{B1/{Bn}/-12/\lb/blue/{pz5}}
			&\to&
			\viA1 \vpz1{\lprec} \viB1,\ldots,\viA{n} \vpz2{\lprec} \viB{n}, \vC1 \vpz3{\oplus} \vD1, \vpz4{\Gamma},\vpz5{\Delta}
			\pzlinks{pz1/pz2/12/\laabb/magenta/{pz3,pz4,pz5}}
		\end{array}$}
		$$

		With $\dualizerof[\laa]= \dualizerof[\la]$ and $\dualizerof[\labb] = \dualizerof[\laabb] = \dualizerof[\la] \duasum \dualizerof[\lb]$.
		\nowDer
		$$
		\adjustbox{max width=\textwidth}{$\begin{array}{c}
			\vlderivation{
			\vlin{\lplus}{}{
				\sdash A_1\lprec B_1, \ldots, A_n\lprec B_n, C\lplus D,\Gamma, \Delta
			}{
				\vliin{\lprec}{}{
					\sdash A_1\lprec B_1,\ldots, A_n\lprec B_n,C,\Gamma, \Delta
				}{
					\vlhy{\sdash A_1,\ldots, A_n,C,\Gamma}
				}{
					\vlhy{\sdash B_1,\ldots, B_n,\Delta}
				}
			}
		}
		\\[20pt]
		\quad\peq\quad
		\\[5pt]
		\vlderivation{
			\vliin{\lprec}{}{\sdash A_1\lprec B_1,\ldots, A_n\lprec B_n, C\lplus D,\Gamma, \Delta}{
				\vlin{\lplus}{}{
					\sdash A_1,\ldots, A_n, C,\Gamma
				}{
					\vlhy{\sdash A_1,\ldots, A_n, C,\Gamma}
				}
			}{
				\vlhy{\sdash B_1,\ldots, B_n,\Delta}
			}
		}
		\end{array}$}
		$$

		\item Case $\lprec/\lprec$:

		$$
		\hskip-1.5em\adjustbox{max width=\textwidth}{$
		\begin{array}{ccc}
			\viA1 \vpz1{\lprec} \viB1,\ldots, \viA n \vpz2{\lprec} \viB n, \viC1 \vpz3{\lprec} \viD1,\ldots, \viC m \vpz4{\lprec} \viD m, \vpz5{\Gamma},\vpz6{\Delta}, \vpz7{\Sigma}
			\pzlinks{A1/{An}/12/\la/red/{C1,{Cm},pz5}}
			\pzlinks{B1/{Bn}/-12/\lb/blue/{pz6}}
			\pzlinks{D1/{Dm}/-18/\lc/magenta/{pz7}}
			&\to&
			\viA1 \vpz1{\lprec} \viB1,\ldots, \viA n \vpz2{\lprec} \viB n, \viC1 \vpz3{\lprec} \viD1,\ldots, \viC m \vpz4{\lprec} \viD m, \vpz5{\Gamma},\vpz6{\Delta}, \vpz7{\Sigma}
			\pzlinks{pz1/pz2/12/\labb/pzgreen/{C1,{Cm},pz5,pz6}}
			\pzlinks{D1/{Dm}/-12/\lc/magenta/{pz7}}
			\\\\
			\downarrow && \downarrow
			\\\\
			\viA1 \vpz1{\lprec} \viB1,\ldots, \viA n \vpz2{\lprec} \viB n, \viC1 \vpz3{\lprec} \viD1,\ldots, \viC m \vpz4{\lprec} \viD m, \vpz5{\Gamma},\vpz6{\Delta}, \vpz7{\Sigma}
			\pzlinks{A1/{An}/12/\lac/violet/{pz3,pz4,pz5,pz7}}
			\pzlinks{B1/{Bn}/-12/\lb/blue/{pz6}}
			&\to&
			\viA1 \vpz1{\lprec} \viB1,\ldots, \viA n \vpz2{\lprec} \viB n, \viC1 \vpz3{\lprec} \viD1,\ldots, \viC m \vpz4{\lprec} \viD m, \vpz5{\Gamma},\vpz6{\Delta}, \vpz7{\Sigma}
			\pzlinks{pz1/pz2/12/\labc/brown/{pz3,pz4,pz5,pz6,pz7}}
		\end{array}
		$}$$

		With $\dualizerof[\labb] = \dualizerof[\la] \duasum \dualizerof[\lb]$, $\dualizerof[\lac] = \dualizerof[\la] \duasum \dualizerof[\lc]$
		and $\dualizerof[\labc] = \dualizerof[\la] \duasum \dualizerof[\lb] \duasum \dualizerof[\lc]$.
		\nowDer
		$$\adjustbox{max width=\textwidth}{$\begin{array}{c}
			\vlderivation{
			\vliin{\lprec}{}{
				\sdash A_1\lprec B_1,\ldots, A_n\lprec B_n, C_1\lprec D_1,\ldots, C_m \lprec D_m,\Gamma, \Delta,\Sigma
			}{
				\vliin{\lprec}{}{
					\sdash A_1,\ldots, A_n, C_1\lprec D_1,\ldots, C_m \lprec D_m,\Gamma,\Sigma
				}{
					\vlhy{\sdash A_1,\ldots, A_n,C_1,\ldots, C_m, \Gamma}
				}{
					\vlhy{\sdash D_1,\ldots, D_m,\Sigma}
				}
			}{
				\vlhy{\sdash B_1,\ldots, B_n,\Delta}
			}
		}
		\\[20pt]
		\quad\peq\quad
		\\[5pt]
		\vlderivation{
			\vliin{\lprec}{}{\sdash A_1\lprec B_1,\ldots, A_n\lprec B_n, C_1\lprec D_1,\ldots, C_m \lprec D_m,\Gamma, \Delta,\Sigma}{
				\vliin{\lprec}{}{
					\sdash A_1\lprec B_1,\ldots, A_n\lprec B_n, C_1,\ldots, C_m,\Gamma, \Delta
				}{
					\vlhy{\sdash A_1,\ldots, A_n, C_1,\ldots, C_m,\Gamma}
				}{
					\vlhy{\sdash B_1,\ldots, B_n,\Delta}
				}
			}{
				\vlhy{\sdash D_1,\ldots, D_m,\Sigma}
			}
		}
		\end{array}$}
		$$

		\item Case $\ltens/\lprec$:
		$$\adjustbox{max width=\textwidth}{$
			\begin{array}{ccc}
				\\
				\viA1 \vpz1{\lprec} \viB1,\dots, \viA n \vpz2{\lprec} \viB n, \vC1 \vpz3{\ltens} \vD1, \vpz4{\Gamma}, \vpz5{\Delta}, \vpz6{\Sigma}
				\pzlinks{A1/{An}/12/\la/red/{C1,pz4}}
				\pzlinks{B1/{Bn}/-12/\lb/blue/{pz5}}
				\pzlinks{D1/pz6/-18/\lc/violet/}
				&\to&
				\viA1 \vpz1{\lprec} \viB1,\dots, \viA n \vpz2{\lprec} \viB n, \vC1 \vpz3{\ltens} \vD1, \vpz4{\Gamma}, \vpz5{\Delta}, \vpz6{\Sigma}
				\pzlinks{pz1/pz2/12/\labb/pzgreen/{C1,pz4,pz5}}
				\pzlinks{D1/pz6/-12/\lc/violet/}
				\\\\
				\downarrow && \downarrow
				\\\\
				\viA1 \vpz1{\lprec} \viB1,\dots, \viA n \vpz2{\lprec} \viB n, \vC1 \vpz3{\ltens} \vD1, \vpz4{\Gamma}, \vpz5{\Delta}, \vpz6{\Sigma}
				\pzlinks{A1/{An}/12/\lac/orange/{pz3,pz4,pz6}}
				\pzlinks{B1/{Bn}/-12/\lb/blue/{pz5}}
				&\to&
				\viA1 \vpz1{\lprec} \viB1,\dots, \viA n \vpz2{\lprec} \viB n, \vC1 \vpz3{\ltens} \vD1, \vpz4{\Gamma}, \vpz5{\Delta}, \vpz6{\Sigma}
				\pzlinks{pz1/pz2/12/\labc/brown/{pz3,pz4,pz5,pz6}}
				\\\\
			\end{array}
			$}$$

		With $\dualizerof[\labb] = \dualizerof[\la] \duasum \dualizerof[\lb]$, $\dualizerof[\lac] = \dualizerof[\la] \duasum \dualizerof[\lc]$
		and $\dualizerof[\labc] = \dualizerof[\la] \duasum \dualizerof[\lb] \duasum \dualizerof[\lc]$.
		\nowDer
		$$\adjustbox{max width=\textwidth}{$
			\begin{array}{c}
				\vlderivation{
				\vliin{\ltens}{}{
					\sdash A_1\lprec B_1,\ldots, A_n\lprec B_n,C\ltens D,\Gamma, \Delta,\Sigma
				}{
					\vliin{\lprec}{}{
						\sdash A_1\lprec B_1,\ldots, A_n\lprec B_n,C,\Gamma,\Delta
					}{
						\vlhy{\sdash A_1,\dots, A_n,C,\Gamma}
					}{
						\vlhy{\sdash B_1,\ldots,B_n,\Delta}
					}
				}{
					\vlhy{\sdash D,\Sigma}
				}
			}
			\\[20pt]
			\quad\peq\quad
			\\[5pt]
			\vlderivation{
				\vliin{\lprec}{}{
					\sdash A_1\lprec B_1,\ldots, A_n\lprec B_n,C\ltens D,\Gamma, \Delta,\Sigma
				}{
					\vliin{\ltens}{}{
						\sdash A_1,\ldots, A_n,C\ltens D,\Gamma, \Sigma
					}{
						\vlhy{\sdash A_1,\ldots,A_n,C,\Gamma}
					}{
						\vlhy{\sdash D,\Sigma}
					}
				}{
					\vlhy{\sdash B_1,\ldots,B_n,\Delta}
				}
			}
			\end{array}
		$}$$

		\

		\item Case $\exists/\lprec$ :
		$$
		\adjustbox{max width=\textwidth}{$\begin{array}{ccc}
			\viA1 \vpz1{\lprec} \viB1,\ldots, \viA n \vpz2{\lprec} \viB n, \vpz3{\exists}x.\vC1, \vpz4{\Gamma}, \vpz5{\Delta}
			\pzlinks{A1/{An}/12/\la/red/{C1,pz4}}
			\pzlinks{B1/{Bn}/-12/\lb/blue/{pz5}}
			&\to&
			\viA1 \vpz1{\lprec} \viB1,\ldots, \viA n \vpz2{\lprec} \viB n, \vpz3{\exists}x.\vC1, \vpz4{\Gamma}, \vpz5{\Delta}
			\pzlinks{pz1/pz2/12/\labb/pzgreen/{C1,pz4,pz5}}
			\\\\
			\downarrow && \downarrow
			\\\\
			\viA1 \vpz1{\lprec} \viB1,\ldots, \viA n \vpz2{\lprec} \viB n, \vpz3{\exists}x.\vC1, \vpz4{\Gamma}, \vpz5{\Delta}
			\pzlinks{A1/{An}/12/\laa/violet/{pz3,pz4}}
			\pzlinks{B1/{Bn}/-12/\lb/blue/{pz5}}
			&\to&
			\viA1 \vpz1{\lprec} \viB1,\ldots, \viA n \vpz2{\lprec} \viB n, \vpz3{\exists}x.\vC1, \vpz4{\Gamma}, \vpz5{\Delta}
			\pzlinks{pz1/pz2/12/\laabb/brown/{pz3,pz4,pz5}}
		\end{array}$}
		$$

		With $\dualizerof[ab] = \dualizerof[\la] \duasum \dualizerof[\lb]$, $\dualizerof[\laa] = \dualizerof[\la]\fsubminus{x}$
		and $\dualizerof[a'b] = \dualizerof[\labb] \fsubminus{x}$
		\nowDer
		$$
		\adjustbox{max width=\textwidth}{$
		\begin{array}{c}
			\vlderivation{
			\vlin{\exists}{}{
				\sdash A_1\lprec B_1,\ldots, A_n\lprec B_n,\exists x C,\Gamma, \Delta
			}{
				\vliin{\lprec}{}{
					\sdash A_1\lprec B_1,\ldots, A_n\lprec B_n,C \fsubst{c}{x},\Gamma, \Delta
				}{
					\vlhy{\sdash A_1,\ldots, A_n,C \fsubst{c}{x},\Gamma}
				}{
					\vlhy{\sdash B_1,\ldots, B_n,\Delta}
				}
			}
		}
		\\[20pt]
		\quad\peq\quad
		\\[5pt]
		\vlderivation{
			\vliin{\lprec}{}{\sdash A_1\lprec B_1,\ldots, A_n\lprec B_n, \exists xC , \Gamma, \Delta}{
				\vlin{\exists}{}{
					\sdash A_1,\ldots, A_n,\exists xC ,\Gamma
				}{
					\vlhy{\sdash A_1,\ldots, A_n,C \fsubst{c}{x} ,\Gamma}
				}
			}{
				\vlhy{\sdash B_1,\ldots, B_n,\Delta}
			}
		}
		\end{array}
		$}
		$$

		\

		\item Case $\forall/\lprec$:
		$$
		\adjustbox{max width=\textwidth}{$\begin{array}{ccc}
			\viA1 \vpz1{\lprec} \viB1,\ldots, \viA n \vpz2{\lprec} \viB n, \vpz3{\forall}x.\vC1, \vpz4{\Gamma}, \vpz5{\Delta}
			\pzlinks{A1/{An}/12/\la/red/{C1,pz4}}
			\pzlinks{B1/{Bn}/-12/\lb/blue/{pz5}}
			&\to&
			\viA1 \vpz1{\lprec} \viB1,\ldots, \viA n \vpz2{\lprec} \viB n, \vpz3{\forall}x.\vC1, \vpz4{\Gamma}, \vpz5{\Delta}
			\pzlinks{pz1/pz2/12/\labb/pzgreen/{C1,pz4,pz5}}
			\\\\
			\downarrow && \downarrow
			\\\\
			\viA1 \vpz1{\lprec} \viB1,\ldots, \viA n \vpz2{\lprec} \viB n, \vpz3{\forall}x.\vC1, \vpz4{\Gamma}, \vpz5{\Delta}
			\pzlinks{A1/{An}/12/\la/violet/{pz3,pz4}}
			\pzlinks{B1/{Bn}/-12/\lb/blue/{pz5}}
			&\to&
			\viA1 \vpz1{\lprec} \viB1,\ldots, \viA n \vpz2{\lprec} \viB n, \vpz3{\forall}x.\vC1, \vpz4{\Gamma}, \vpz5{\Delta}
			\pzlinks{pz1/pz2/12/\labb/brown/{pz3,pz4,pz5}}
		\end{array}$}
		$$

		With $\dualizerof[\labb] = \dualizerof[\la] \duasum \dualizerof[\lb]$.
		\nowDer
		$$
		\adjustbox{max width=\textwidth}{$
		\begin{array}{c}
			\vlderivation{
			\vlin{\forall}{}{
				\sdash A_1\lprec B_1,\ldots, A_n\lprec B_n,\lFa x C,\Gamma, \Delta
			}{
				\vliin{\lprec}{}{
					\sdash A_1\lprec B_1,\ldots, A_n\lprec B_n,C  ,\Gamma, \Delta
				}{
					\vlhy{\sdash A_1,\ldots, A_n,C ,\Gamma}
				}{
					\vlhy{\sdash B_1,\ldots, B_n,\Delta}
				}
			}
		}
		\\[20pt]
		\quad\peq\quad
		\\[5pt]
		\vlderivation{
			\vliin{\lprec}{}{\sdash A_1\lprec B_1,\ldots, A_n\lprec B_n, \lFa xC , \Gamma, \Delta}{
				\vlin{\forall}{}{
					\sdash A_1,\ldots, A_n,\lFa xC ,\Gamma
				}{
					\vlhy{\sdash A_1,\ldots, A_n,C  ,\Gamma}
				}
			}{
				\vlhy{\sdash B_1,\ldots, B_n,\Delta}
			}
		}
		\end{array}
		$}
		$$

		\

		\item Case $\naloadr/\lprec$ with $\nabla\in\set{\lnewsymb,\lyasymb}$:
		similarly to the case $\forall/\lprec$, but considering the rule $\naloadr$ and the nominal quantifier $\nabla$ instead of $\forall$.

		$$
		\adjustbox{max width=\textwidth}{$\begin{array}{ccc}
			\viA1 \vpz1{\lprec} \viB1,\ldots, \viA n \vpz2{\lprec} \viB n, \vpz3{\nabla}x.\vC1, \vpz4{\Gamma}, \vpz5{\Delta}
			\pzlinks{A1/{An}/12/\la/red/{C1,pz4}}
			\pzlinks{B1/{Bn}/-12/\lb/blue/{pz5}}
			&\to&
			\viA1 \vpz1{\lprec} \viB1,\ldots, \viA n \vpz2{\lprec} \viB n, \vpz3{\nabla}x.\vC1, \vpz4{\Gamma}, \vpz5{\Delta}
			\pzlinks{pz1/pz2/12/\labb/pzgreen/{C1,pz4,pz5}}
			\\\\
			\downarrow && \downarrow
			\\\\
			\viA1 \vpz1{\lprec} \viB1,\ldots, \viA n \vpz2{\lprec} \viB n, \vpz3{\nabla}x.\vC1, \vpz4{\Gamma}, \vpz5{\Delta}
			\pzlinks{A1/{An}/12/\la/violet/{pz3,pz4}}
			\pzlinks{B1/{Bn}/-12/\lb/blue/{pz5}}
			&\to&
			\viA1 \vpz1{\lprec} \viB1,\ldots, \viA n \vpz2{\lprec} \viB n, \vpz3{\nabla}x.\vC1, \vpz4{\Gamma}, \vpz5{\Delta}
			\pzlinks{pz1/pz2/12/\labb/brown/{pz3,pz4,pz5}}
		\end{array}$}
		$$
		With $\dualizerof[\labb] = \dualizerof[\la] \duasum \dualizerof[\lb]$.
		\nowDer
		$$
		\adjustbox{max width=\textwidth}{$
		\begin{array}{c}
			\vlderivation{
			\vlin{\naloadr}{}{
				\sdash A_1\lprec B_1,\ldots, A_n\lprec B_n,\lNa x C,\Gamma, \Delta
			}{
				\vliin{\lprec}{}{
					\sdash[x] A_1\lprec B_1,\ldots, A_n\lprec B_n,C ,\Gamma, \Delta
				}{
					\vlhy{\sdash[x] A_1,\ldots, A_n,C ,\Gamma}
				}{
					\vlhy{\sdash[x] B_1,\ldots, B_n,\Delta}
				}
			}
		}
		\\[20pt]
		\quad\peq\quad
		\\[5pt]
		\vlderivation{
			\vliin{\lprec}{}{\sdash A_1\lprec B_1,\ldots, A_n\lprec B_n, \lNa xC , \Gamma, \Delta}{
				\vlin{\naloadr}{}{
					\sdash A_1,\ldots, A_n,\lNa xC ,\Gamma
				}{
					\vlhy{\sdash[x] A_1,\ldots, A_n,C ,\Gamma}
				}
			}{
				\vlhy{\sdash B_1,\ldots, B_n,\Delta}
			}
		}
		\end{array}
		$}
		$$

		\

		\item Case $\naur/\lprec$ with $\nabla\in\set{\lnewsymb,\lyasymb}$:
		similar to the previous case, but considering the rule $\naur$ instead of $\naloadr$.

		\item Case $\napopr/\lprec$ with $\nabla\in\set{\lnewsymb,\lyasymb}$:

		$$
		\adjustbox{max width=\textwidth}{$\begin{array}{ccc}
			\viA1 \vpz1{\lprec} \viB1,\ldots, \viA n \vpz2{\lprec} \viB n, \vpz3{\nabla} \vx1.\vC1,\vpz4{\cneg \nabla} \vy1.\vD1, \vpz5{\Gamma}, \vpz6{\Delta}
			\pzlinks{A1/{An}/12/\la/red/{C1,D1,pz5}}
			\pzlinks{B1/{Bn}/-12/\lb/blue/{pz6}}
			\pzlinks{x1/y1/15/\lc/brown/}
			&\to&
			\viA1 \vpz1{\lprec} \viB1,\ldots, \viA n \vpz2{\lprec} \viB n, \vpz3{\nabla}\vx1.\vC1,\vpz4{\cneg \nabla}\vy1.\vD1, \vpz5{\Gamma}, \vpz6{\Delta}
			\pzlinks{pz1/pz2/12/\labb/pzgreen/{C1,D1,pz5,pz6}}
			\pzlinks{x1/y1/15/\lc/brown/}
			\\\\
			\downarrow && \downarrow
			\\\\
			\viA1 \vpz1{\lprec} \viB1,\ldots, \viA n \vpz2{\lprec} \viB n, \vpz3{\nabla}x.\vC1,\vpz4{\cneg \nabla}y.\vD1, \vpz5{\Gamma}, \vpz6{\Delta}
			\pzlinks{A1/{An}/12/\laa/violet/{C1,pz4,pz5}}
			\pzlinks{B1/{Bn}/-12/\lb/blue/{pz6}}
			&\to&
			\viA1 \vpz1{\lprec} \viB1,\ldots, \viA n \vpz2{\lprec} \viB n, \vpz3{\nabla}x.\vC1,\vpz4{\cneg \nabla}y.\vD1, \vpz5{\Gamma}, \vpz6{\Delta}
			\pzlinks{pz1/pz2/12/\laabb/purple/{C1,pz4,pz5,pz6}}
		\end{array}$}
		$$

		With $\dualizerof[\laa] = \dualizerof[a] \fsubminus{y}$, $\dualizerof[\labb] = \dualizerof[\la] \duasum \dualizerof[\lb]$
		and $\dualizerof[\laabb] = \dualizerof[\labb]\fsubminus{y}$.
		\nowDer
		$$\adjustbox{max width=\textwidth}{$
		\begin{array}{c}
			\vlderivation{
			\vlin{\napopr}{}{
				\sdash[\sS_1,\sS_2,\isna x] A_1\lprec B_1,\ldots, A_n\lprec B_n, C, \lnNa yD,\Gamma, \Delta
			}{
				\vliin{\lprec}{}{
					\sdash[\sS_1,\sS_2] A_1\lprec B_1,\ldots, A_n\lprec B_n, C , D\fsubst xy,\Gamma, \Delta
				}{
					\vlhy{\sdash[\sS_1] A_1,\ldots, A_n, C, D\fsubst xy ,\Gamma}
				}{
					\vlhy{\sdash[\sS_1] B_1,\ldots, B_n,\Delta}
				}
			}
		}
		\\[20pt]
		\quad\peq\quad
		\\[5pt]
		\vlderivation{
			\vliin{\lprec}{}{
				\sdash[\sS_1,\isna x] A_1\lprec B_1,\ldots, A_n\lprec B_n, C,\lnNa yD,\Gamma, \Delta
			}{
				\vlin{\napopr}{}{
					\sdash[\sS_1, \isna x] A_1,\ldots, A_n, C,\lnNa yD,\Gamma
				}{
					\vlhy{\sdash[\sS_1] A_1,\ldots, A_n,C ,D\fsubst{x}{y} ,\Gamma}
				}
			}{
				\vlhy{\sdash[\sS_2] B_1,\ldots, B_n,\Delta}
			}
		}
		\end{array}
		$}$$

\
		\item
		Case $\exists / \naloadr$ with  $\nabla\in\set{\lnewsymb,\lyasymb}$:

		$$
		\adjustbox{max width=\textwidth}{$
		\begin{array}{ccc}
			\vpz1{\exists}x. \vA1, \vpz3{\nabla}y.\vB1, \vpz4{\Gamma}
			\pzlinks{A1/B1/12/\la/red/{pz4}}
			&\to&
			\vpz1{\exists}x. \vA1, \vpz3{\nabla}y.\vB1, \vpz4{\Gamma}
			\pzlinks{pz1/B1/12/\link[black]{a_1}/pzgreen/{pz4}}
			\\
			\downarrow && \downarrow
			\\\\
			\vpz1{\exists}x. \vA1,\quad \vpz3{\nabla}y.\vB1, \vpz4{\Gamma}
			\pzlinks{A1/pz3/12/\link[black]{a_2}/violet/{pz4}}
			&\to&
			\vpz1{\exists}x. \vA1,\quad \vpz3{\nabla}y.\vB1, \vpz4{\Gamma}
			\pzlinks{pz1/pz3/12/\link[black]{a_3}/purple/{pz4}}
		\end{array}$}
		$$

		With $\dualizerof[a_1] = \dualizerof[\la] \fsubminus{x}$, $\dualizerof[a_2] = \dualizerof[\la] $
		and $\dualizerof[a_3] = \dualizerof[a_1]$.
		\nowDer

		$$
		\adjustbox{max width=\textwidth}{$
		\begin{array}{c}
			\vlderivation{
			\vlin{\naloadr}{}{
				\sdash \lEx xA, \lNa yB,\Gamma
			}{
				\vlin{\exists}{}{
					\sdash[\sS,\isna y] \lEx xA, B ,\Gamma
				}{
					\vlhy{\sdash[\sS,\isna y] A \fsubst{c}{x}, B ,\Gamma}
				}
			}
		}
		\quad\peq\quad
		\vlderivation{
			\vlin{\exists}{}{
				\sdash \lEx xA, \lNa yB,\Gamma
			}{
				\vlin{\naloadr}{}{
					\sdash A\fsubst{c}{x}, \lNa yB,\Gamma
				}{
					\vlhy{\sdash[\sS,\isna y] A\fsubst{c}{x}, B ,\Gamma}
				}
			}
		}
		\end{array}
		$}
		$$

		\

		\item
		Case $\exists / \naur$ with $\nabla\in\set{\lnewsymb,\lyasymb}$:
		similar to the previous case, but considering the rule $\naur$ instead of $\naloadr$.

		\item
		Case $\exists / \napopr$ with  $\nabla\in\set{\lnewsymb,\lyasymb}$:

		$$
		\adjustbox{max width=\textwidth}{$\begin{array}{ccc}
			\vpz1{\exists}z. \vA1, \vpz3{\nabla}\vx1.\vB1,\vpz4{\cneg\nabla}\vy1.\vC1, \vpz5{\Gamma}
			\pzlinks{A1/B1/12/\la/red/{C1,pz5}}
			\pzlinks{x1/y1/15/\lc/brown/}
			&\to&
			\vpz1{\exists}z. \vA1, \vpz3{\nabla}\vx1.\vB1,\vpz4{\cneg\nabla}\vy1.\vC1, \vpz5{\Gamma}
			\pzlinks{pz1/B1/12/\labb/pzgreen/{C1,pz5}}
			\pzlinks{x1/y1/15/\lc/brown/}
			\\
			\downarrow && \downarrow
			\\\\
			\vpz1{\exists}z. \vA1, \vpz3{\nabla}x.\vB1,\vpz4{\cneg\nabla}y.\vC1, \vpz5{\Gamma}
			\pzlinks{B1/pz4/12/\laaa/violet/{A1,pz5}}
			&\to&
			\vpz1{\exists}z. \vA1, \vpz3{\nabla}x.\vB1,\vpz4{\cneg\nabla}y.\vC1, \vpz5{\Gamma}
			\pzlinks{pz1/B1/12/\laaabb/purple/{pz4,pz5}}
		\end{array}$}
		$$

		With $\dualizerof[\labb] = \dualizerof[\la] \fsubminus{z}$,
		$\dualizerof[\laaa] = \dualizerof[\la] \fsubminus{y}$ and  $\dualizerof[\laaabb] = (\dualizerof[\la]\fsubminus{z} )\fsubminus{y}$
		\nowDer

		\
		
		$$
		\adjustbox{max width=\textwidth}{$
		\begin{array}{c}
			\vlderivation{
			\vlin{\napopr}{}{
				\sdash[\sS,\isna x] \lEx zA, B, \lnNa yC,\Gamma
			}{
				\vlin{\exists}{}{
					\sdash \lEx zA,B , C\fsubst{x}{y}, \Gamma
				}{
					\vlhy{\sdash A\fsubst{c}{z}, B ,C\fsubst{x}{y}, \Gamma}
				}
			}
		}
		\qquad
		\peq
		\qquad
		\vlderivation{
			\vlin{\exists}{}{
				\sdash[\sS,\isna x] \lEx zA, B, \lnNa yC,\Gamma
			}{
				\vlin{\napopr}{}{
					\sdash[\sS, \isna x] A \fsubst{c}{z}, B, \lnNa yC,\Gamma
				}{
					\vlhy{\sdash A \fsubst{c}{z}, B ,C \fsubst{x}{y}, \Gamma}
				}
			}
		}
		\end{array}
		$}
		$$

		\

		\item
		Case $\rrule_1/\rrule_2$ with $\rrule_1,\rrule_2\in\set{\forall, \naloadr,\naur}$ 
		with $\nabla\in\set{\lnewsymb,\lyasymb}$:

		$$
		\adjustbox{max width=\textwidth}{$\begin{array}{ccc}
			\vpz1{\lqusymb_1}x. \vA1, \vpz3{\lqusymb_2}y.\vB1, \vpz4{\Gamma}
			\pzlinks{A1/B1/12/\la/red/{pz4}}
			&\to&
			\vpz1{\lqusymb_1}x. \vA1, \vpz3{\lqusymb_2}y.\vB1, \vpz4{\Gamma}
			\pzlinks{pz1/B1/12/\link[black]{a_1}/pzgreen/{pz4}}
			\\
			\downarrow && \downarrow
			\\\\
			\vpz1{\lqusymb_1}x. \vA1,\vpz3{\lqusymb_2}y.\vB1, \vpz4{\Gamma}
			\pzlinks{pz4/pz3/12/\link[black]{a_2}/violet/{A1}}
			&\to&
			\vpz1{\lqusymb_1}x. \vA1, \vpz3{\lqusymb_2}y.\vB1, \vpz4{\Gamma}
			\pzlinks{pz1/pz3/12/\link[black]{a_3}/brown/{pz4}}
		\end{array}$}
		$$

		With $\dualizerof[\la] = \dualizerof[a_i]$ for $i\in\set{1,2,3}$.
		\nowDer

		$$
		\adjustbox{max width=\textwidth}{$
		\begin{array}{c}
			\begin{array}{ccl}
			\vlderivation{
			\vlin{\rrule_1}{}{
				\sdash \lQu[1] xA, \lQu[2] yB, \Gamma
			}{
				\vlin{\rrule_2}{}{
					\sdash A , \lQu[2] yB, \Gamma
				}{
					\vlhy{\sdash A, B, \Gamma}
				}
			}
		}
		&
		{\;\peq\;}
		&
		\vlderivation{
			\vlin{\rrule_2}{}{
				\sdash \lQu[1] xA, \lQu[2] yB, \Gamma
			}{
				\vlin{\rrule_1}{}{
					\sdash \lQu[1] xA, B, \Gamma
				}{
					\vlhy{\sdash A, B, \Gamma}
				}
			}
		}
			\end{array}
		\\\\
		\mbox{if }
		\rrule_1,\rrule_2\in\set{\forall,\nuur,\yaur}
		\\\\
		\begin{array}{ccl}
		\vlderivation{
			\vlin{\rrule_1}{}{
				\sdash \lNa xA, \lQu yB, \Gamma
			}{
				\vlin{\rrule_2}{}{
					\sdash[\sS,\isna x] A , \lQu yB, \Gamma
				}{
					\vlhy{\sdash[\sS,\isna x] A,B, \Gamma}
				}
			}
		}
		&{\;\peq\;}&
		\vlderivation{
			\vlin{\rrule_2}{}{
				\sdash \lNa xA, \lQu yB, \Gamma
			}{
				\vlin{\rrule_1}{}{
					\sdash[\sS,\isna x] \lNa xA, B, \Gamma
				}{
					\vlhy{\sdash[\sS,\isna x] A, B, \Gamma}
				}
			}
		}
		\end{array}
		\\\\
		\mbox{if }
		\rrule_1\in\set{\nuloadr,\yaloadr}
		\mbox{ and }
		\rrule_2\in\set{\forall,\nuur,\yaur}
		\\\\
		\begin{array}{ccl}
		\vlderivation{
			\vlin{\rrule_1}{}{
				\sdash \lNai1 xA, \lNai2 yB, \Gamma
			}{
				\vlin{\rrule_2}{}{
					\sdash[\sS, {\isna[1]x}] A , \lNai2 yB, \Gamma
				}{
					\vlhy{\sdash[\sS,{\isna[1]x},{\isna[2]y}] A, B, \Gamma}
				}
			}
		}
		&\;\peq\;&
		\vlderivation{
			\vlin{\rrule_2}{}{
				\sdash \lNai1 xA, \lNai2 yB, \Gamma
			}{
				\vlin{\rrule_1}{}{
					\sdash[\sS,{\isna[2]y}] \lNai1 xA, B, \Gamma
				}{
					\vlhy{\sdash[\sS,{\isna[1]x},{\isna[2]y}] A, B, \Gamma}
				}
			}
		}
		\end{array}
		\\\\
		\mbox{if }
		\rrule_1,\rrule_2\in\set{\nuloadr,\yaloadr}
		\end{array}
		$}
		$$

		\

		\item
		Case $\napopr/\rrule$ with $\rrule\in\set{\forall,\naloadr,\naur}$ and  $\nabla\in\set{\lnewsymb,\lyasymb}$:

		\

		$$
		\adjustbox{max width=\textwidth}{$\begin{array}{ccc}
			\vpz1{\lqusymb}z. \vA1, \vpz3{\nabla}\vx1.\vB1,\vpz4{\cneg\nabla}\vy1.\vC1, \vpz5{\Gamma}
			\pzlinks{A1/B1/16/\la/red/{C1,pz5}}
			\pzlinks{x1/y1/12/\lc/brown/}
			&\to&
			\vpz1{\lqusymb}z. \vA1, \vpz3{\nabla}\vx1.\vB1,\vpz4{\cneg\nabla}\vy1.\vC1, \vpz5{\Gamma}
			\pzlinks{pz1/B1/16/\labb/pzgreen/{C1,pz5}}
			\pzlinks{x1/y1/12/\lc/brown/}
			\\
			\downarrow && \downarrow
			\\\\
			\vpz1{\lqusymb}z. \vA1,\quad \vpz3{\nabla}x.\vB1,\vpz4{\cneg\nabla}y.\vC1, \vpz5{\Gamma}
			\pzlinks{A1/B1/16/\laa/violet/{pz4,pz5}}
			&\to&
			\vpz1{\lqusymb}z. \vA1,\quad \vpz3{\nabla}x.\vB1,\vpz4{\cneg\nabla}y.\vC1, \vpz5{\Gamma}
			\pzlinks{pz1/B1/16/\laabb/purple/{pz4,pz5}}
		\end{array}$}
		$$

		With $\dualizerof[\labb] = \dualizerof[\la]$,
		$\dualizerof[\laaa] = \dualizerof[\la] \fsubminus{y}$ and  $\dualizerof[\laaabb] = \dualizerof[\la]\fsubminus{y}$
		\nowDer

		$$
		\adjustbox{max width=\textwidth}{$
		\begin{array}{c}
			\begin{array}{ccl}
			\vlderivation{\vlin{\napopr}{}{
				\sdash[\sS,\isna x] \lQu zA, B,\lnNa yC, \Gamma
			}{
				\vlin{\rrule}{}{
					\sdash \lQu zA, B, C \fsubst{x}{y}, \Gamma
				}{
					\vlhy{\sdash A,B, C\fsubst{x}{y},\Gamma}
				}
			}}
		&
		{\;\peq\;}
		&
		\vlderivation{
			\vlin{\rrule}{}{
				\sdash[\sS,\isna x] \lQu zA, B,\lnNa yC, \Gamma
			}{
				\vlin{\napopr}{}{
					\sdash[\sS,\isna x] A, B,\lnNa yC, \Gamma
				}{
					\vlhy{\sdash A ,B , C\fsubst{x}{y},\Gamma}
				}
			}
		}
			\end{array}
		\\\\
		\mbox{if }
		\rrule\in\set{\forall,\nuur,\yaur} \mbox{:}
		\\\\
		\begin{array}{ccl}
			\vlderivation{
			\vlin{\nnapopr}{}{
				\sdash[\sS,\isnna x] \lQu zA, B,\lNa yC, \Gamma
			}{
				\vlin{\rrule}{}{
					\sdash \lQu zA, B, C \fsubst xy, \Gamma
				}{
					\vlhy{\sdash[\sS,\isqu z] A,B, C\fsubst xy,\Gamma}
				}
			}
		}
		&
		{\;\peq\;}
		&
		\vlderivation{
			\vlin{\rrule}{}{
				\sdash[\sS, \isnna x] \lQu zA, B,\lNa yC, \Gamma
			}{
				\vlin{\nnapopr}{}{
					\sdash[\sS,\isqu z, \isnna x] A, B,\lNa yC, \Gamma
				}{
					\vlhy{\sdash[\sS,\isqu z ] A ,B , C\fsubst xy,\Gamma}
				}
			}
		}
		\end{array}
		\\\\
		\mbox{if }
		\rrule\in\set{\nuloadr,\yaloadr} \mbox{:}
		\end{array}
		$}
		$$

		\

		\item Case $\lwith/\exists$:
		$$
		\adjustbox{max width=\textwidth}{$\begin{array}{ccccc}
			&&
			\vA1 \vlwith1 \vB1,\quad \vpz1{\exists}x.\vC1 , \vpz2{\Gamma}
			\pzlinks{A1/pz1/12/\laaa/magenta/{pz2}}
			\pzlinks{B1/C1/-12/\lb/blue/{pz2}}
			&\to&
			\vA1 \vlwith1 \vB1, \quad \vpz1{\exists}x.\vC1 , \quad \vpz2{\Gamma}
			\pzlinks{A1/pz1/12/\laaa/magenta/{pz2}}
			\pzlinks{pz1/pz2/-12/\lbbb/violet/{B1}}
			\\
			&\nearrow &&&
			\\
			\vA1 \vlwith1 \vB1, \vpz1{\exists}x.\vC1 , \vpz2{\Gamma}
			\pzlinks{A1/C1/12/\la/red/{pz2}}
			\pzlinks{B1/C1/-12/\lb/blue/{pz2}}
			&&&&\downarrow
			\\
			&\searrow &&&
			\\
			&&
			\vA1 \vlwith1 \vB1, \vpz1{\exists}x.\vC1 , \vpz2{\Gamma}
			\pzlinks{lwith1/C1/12/\lcd/pzgreen/{pz2}}
			&\to&
			\vA1 \vlwith1 \vB1,\quad \; \vpz1{\exists}x.\vC1 , \vpz2{\Gamma}
			\pzlinks{lwith1/pz1/12/\lcccd/brown/{pz2}}
		\end{array}$}
		$$

		With $\dualizerof[\laa] = \dualizerof[\la] \fsubminus{y}$,
		$\dualizerof[\lbb] = \dualizerof[\lb] \fsubminus{y}$,
		 $\dualizerof[\lc] = \dualizerof[\la] \join \dualizerof[\lb]$ and
		 $\dualizerof[\lccd] = \dualizerof[\lc]\fsubminus{y}$.
		\nowDer
		$$
		\hskip-1em
		\begin{array}{c}
			\vlderivation{
				\vlin{\exists}{}{
				\sdash A\lwith B, \lEx xC, \Gamma
			}{
				\vliin{\lwith}{}{
					\sdash A\lwith B, C\fsubst{c}{x}, \Gamma
				}{
					\vlhy{\sdash A,C\fsubst{c}{x},\Gamma}
				}{
					\vlhy{\sdash B,C\fsubst{c}{x},\Gamma}
				}
			}
		}
		\;
		\peq
		\;
		\vlderivation{
			\vliin{\lwith}{}{
				\sdash A\lwith B, \lEx xC, \Gamma
			}{
				\vlin{\exists}{}{
					\sdash A, \lEx xC, \Gamma
				}{
					\vlhy{\sdash A,C\fsubst{c}{x},\Gamma}
				}
			}{
				\vlin{\exists}{}{
					\sdash  B, \lEx xC, \Gamma
				}{
					\vlhy{\sdash B,C\fsubst{c}{x},\Gamma}
				}
			}
		}
		\end{array}
		$$

		\

		\item
		Case $\lwith/\rrule$ with $\rrule\in\set{\forall,\naloadr,\naur}$ with $\nabla\in\set{\lnewsymb,\lyasymb}$:

		$$
		\adjustbox{max width=\textwidth}{$\begin{array}{ccccc}
			&&
			\vA1 \vlwith1 \vB1, \vpz1{\lqusymb}x.\vC1 , \vpz2{\Gamma}
			\pzlinks{A1/pz1/12/\laa/magenta/{pz2}}
			\pzlinks{B1/C1/-12/\lb/blue/{pz2}}
			&\to&
			\vA1 \vlwith1 \vB1, \quad \vpz1{\lqusymb}x.\vC1 , \vpz2{\Gamma}
			\pzlinks{A1/pz1/12/\laa/magenta/{pz2}}
			\pzlinks{B1/pz1/-12/\lbb/violet/{pz2}}
			\\
			&\nearrow &&&
			\\
			\vA1 \vlwith1 \vB1, \vpz1{\lqusymb}x.\vC1 , \vpz2{\Gamma}
			\pzlinks{A1/C1/12/\la/red/{pz2}}
			\pzlinks{B1/C1/-12/\lb/blue/{pz2}}
			&&&&\downarrow
			\\
			&\searrow &&&
			\\
			&&
			\vA1 \vlwith1 \vB1, \vpz1{\lqusymb}x.\vC1 , \vpz2{\Gamma}
			\pzlinks{lwith1/C1/12/\lcd/pzgreen/{pz2}}
			&\to&
			\vA1 \vlwith1 \vB1,\quad \vpz1{\lqusymb}x.\vC1 , \vpz2{\Gamma}
			\pzlinks{lwith1/pz1/12/\lccd/brown/{pz2}}
		\end{array}$}
		$$

		With $\dualizerof[\laa] = \dualizerof[\la] $,
		$\dualizerof[\lbb] = \dualizerof[\lb] $ and
		 $\dualizerof[\lc] = \dualizerof[\la] \join \dualizerof[\lb] =\dualizerof[\lccd]$ .
		\nowDer

		$$\adjustbox{max width=\textwidth}{$
		\begin{array}{c}
			\begin{array}{ccl}
			\vlderivation{
			\vlin{\rrule}{}{
				\sdash A\lwith B, \lQu xC, \Gamma
			}{
				\vliin{\lwith}{}{
					\sdash A\lwith B, C, \Gamma
				}{
					\vlhy{\sdash A, C,\Gamma}
				}{
					\vlhy{\sdash B, C,\Gamma}
				}
			}
		}
		& 
		{\;\peq\;}
		&
		\vlderivation{
			\vliin{\lwith}{}{
				\sdash A\lwith B, \lQu xC, \Gamma
			}{
				\vlin{\rrule}{}{
					\sdash A, \lQu xC, \Gamma
				}{
					\vlhy{\sdash A,C,\Gamma}
				}
			}{
				\vlin{\rrule}{}{
					\sdash  B, \lQu xC, \Gamma
				}{
					\vlhy{\sdash B,C ,\Gamma}
				}
			}
		}
			\end{array}
		\\\\
		\mbox{if }
		\lqusymb \in\set{\forall,\nuur,\yaur}
		\\\\
		\begin{array}{ccl}
		\vlderivation{
			\vlin{\rrule}{}{
				\sdash A\lwith B, \lQu xC, \Gamma
			}{
				\vliin{\lwith}{}{
					\sdash[\sS,\isqu x] A\lwith B, C, \Gamma
				}{
					\vlhy{\sdash[\sS,\isqu x] A, C,\Gamma}
				}{
					\vlhy{\sdash[\sS,\isqu x] B, C,\Gamma}
				}
			}
		}
		&
		{\;\peq\;}
		&
		\vlderivation{
			\vliin{\lwith}{}{
				\sdash A\lwith B, \lQu xC, \Gamma
			}{
				\vlin{\rrule}{}{
					\sdash A, \lQu xC, \Gamma
				}{
					\vlhy{\sdash[\sS,\isqu x] A,C ,\Gamma}
				}
			}{
				\vlin{\rrule}{}{
					\sdash  B, \lQu xC, \Gamma
				}{
					\vlhy{\sdash[\sS,\isqu x] B, C,\Gamma}
				}
			}
		}
		\end{array}
		\\\\
		\mbox{if }
		\rrule \in\set{\nuloadr,\yaloadr}
		\end{array}	
		$}$$

		\

		\item
		Case $\rrule / \conf$ with $\rrule \in \set{\exists,\forall, \nuloadr,\yaloadr, \nuur,\yaur}$:

		$$
		\adjustbox{max width=\textwidth}{$\begin{array}{ccc}
			\conf (a_1 , \dots, a_l , b_1, \dots, b_k , c) & \to & \conf (\conf(a_1 , \dots, a_l ),\conf( b_1, \dots, b_k , c))
			\\
			\downarrow && \downarrow
			\\
			\conf (a_1 , \dots, a_l , b_1, \dots, b_k , c') & \to & \conf (\conf(a_1 , \dots, a_l ),\conf( b_1, \dots, b_k , c'))
		\end{array}$}
		$$

		Where either
		\begin{itemize}
			\item $\rrule\in\set{\forall,\nuur,\yaur,\nuloadr,\yaloadr}$, thus
			$c = \Gamma , A$, $c' = \Gamma, \lQu xA$, and $\dualizerof[c'] = \dualizerof[c]$; or
			\item $\rrule\in\set{\exists,\nupopr,\yapopr}$, thus
			$c = \Gamma , A$ ,  $c' = \Gamma, \lEx x A$ and $\dualizerof[c'] = \dualizerof[c] \fsubminus{x}$.
		\end{itemize}
		The derivation labelling the link in the bottom-right corner of the diagram is the same, independently of the sequence of coalescence steps.

	\end{itemize}

	Finally, we have the following special additional cases, in which the dualizer is not the same, but the proof nets are.

	\begin{itemize}
		\item
		Case
		$\nupopr/\yapopr$:

		\

		$$
		\adjustbox{max width=\textwidth}{$\begin{array}{ccc}
			\vpz1{\lnewsymb}\vx1. \vA1, \vpz3{\lyasymb}\vy2.\vB1, \vpz4{\Gamma}
			\pzlinks{A1/B1/12/\la/red/{pz4}}
			\pzlinks{x1/y2/16/\lb/blue/}
			&\to&
			\vpz1{\lnewsymb}\vx1. \vA1, \vpz3{\lyasymb}\vy2.\vB1, \vpz4{\Gamma}
			\pzlinks{pz1/B1/12/c/pzgreen/{pz4}}
			\\
			\downarrow && \downarrow
			\\\\
			\vpz1{\lnewsymb}\vx1. \vA1, \vpz3{\lyasymb}\vy2.\vB1, \vpz4{\Gamma}
			\pzlinks{A1/pz3/12/d/violet/{pz4}}
			&\to&
			\vpz1{\lnewsymb}\vx1. \vA1, \vpz3{\lyasymb}\vy2.\vB1, \vpz4{\Gamma}
			\pzlinks{pz1/pz3/12/e/brown/{pz4}}
		\end{array}$}
		$$
		where
		$\dualizerof[c]=\dualizerof[d]=\dualizerof[e]=\dualizerof[a]\setminus\set{y}$.
		\nowDer

		$$
		\begin{array}{lcr}
			\dD_1
		\quad=\quad
			\vlderivation{
					\vlin{\nuloadr}{}{
						\sdash \Gamma, \lNu xA, \lYa yB
					}{
						\vlin{\nupopr}{}{
							\sdash[\sS,\isnu x] \Gamma, A, \lYa yB
						}{\vlhy{\sdash \Gamma, A, B\fsubst xy}}
					}
				}
			&
		\quad\pweq\quad
			&
				\vlderivation{
					\vlin{\yaloadr}{}{
						\sdash \Gamma, \lNu xA, \lYa yB
					}{
						\vlin{\yapopr}{}{
							\sdash[\sS,\isya y] \Gamma, \lNu xA, B
						}{\vlhy{\sdash \Gamma, A\fsubst yx, B}}
					}
				}
		\quad=\quad
			\dD_2
		\end{array}
		$$

		Where $\dD_1\pweq \dD_2$ because

		$$
		\begin{array}{lcr}
			\dD_1
		\quad=\quad
			\vlderivation{
					\vlin{\nuloadr}{}{
						\sdash \Gamma, \lNu xA, \lYa yB
					}{
						\vlin{\nupopr}{}{
							\sdash[\sS,\isnu z] \Gamma, A\fsubst zx, \lYa yB
						}{\vlhy{\sdash \Gamma, A\fsubst zx, B\fsubst zy}}
					}
				}
			&
			\quad\pweq\quad
			&
			\vlderivation{
				\vlin{\yaloadr}{}{
					\sdash \Gamma, \lNu xA, \lYa xB
				}{
					\vlin{\yapopr}{}{
						\sdash[\sS,\isya z] \Gamma, \lNu xA, B\fsubst zx
					}{\vlhy{\sdash \Gamma, A\fsubst zx , B\fsubst zx}}
				}
			}
		\quad=\quad
			\dD_2
		\end{array}
		$$

	\end{itemize}
\end{proof}
%%%%%%%%%%%%%%%%%%%%%%%%%%%%%%%%%%%%%%%%%%%%%%%%%%%%%%%%%%%%%%%%

%%%%%%%%%%%%%%%%%%%%%%%%%%%%%%%%%%%%%%%%%%%%%%%%%%%%%%%%%%%%%%%%
\begin{remark}\label{remark:notcritical}
	As already observed in \cite{hei:hug:conflict}, coalescence is not confluent on non-coalescent \cotree, as shown in the following examples.
	$$\adjustbox{max width=\textwidth}{$\begin{array}{c}
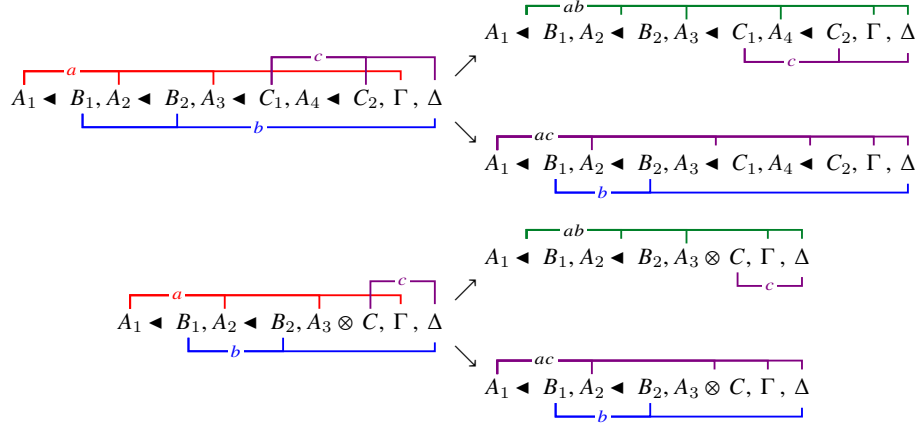
		\begin{array}{ccc}
			&&
			\viA1 \vpz1{\lprec} \viB1,\viA2 \vpz2{\lprec} \viB2, \viA3 \vpz3{\lprec} \viC1, \viA4 \vpz4{\lprec} \viC2, \vpz5{\Gamma},\vpz6{\Delta}
			\pzlinks{pz1/pz2/12/\labb/pzgreen/{A3,A4,pz5,pz6}}
			\pzlinks{C1/C2/-12/\lc/violet/{pz6}}
			\\
			&\nearrow &
			\\
			\viA1 \vpz1{\lprec} \viB1,\viA2 \vpz2{\lprec} \viB2, \viA3 \vpz3{\lprec} \viC1, \viA4 \vpz4{\lprec} \viC2, \vpz5{\Gamma},\vpz6{\Delta}
			\pzlinks{A1/A2/12/\la/red/{A3,pz5}}
			\pzlinks{B1/pz6/-12/\lb/blue/{B2}}
			\pzlinks{C1/C2/18/\lc/violet/{pz6}}
			&&
			\\
			&\searrow &
			\\
			&&
			\viA1 \vpz1{\lprec} \viB1,\viA2 \vpz2{\lprec} \viB2, \viA3 \vpz3{\lprec} \viC1, \viA4 \vpz4{\lprec} \viC2, \vpz5{\Gamma},\vpz6{\Delta}
			\pzlinks{A1/A2/12/\lac/violet/{pz3,pz4,pz5,pz6}}
			\pzlinks{B1/B2/-12/\lb/blue/{pz6}}
		\end{array}
	\\\\\\
		\begin{array}{ccc}
			&&
			\viA1 \vpz1{\lprec} \viB1,\viA2 \vpz2{\lprec} \viB2, \viA3 \vpz3{\ltens} \vC1, \vpz4{\Gamma}, \vpz5{\Delta}
			\pzlinks{pz1/pz2/12/\labb/pzgreen/{A3,pz4,pz5}}
			\pzlinks{C1/pz5/-12/\lc/violet/}
			\\
			&\nearrow &
			\\
			\viA1 \vpz1{\lprec} \viB1,\viA2 \vpz2{\lprec} \viB2, \viA3 \vpz3{\ltens} \vC1, \vpz4{\Gamma}, \vpz5{\Delta}
			\pzlinks{A1/A2/12/\la/red/{A3,pz4}}
			\pzlinks{B1/B2/-12/\lb/blue/{pz5}}
			\pzlinks{C1/pz5/18/\lc/violet/}
			&&
			\\
			&\searrow &
			\\
			&&
			\viA1 \vpz1{\lprec} \viB1,\viA2 \vpz2{\lprec} \viB2, \viA3 \vpz3{\ltens} \vC1, \vpz4{\Gamma}, \vpz5{\Delta}
			\pzlinks{A1/A2/12/\lac/violet/{pz3,pz4,pz5}}
			\pzlinks{B1/B2/-12/\lb/blue/{pz5}}
		\end{array}
	\end{array}$}$$
\end{remark}
%%%%%%%%%%%%%%%%%%%%%%%%%%%%%%%%%%%%%%%%%%%%%%%%%%%%%%%%%%%%%%%%

%%%%%%%%%%%%%%%%%%%%%%%%%%%%%%%%%%%%%%%%%%%%%%%%%%%%%%%%%%%%%%%%%%%
%%%%%%%%%%%%%%%%%%%%%%%%%%%%%%%%%%%%%%%%%%%%%%%%%%%%%%%%%%%%%%%%%%%
\section{To Witness or Not To Witness?}\label{sec:witnesses}
%%%%%%%%%%%%%%%%%%%%%%%%%%%%%%%%%%%%%%%%%%%%%%%%%%%%%%%%%%%%%%%%%%%
%%%%%%%%%%%%%%%%%%%%%%%%%%%%%%%%%%%%%%%%%%%%%%%%%%%%%%%%%%%%%%%%%%%

The choice of witnesses of the nominal quantifier rules $\napopr$ and $\naloadr$ should, in principle, be irrelevant for the proof, as the role of the nominal quantifiers is only to enforce freshness conditions.
So one may expect to be willing to identify derivations differing for the choice of the witness of nominal quantifier rules, as in the case of the $\alpha$-equivalence of bound variables.
That is, we may want to identify the two following derivations, which differ only for the choice of the witness of the nominal quantifiers.
\begin{equation}\label{eq:nomeq}
	\vlderivation{
		\vliq{\nuloadr + \nupopr}{}{
			\vdash \lNu x P(x,z), \lYa y \cneg{P(y,z)}
		}{
			\vlin{\axrule}{}{\vdash P(x,z), \cneg{P(x,z)}}{\vlhy{}}
		}
	}
	\quad{\sim}\quad
	\vlderivation{
		\vliq{\yaloadr + \yapopr}{}{
			\vdash \lNu x P(x,z), \lYa y \cneg{P(y,z)}
		}{
			\vlin{\axrule}{}{\vdash P(y,z), \cneg{P(y,z)}}{\vlhy{}}
		}
	}
\end{equation}

Due to the similarity of the derivations above with the ones in \Cref{eq:intro2}, and in view of the fact that the syntax of $\PIL$ does not include function symbols, therefore witnesses are always variable, we can argue that also the choice of the witness of existential quantifier is negligible%
\footnote{
	More precisely, given an axiomatic linking for a sequent, it is possible to compute (if it exists) a witness map such in polynomial time using the unification algorithm for first-order terms.
}%
, 
as it does not change the topology of the proof net associated to the derivation. 
Formally, we define the two following additional equivalence relations over derivations.

\begin{definition}
	Let $\dD_1$ and $\dD_2$ be two derivations in $\PIL$.
	We say that $\dD_1$ and $\dD_2$ are \defn{equivalent modulo fresh name} if it is possible to transform $\dD_1$ into $\dD_2$ by changing the active term of the nominal quantifier rules and propagating the changes upwards in the derivation.
	We say that $\dD_1$ and $\dD_2$ are \defn{equivalent modulo witness assignment} if it is possible to transform $\dD_1$ into $\dD_2$ by changing the active term of the existential and nominal quantifier rules and propagating the changes upwards in the derivation.
\end{definition}

We can prove that we can extend our local and strong canonicity results with the equivalence relations defined above by modifying the definition of our proof nets, i.e., both \conets and slice nets, by considering new definitions of \emph{nominal} and \emph{unification} proof nets, the latter in the spirit of Hughes' unification nets \cite{hug:unification} for multiplicative linear logic, the unification nets for purely additive linear logic from \cite{hei:hug:str:ALL1}, and combinatorial proofs for first-order classical logic \cite{hug:str:wu:CP1}.

\newcommand{\mgdual}[1][]{\mu_{#1}^\linking}
\begin{definition}
	The \defn{initial witness map} $\mgdual$ of a linking $\linking$ is defined by letting $\mgdual[\la]$ be defined as follows:
	\begin{itemize}
			\item if $\la= \set{\lsend xy,\lrecv zt}$, then $\mgdual[\la]=\fsubsts{x/z,y/t}$;
			\item if $\la=\set{x,y}$ is a nominal link with $x$ bound by $\lnewsymb$, then $\mgdual[\la]=\fsubsts xy$;
			\item if $\la=\set{\lunit}$, then $\mgdual[\la]=\dualizerof[\emptyset]$.
	\end{itemize}

	We write $\dualizerof^\linking\geq \mgdual$ if for every link $\la$ in $\linking$, we have $\dualizerof[\la]^\linking\geq\mgdual[\la]$, and we may refer to proof nets as either conflict or slice nets.

	A \defn{nominal proof net} is a pair $\tuple{\linktree,\dualizerof^\linking}$ 
	with $\mgdual[\la]\geq \dualizerof[\la]^\linking$ for every link $\la$ in $\linking$ and such that the domain of each dualizer $\dualizerof[\la]^\linking$ is a variable bound by an existential quantifier in the sequent $\Gamma$,
	such that
	$\tuple{\linktree,\mgdual}$ is a proof net.
	A \defn{unification proof net} is a \cotree $\linktree$ such that 
	$\tuple{\linktree,\mgdual}$ is a proof net.

	Note however that for unification nets we weaken the coherence condition for coalescence, asking that the dualizers only are unifiable, and for nominal nets we restrict the condition of coherence for variables bound by the existential quantifiers only.
\end{definition}

We conjecture that for equivalence modulo fresh name assignment, we should consider the following additional equivalence:
$$
	\vlderivation{
		\vlin{\nuloadr}{}{
			\sdash \Gamma, \rclr{\lNu xA}, \lYa xB
		}{
			\vlin{\nupopr}{}{
				\sdash[\sS,\isnu x] \Gamma, \rclr A, \bclr{\lYa xB}
			}{\vlhy{\sdash \Gamma, A, \bclr B}}
		}
	}
\qquad\peq\qquad
	\vlderivation{
		\vlin{\yaloadr}{}{
			\sdash \Gamma, \lNu xA, \bclr{\lYa xB}
		}{
			\vlin{\yapopr}{}{
				\sdash[\sS,\isya x] \Gamma, \rclr{\lNu xA}, \bclr B
			}{\vlhy{\sdash \Gamma, \rclr A, B}}
		}
	}
$$
as well as the transformation on a derivation obtained by changing the witness of a nominal quantifier rule and propagating the change upwards in the derivation, as in the following example.
Similarly, for equivalence modulo witness assignment, we should consider the transformation on a derivation obtained by changing the witness of an existential quantifier rule and propagating the change upwards in the derivation.

%%%%%%%%%%%%%%%%%%%%%%%%%%%%%%%%%%%%%%%%%%%%%%%%%%%%%%%%%%%%%%%%%%%
%%%%%%%%%%%%%%%%%%%%%%%%%%%%%%%%%%%%%%%%%%%%%%%%%%%%%%%%%%%%%%%%%%%
%%%%%%%%%%%%%%%%%%%%%%%%%%%%%%%%%%%%%%%%%%%%%%%%%%%%%%%%%%%%%%%%%%%
%%%%%%%%%%%%%%%%%%%%%%%%%%%%%%%%%%%%%%%%%%%%%%%%%%%%%%%%%%%%%%%%%%%
\end{document}